%% file: ms.tex
\documentclass[acmsmall,screen]{acmart}

\newif\ifhideproofs
\hideproofstrue

\acmPrice{}
\acmDOI{10.1145/3485475}
\acmYear{2021}
\acmSubmissionID{oopsla21main-p9-p}
\acmJournal{PACMPL}
\acmVolume{5}
\acmNumber{OOPSLA}
\acmArticle{98}
\acmMonth{10}

\startPage{1}

\setcopyright{none}
\copyrightyear{2021}

\usepackage{booktabs}   %
\usepackage{subcaption} %

\input{defs}

\begin{document}

\title{Making Weak Memory Models Fair}

\author{Ori Lahav}
\affiliation{
  \institution{Tel Aviv University}            %
  \country{Israel}                    %
}
\email{orilahav@tau.ac.il}          %

\author{Egor Namakonov}
\affiliation{
  \institution{St. Petersburg University}           %
  \country{Russia}                   %
}
\affiliation{
  \institution{JetBrains Research}           %
  \country{Russia}                   %
}
\email{egor.namakonov@jetbrains.com}         %

\author{Jonas Oberhauser}
\affiliation{
  \institution{Huawei Dresden Research Center} %
  \country{Germany}                    %
}
\affiliation{
  \institution{Huawei OS Kernel Lab}            %
  \country{Germany}                    %
}
\email{jonas.oberhauser@huawei.com}          %

\author{Anton Podkopaev}
\affiliation{
  \institution{HSE University}            %
  \country{Russia}                    %
}
\affiliation{
  \institution{JetBrains Research}            %
  \country{Russia}                    %
}
\email{apodkopaev@hse.ru}

\author{Viktor Vafeiadis}
\affiliation{
  \institution{MPI-SWS}            %
  \country{Germany} %
}
\email{viktor@mpi-sws.org}

\begin{abstract}
\input{abstract.tex}
\end{abstract}



\maketitle

\begin{toappendix}
\section{Proofs of Propositions and Lemmas in the Paper}
\end{toappendix}

\input{intro.tex}
\input{op_sem.tex}
\input{dec_sem.tex}

\input{fair_dec_sem.tex}

\input{rc11.tex}
\input{robustness.tex}

\input{lockexamples.tex}

\input{discussion.tex}

\begin{acks}                            %
We thank the anonymous reviewers for their helpful feedback.
This research was supported in part by the European Research Council (ERC)
under the European Union's Horizon 2020 research and innovation programme
(grant agreement no. 851811 and 101003349).
Lahav was also supported by the Israel Science Foundation (grant number 1566/18)
and by the Alon Young Faculty Fellowship.
\end{acks}

\bibliography{ms}

\newpage
\appendix
\onecolumn

\printproofs

\input{equiv_proof}

\input{termination-checking}

\end{document}

%% file: defs.tex
\usepackage{booktabs}
\usepackage{siunitx}
\usepackage{microtype}
\usepackage{amsmath}
\usepackage{stmaryrd}
\usepackage{thm-restate}
\usepackage{multicol}
\usepackage{color}
\usepackage{algorithm}
\usepackage{environ}
\usepackage{float}
\usepackage{multirow}
\usepackage[noend]{algpseudocode}
\usepackage{graphicx}
\usepackage{tikz}
\usepackage{thm-restate}
\usepackage{mathpartir}
\usepackage{xspace}
\usepackage{paralist}
\usepackage[inline]{enumitem}
\usepackage[capitalize]{cleveref}
\usetikzlibrary{fadings,decorations.pathmorphing,decorations.pathreplacing}
\usetikzlibrary{calc}
\usetikzlibrary{fit}
\usetikzlibrary{arrows,automata}
\usetikzlibrary{positioning}
\usepackage{ifthen}
\usepackage{xargs}
\usepackage[makeroom]{cancel}
\usepackage{fontawesome}

\theoremstyle{acmdefinition}
\newtheorem{remark}{Remark}

\AtEndPreamble{%
\theoremstyle{acmplain}
\newtheorem{notation}[theorem]{Notation}
\newtheorem*{notation*}{Notation}
}

\ifhideproofs

\toks1000={}
\newcounter{proofcount}

\NewEnviron{toappendix}{%
  \global\toks\numexpr\the\toks1000+\value{proofcount}\relax=\expandafter{\BODY}
  \stepcounter{proofcount}}

\NewEnviron{proofof}[1]{%
  \edef\next{%
    \noexpand\begin{proof}[Proof of \noexpand\cref{\unexpanded\expandafter{#1}}]%
    \unexpanded\expandafter{\BODY}}%
  \global\toks\numexpr\the\toks1000+\value{proofcount}\relax=\expandafter{\next\end{proof}}
  \stepcounter{proofcount}}

\makeatletter
\def\printproofs{%
  \count@=\z@
  \loop
    \the\toks\numexpr\the\toks1000+\count@\relax
    \ifnum\count@<\value{proofcount}%
    \advance\count@\@ne
  \repeat}
\makeatother

\else %

\NewEnviron{toappendix}{}
\NewEnviron{proofof}[1]{\begin{proof}\BODY\end{proof}}
\newcommand\printproofs{}

\fi %

\crefformat{section}{#2\S{}#1#3}
\Crefname{section}{Section}{Sections}
\Crefformat{section}{Section #2#1#3}
\Crefname{figure}{\text{Figure}}{\text{Figures}}
\crefname{corollary}{\text{Corollary}}{\text{corollaries}}
\Crefname{corollary}{\text{Corollary}}{\text{Corollaries}}
\crefname{lemma}{\text{Lemma}}{\text{Lemmas}}
\Crefname{lemma}{\text{Lemma}}{\text{Lemmas}}
\crefname{proposition}{\text{Prop.}}{\text{Propositions}}
\Crefname{proposition}{\text{Proposition}}{\text{Propositions}}
\crefname{definition}{\text{Def.}}{\text{Definitions}}
\Crefname{definition}{\text{Definition}}{\text{Definitions}}
\crefname{notation}{\text{Notation}}{\text{Notations}}
\Crefname{notation}{\text{Notation}}{\text{Notations}}
\crefname{theorem}{\text{Theorem}}{\text{Theorems}}
\Crefname{theorem}{\text{Theorem}}{\text{Theorems}}
\crefname{figure}{\text{Fig.}}{\text{Figures}}
\Crefname{figure}{\text{Figure}}{\text{Figures}}



\AtEndPreamble{%
\newcounter{claimcounter}
\numberwithin{claimcounter}{theorem}

\crefname{claimcounter}{\text{Claim}}{\text{Claims}}
\Crefname{claimcounter}{\text{Claim}}{\text{Claims}}
}

\let\originalparagraph\paragraph
\renewcommand{\paragraph}[1]{\originalparagraph{\textbf{#1}}}

\newcommand{\textcode}[1]{\texorpdfstring{\texttt{#1}}{#1}}
\newcommand{\kw}[1]{\textbf{\textcode{#1}}}
\newcommand{\mln}[1]{\textit{#1}}

\newcommand{\iteml}[3]{
  \kw{if} \; #1\\
  \begin{array}[t]{@{}l@{}l}
    \kw{then}& \begin{array}[t]{l} #2 \end{array} \\
    \kw{else}& \begin{array}[t]{l} #3 \end{array} \\
  \end{array}
}

\newcommand{\itneml}[2]{\kw{if}\;#1\\
\begin{array}[t]{@{}l@{}l}
\kw{then}& \begin{array}[t]{l} #2 \end{array} \\
\end{array}}

\newcommand{\ie}{\emph{i.e.,} }
\newcommand{\eg}{\emph{e.g.,} }

\newcommand{\wrt}{w.r.t.~}
\newcommand{\aka}{a.k.a.~}

\newcommand{\inarrT}[1]{\begin{array}[t]{@{}l@{}}#1\end{array}}

\newcommand{\inarr}[1]{\begin{array}{@{}l@{}}#1\end{array}}
\newcommand{\inarrII}[2]{\begin{array}{@{}l@{~~}||@{~~}l@{}}\inarr{#1}&\inarr{#2}\end{array}}

\newcommand{\inarrIV}[4]{\begin{array}{@{}l@{~~}||@{~~}l@{~~}||@{~~}l@{~~}||@{~~}l@{}}\inarr{#1}&\inarr{#2}&\inarr{#3}&\inarr{#4}\end{array}}

\newcommand{\inarrTwo}[2]{\begin{array}{@{}l@{~~}m{2em}@{~~}l@{}}\inarr{#1}&\qquad&\inarr{#2}\end{array}}

\renewcommand{\comment}[1]{\color{teal}{~~\texttt{/\!\!/}\textit{#1}}}

\newcommand{\set}[1]{\{{#1}\}}

\newcommand{\st}{\; | \;}
\newcommand{\N}{{\mathbb{N}}}
\newcommand{\dom}[1]{\textit{dom}{({#1})}}

\newcommand{\codom}[1]{\textit{codom}{({#1})}}

\newcommand{\tup}[1]{{\langle{#1}\rangle}}
\newcommand{\nin}{\not\in}
\newcommand{\suq}{\subseteq}

\newcommand{\maketil}[1]{{#1}\ldots{#1}}
\newcommand{\til}{\maketil{,}}

\newcommand{\cdottil}{\maketil{\cdot}}

\newcommand{\rst}[1]{|_{#1}}

\newcommand{\defiff}{\mathrel{\stackrel{\mathsf{\triangle}}{\Leftrightarrow}}}
\newcommand{\defeq}{\triangleq}
\newcommand{\powerset}[1]{\mathcal{P}({#1})}

\makeatletter
\newcommand{\raisemath}[1]{\mathpalette{\raisem@th{#1}}}
\newcommand{\raisem@th}[3]{\raisebox{#1}{$#2#3$}}
\makeatother

\newcommandx{\yaHelper}[2][1=\empty]{%
\ifthenelse{\equal{#1}{\empty}}%
  { \ensuremath{ \scriptstyle{ #2 } } } %
  { \raisebox{ #1 }[0pt][0pt]{ \ensuremath{ \scriptstyle{ #2 } } } }  %
}

\newcommandx{\yrightarrow}[4][1=\empty, 2=\empty, 4=\empty, usedefault=@]{%
  \ifthenelse{\equal{#2}{\empty}}
  { \xrightarrow{ \protect{ \yaHelper[ #4 ]{ #3 } } } } %
  { \xrightarrow[ \protect{ \yaHelper[ #2 ]{ #1 } } ]{ \protect{ \yaHelper[ #4 ]{ #3 } } } } %
}

\newcommand{\astep}[1]{
\ifthenelse{\equal{#1}{\empty}}{\xrightarrow{}}
{\mathrel{\raisebox{-0.8pt}{\ensuremath{\xrightarrow{#1}}}}}}

\colorlet{colorPO}{gray!60!black}
\colorlet{colorRF}{green!60!black}
\colorlet{colorMO}{orange}
\colorlet{colorFR}{purple}
\colorlet{colorECO}{red!80!black}
\colorlet{colorSYN}{green!40!black}
\colorlet{colorHB}{blue}
\colorlet{colorPPO}{magenta}
\colorlet{colorPB}{olive}
\colorlet{colorSBRF}{olive}
\colorlet{colorRMW}{olive!70!black}
\colorlet{colorRSEQ}{blue}
\colorlet{colorSC}{violet}
\colorlet{colorPSC}{violet}
\colorlet{colorREL}{olive}
\colorlet{colorCONFLICT}{olive}
\colorlet{colorRACE}{olive}
\colorlet{colorWB}{orange!70!black}
\colorlet{colorPSC}{violet}
\colorlet{colorSCB}{violet}
\colorlet{colorDEPS}{violet}
\colorlet{colorBF}{green!40!black}

\tikzset{
   every path/.style={>=stealth},
   po/.style={->,color=colorPO,thin,shorten >=-0.5mm,shorten <=-0.5mm},
   sw/.style={->,color=colorSYN,shorten >=-0.5mm,shorten <=-0.5mm},
   rf/.style={->,color=colorRF,dashed,,shorten >=-0.5mm,shorten <=-0.5mm},
   hb/.style={->,color=colorHB,thick,shorten >=-0.5mm,shorten <=-0.5mm},
   mo/.style={->,color=colorMO,dotted,very thick,shorten >=-0.5mm,shorten <=-0.5mm},
   no/.style={->,dotted,thick,shorten >=-0.5mm,shorten <=-0.5mm},
   fr/.style={->,color=colorFR,dotted,thick,shorten >=-0.5mm,shorten <=-0.5mm},
   deps/.style={->,color=colorDEPS,dotted,thick,shorten >=-0.5mm,shorten <=-0.5mm},
   rmw/.style={->,color=colorRMW,thick,shorten >=-0.5mm,shorten <=-0.5mm},
}

\newcommand{\evlab}[4]{{#1}^{#2}({#3},{#4})}

\newcommand{\evulab}[4]{{\lU}^{#1}({#2},{#3},{#4})}
\newcommand{\rlab}[3]{{\lR}^{#1}({#2},{#3})}

\newcommand{\wlab}[3]{{\lW}^{#1}({#2},{#3})}

\newcommand{\ulab}[4]{{\lU}^{#1}({#2},{#3},{#4})}

\newcommand{\proplab}{{\mathtt{prop}}}

\newcommand{\lE}{{\mathtt{E}}}

\newcommand{\lR}{{\mathtt{R}}}
\newcommand{\lW}{{\mathtt{W}}}

\newcommand{\lQ}{{\mathtt{Q}}}

\newcommand{\lU}{{\mathtt{RMW}}}
\newcommand{\lF}{{\mathtt{F}}}

\newcommand{\linit}{{\mathtt{init}}}
\newcommand{\lSigma}{{\mathbf{\Sigma}}}
\newcommand{\lTheta}{{\mathbf{\Theta}}}

\newcommand{\lTLAB}{{\mathtt{tlab}}}
\newcommand{\lLAB}{{\mathtt{elab}}}
\newcommand{\lTID}{{\mathtt{tid}}}
\newcommand{\lSN}{{\mathtt{sn}}}
\newcommand{\lTYP}{{\mathtt{typ}}}
\newcommand{\lLOC}{{\mathtt{loc}}}

\newcommand{\lVALR}{{\mathtt{val_r}}}
\newcommand{\lVALW}{{\mathtt{val_w}}}

\newcommand{\lSRC}{{\mathtt{src}}}
\newcommand{\lTGT}{{\mathtt{tgt}}}

\newcommand{\lTS}{{\mathtt{ts}}}
\newcommand{\lVIEW}{{\mathtt{view}}}
\newcommand{\lVAL}{{\mathtt{val}}}

\newcommand{\rf}{{\color{colorRF}\mathit{rf}}}
\newcommand{\mo}{{\color{colorMO}\mathit{mo}}}

\newcommand{\lX}{\mathtt{X}}
\newcommand{\lPO}{{\color{colorPO}\mathtt{po}}}
\newcommand{\lRF}{{\color{colorRF} \mathtt{rf}}}

\newcommand{\lMO}{{\color{colorMO} \mathtt{mo}}}

\newcommand{\lFR}{{\color{colorFR} \mathtt{fr}}}

\newcommand{\lHB}{{\color{colorHB}\mathtt{hb}}}
\newcommand{\lHBSC}{{\color{colorHB}\mathtt{hb}}_{\SC}}
\newcommand{\lHBTSO}{{\color{colorHB}\mathtt{hb}}_{\TSO}}
\newcommand{\lHBRA}{{\color{colorHB}\mathtt{hb}}_{\RA}}

\newcommand{\lPPO}{{\color{colorPPO}\mathtt{ppo}}}

\newcommand{\lSC}{{\mathtt{sc}}}
\newcommand{\lRA}{{\mathtt{ra}}}

\newcommand{\lmakeE}[1]{#1\mathtt{e}}
\newcommand{\lRFE}{\lmakeE{\lRF}}

\newcommand{\lFRE}{\lmakeE{\lFR}}
\newcommand{\lMOE}{\lmakeE{\lMO}}

\newcommand{\dpo}[2]{\draw[po] (#1) edge (#2);}
\newcommand{\dfr}[3][below]{\draw[fr] (#2) edge node[#1] {\smaller$\lFR$} (#3);}
\newcommandx{\drf}[4][1=, 2=below]{\draw[rf, #1] (#3) edge node[#2] {\smaller$\lRF$} (#4);}
\newcommandx{\dmo}[4][1=, 2=right]{\draw[mo, #1] (#3) edge node[#2] {\smaller$\lMO$} (#4);}

\newcommand{\Event}{\mathsf{Event}}
\newcommand{\Init}{\mathsf{Init}}
\newcommand{\Tid}{\mathsf{Tid}}
\newcommand{\Loc}{\mathsf{Loc}}
\newcommand{\Val}{\mathsf{Val}}
\newcommand{\Lab}{\mathsf{ELab}}

\newcommand{\Exec}{\mathsf{EGraph}}

\newcommand{\readInst }[2]{#1 \;{:=}\;#2}

\newcommand{\ifGotoInst}[2]{\kw{if} \; #1 \; \kw{goto} \; #2}
\newcommand{\gotoInst}[1]{\kw{goto} \; #1}
\newcommand{\writeInst}[2]{#1\;{:=}\;#2}
\newcommand{\assignInst}[2]{#1\;{:=}\;#2}
\newcommand{\incInstP}[2]{\faddInstn({#1},{#2})}
\newcommand{\incInst}[3]{#1 \;{:=}\;\faddInstn({#2},{#3})}

\newcommand{\casInstn}{\kw{CAS}}
\newcommand{\faddInstn}{\kw{FADD}}
\newcommand{\funcDecl}[4]{
    #1 \; #2(#3) \; \{ \;
    \begin{array}[t]{@{}l@{}}
    #4
    \; \}
    \end{array}  \\
}
\newcommand{\funcDeclml}[4]{
    #1 \; #2(#3) \; \{ \\
    \quad \begin{array}[t]{@{}l@{}}
    #4
    \end{array}  \\
    \} \\
}
\newcommand{\funcDeclsl}[4]{
    #1 \; #2(#3) \; \{ \; #4 \; \}
}

\newcommand{\repeatInst}[2]{\kw{repeat}\; \{\; #1 \;\}\;\kw{until}\;#2}
\newcommand{\while}[2]{\kw{while}\;#1\;\{\; #2 \; \}}

\newcommand{\repeatml}[2]{
\kw{repeat} \; \{ \;
\begin{array}[t]{@{}l@{}}
#1 \; \}
\end{array} \\
\quad \kw{until} \, #2
}

\newcommand{\RC}{\ensuremath{\mathsf{RC11}}\xspace}

\newcommand{\RA}{\ensuremath{\mathsf{RA}}\xspace}
\newcommand{\SCOH}{\ensuremath{\mathsf{StrongCOH}}\xspace}
\newcommand{\SC}{\ensuremath{\mathsf{SC}}\xspace}
\newcommand{\TSO}{\ensuremath{\mathsf{TSO}}\xspace}

\newcounter{mylabelcounter}

\makeatletter
\newcommand{\labelAxiom}[2]{%
\hfill{\normalfont\textsc{(#1)}}\refstepcounter{mylabelcounter}
\immediate\write\@auxout{%
  \string\newlabel{#2}{{\unexpanded{\normalfont\textsc{#1}}}{\thepage}{{\unexpanded{\normalfont\textsc{#1}}}}{mylabelcounter.\number\value{mylabelcounter}}{}}
}%
}
\makeatother

\newcommand{\squishlist}[1][$\bullet$]{%
 \begin{list}{#1}
  { \setlength{\itemsep}{0pt}
     \setlength{\parsep}{0pt}
     \setlength{\topsep}{1pt}
     \setlength{\partopsep}{0pt}
     \setlength{\leftmargin}{1.2em}
     \setlength{\labelwidth}{0.5em}
     \setlength{\labelsep}{0.4em} } }
\newcommand{\squishend}{
  \end{list}  }

\newcommand{\m}{{M}}
\renewcommand{\b}{{B}}
\newcommand{\buff}{{b}}

\newcommand{\A}{{A}}
\newcommand{\init}{\mathit{init}}
\newcommand{\M}{{\mathcal{M}}}

\newcommand{\prog}{{P}}
\newcommand{\loc}{{x}}
\newcommand{\loca}{{y}}

\newcommand{\tid}{{\tau}}
\newcommand{\tida}{{\pi}}
\newcommand{\lab}{{l}}
\newcommand{\val}{v}
\newcommand{\valr}{\val_\lR}
\newcommand{\valw}{\val_\lW}

\newcommand{\beh}{\beta}

\newcommand{\Beh}[1]{\mathsf{B}({#1})}
\newcommand{\Behtf}[1]{\mathsf{B^{tf}}({#1})}
\newcommand{\Behmf}[1]{\mathsf{B^{mf}}({#1})}
\newcommand{\Behfin}[1]{\mathsf{B^{fin}}({#1})}

\newcommand{\OTraces}[1]{\mathsf{OTr}({#1})}
\newcommand{\OTracesfin}[1]{\mathsf{OTr^{fin}}({#1})}

\newcommand{\en}{\nu}
\providecommand{\G}{}
\renewcommand{\G}{\mathcal{G}} %

\newcommand{\fair}{\texttt{fair}}

\newcommand{\Time}{\mathsf{Time}}
\newcommand{\View}{\mathsf{View}}
\newcommand{\Msg}{\mathsf{Msg}}
\newcommand{\msga}{m}

\newcommand{\ts}{t}
\newcommand{\msg}[4]{\tup{#1:#2@#3,#4}}
\newcommand{\view}{\mathit{V}}
\newcommand{\viewinit}{\mathtt{V}_{\text{init}}}
\newcommand{\memoryinit}{\mathtt{memory}_{\text{init}}}
\newcommand{\Tview}{\mathit{T}}

\newcommand{\ev}[3]{\tup{#1, #2 : #3}}

\newcommand{\progstate}{\overline{p}}

\newcommand{\seq}{\mathbin{;}}

\newcommand{\run}{\mu}
\newcommand{\trace}{\rho}
\newcommand{\traces}{S}

\newcommand{\toeventsgen}[2]{\overline{#1}^{#2}}
\newcommand{\toevents}[1]{\toeventsgen{\trace}{#1}}
\newcommand{\evorder}{<_{\trace}}
\newcommand{\tstep}[2]{#1 : #2}

\newcommand{\tmap}{\mathsf{tmap}}

\newcommand{\smap}{\mathsf{smap}}
\newcommand{\vmap}{\mathsf{vmap}}
\newcommand{\tslot}{\mathsf{tslot}}

\newcommand{\safepoints}{\mathsf{safepoints}}

\newcommand{\etom}{\mathsf{e2m}}
\newcommand{\vmapprop}{\mathsf{vmap}_\proplab}
\newcommand{\vmapfull}{\mathsf{vmap}_{\mathsf{full}}}
\newcommand{\settoview}{\mathsf{set2view}}

\newcommand{\lRs}{(\lR \setminus \lU)}
\newcommand{\lWs}{(\lW \setminus \lU)}

\pgfdeclarelayer{background}
\pgfsetlayers{background,main}

%% file: abstract.tex
Liveness properties, such as termination, of even the simplest shared-memory concurrent programs under
sequential consistency typically require some fairness assumptions about the
scheduler.  Under weak memory models, we observe that the
standard notions of \emph{thread fairness} are insufficient, 
and an additional fairness property, which we call \emph{memory fairness},
is needed.

In this paper, we propose a uniform definition for memory fairness
that can be integrated into any declarative memory model enforcing
acyclicity of the union of the program order and the reads-from relation.
For the well-known models, SC, x86-TSO, RA, and StrongCOH, that have equivalent operational and
declarative presentations, we show that our declarative memory fairness condition
is equivalent to an intuitive model-specific operational notion of memory fairness, 
which requires the memory system to fairly execute its internal propagation steps.
Our fairness condition preserves the correctness of local
transformations and the compilation scheme from RC11 to x86-TSO,
and also enables the first formal proofs of termination of mutual exclusion
lock implementations under declarative weak memory models.

%% file: intro.tex
\section{Introduction}
\label{sec:intro}

Suppose we want to prove termination of a concurrent program under a
full-featured weak memory model, such as RC11~\cite{scfix}.
Sadly, this is not currently possible because RC11 does not support reasoning about liveness.
Extending its formal definition to enable reasoning about liveness properties
is very important because, as shown by \citet[Table 2]{vsync},
multiple existing mutual exclusion lock implementations hang if too few fences are used.
This is also the case for the published version of the HMCS algorithm~\cite{hmcs}:
it contains such a termination bug, a simplified version of which we describe in~\cref{sec:mcslock}.

Termination of concurrent programs typically relies on some fairness assumptions about
concurrency as illustrated by the following program, whose variables are initialized with $0$.
\begin{equation*}
\tag{SpinLoop}\label{ex:spin-loop}
\inarrII{ x := 1 }{ \repeatInst{a := x}{(a \neq 0)} }
\end{equation*}
Under \emph{sequential consistency} (SC), the program can diverge if, \eg thread 2 is always scheduled and thread 1 never gets a chance to run.
This run is considered unfair because although thread 1 is always available to be scheduled, it is never selected.
A standard assumption is \emph{thread fairness}
(which is typically simply called fairness in the literature \cite{Lamport77,Park-ASS80,Lehmann1981,Francez86}),
namely that every (unblocked) non-terminated thread is eventually scheduled.
With a fair scheduler, \ref{ex:spin-loop} is guaranteed to terminate.

Under weak memory consistency,
thread fairness alone does not suffice to ensure termination of \ref{ex:spin-loop}
because merely executing the $x := 1$ write does not mean that
its effect is propagated to the other threads.
Take, for example, the operational TSO model \cite{x86-tso},
where writes are appended to a thread-local buffer
and are later asynchronously applied to the shared memory.
With such a model,
it is possible that the $x := 1$ write is forever stuck in the first thread's buffer
and so thread 2 never gets a chance to read $x=1$.
To rule out such behaviors, we introduce another property, \emph{memory fairness} (MF),
that ensures that threads do not indefinitely observe the same stale memory state.

Operational models can easily be extended to support MF by requiring fairness of the
internal transitions of the model, 
which correspond to the propagation of writes to the different threads.
For the standard interleaving semantics of SC~\cite{lamport-sc},
MF holds vacuously (because the model does not have any internal transitions).
For the usual TSO operational model \cite{x86-tso}, MF requires that every buffered write eventually propagates to the main memory.
For the operational characterization of \emph{release-acquire} (RA)
following \citet{Kang-al:POPL17}, more adaptations are necessary:
(1) we constrain the timestamp ordering
so that no write can overtake infinitely many other writes;
and (2) add a transition that forcefully updates the views of threads
so that all executed writes eventually become globally visible.
The same criteria are required for MF  in the model of \emph{strong coherence} (StrongCOH),
which is essentially a restriction of the promise-free fragment of \citet{Kang-al:POPL17}'s model (as well as of RC11)
to relaxed accesses.

In contrast, it is quite challenging to support MF in declarative (a.k.a.\ axiomatic) models,
which have become the norm for hardware architectures
(x86-TSO~\cite{x86-tso}, Power~\cite{herding-cats}, Arm~\cite{Pulte-al:POPL18}) and
programming languages (e.g., RC11~\cite{scfix}, OCaml~\cite{OCAMLraces},
JavaAtomics~\cite{Bender-al:OOPSLA19}, Javascript~\cite{Watt-al:PLDI20},
WebAssembly~\cite{Watt-al:OOPSLA19}) alike.
In these models, there are no explicit write propagation transitions so that MF could require them to eventually take place.
Further, the memory accesses of different threads are not even totally ordered,
so even the concept of an event eventually happening is not immediate.
We observe, however, neither internal transitions nor a total order are necessary for defining fairness;
what is important is that every event is preceded by only a finite number of other events,
and this can be defined on the execution graphs used by declarative models.

Specifically, for declarative models satisfying $(\lPO\cup\lRF)$-acyclicity 
(\ie acyclicity of the union of the program order and the reads-from relation),
such as RC11, SC, TSO, RA, and StrongCOH,
we show that MF can be defined in a uniform fashion
as prefix-finiteness of the extended coherence order.
The latter is a relation used in declarative models to order accesses to the same location
for guaranteeing SC-per-location~\cite{herding-cats}.
Requiring this relation to be prefix-finite 
means that in a fair execution no write can be preceded by an infinite number
of other events in this order (\eg reads that have not yet observed the write).

We justify the uniform declarative definition of memory fairness in three ways.
First,
we show that our declarative MF condition is equivalent to operational MF
for models that have equivalent declarative and operational presentations
(\ie SC, TSO, RA, and StrongCOH).
This requires extending the existing equivalence results between operational
and declarative models to \emph{infinite} executions, and involves more
advanced constructions that make use of memory fairness.
Second, we show that including our MF condition in the RC11 declarative language model,
which currently lacks any fairness guarantees,
incurs no performance overhead: the correctness of local program
transformations and the compilation scheme to TSO are unaffected.
Third, we show that memory fairness allows lifting robustness theorems about
finite executions to infinite ones.

We finally demonstrate that our declarative MF condition enables verification
of liveness properties of concurrent programs under RC11
by verifying termination and/or fairness of multiple lock implementations
(see \cref{sec:lockexamples}),
including the MCS lock once the fence missing in the presentation of \citet{hmcs}
is added.
Key to those proofs is a reduction theorem we show for the termination of
spinloops.
Under certain conditions about the program, which hold for multiple standard implementations, a spinloop terminates
under a fair model if and only if it exits whenever an iteration reads
only the latest writes in the coherence order.  For example, the loop in
\ref{ex:spin-loop} terminates because reading the latest write ($x := 1$)
exits the loop.

\paragraph{Outline}
In~\cref{sec:op_sem} we define fairness operationally
and incorporate it in the operational definitions of SC, x86-TSO, RA, and StrongCOH\@.
In~\cref{sec:dec_sem} we recap the declarative framework for defining memory models.
In~\cref{sec:fair-dec-sem} we present our declarative MF condition;
we establish its equivalence to the operational MF notions
and show that it preserves the existing compilation and optimization results for RC11
and that it allows lifting of robustness theorems to infinite executions.
In~\cref{sec:lockexamples} we show that the declarative fairness characterization yields an effective method for proving (non-)termination of spinloops
and illustrate it to prove deadlock-freedom and/or fairness of three lock implementations.
We conclude with a discussion of fairness in other models in~\cref{sec:discussion}.

\paragraph{Supplementary Material}
Our technical appendix \cite{appendix} contains typeset proofs for the lemmas and propositions of the article.
We also provide a Coq development \cite{artifact} containing:
\begin{itemize}
\item a formalization of operational and declarative fairness for \SC, \TSO, \RA, and \SCOH;
\item proofs of the aforementioned definitions' equivalence (\cref{thm:equiv_proof});
\item a proof of \cref{thm:spinloop termination basic} stating a sufficient loop termination condition;
\item proofs of termination of the spinlock client and of progress of the ticket lock client for all models satisfying "SC per location" property (which generalizes \cref{thm:spinlock-termination,thm:tickelock}) and of termination of the MCS lock client for \SC, \TSO and \RA (\cref{thm:hmcs-termination} without the RC11 part); and 
\item a proof of infinite robustness property (\cref{cor:robustness}, excluding the RC11 case).
\end{itemize}

%% file: op_sem.tex
\section{What is a Fair Operational Semantics?}
\label{sec:op_sem}

\begin{toappendix}
\subsection{Proofs for \Cref{sec:op_sem}}
\end{toappendix}

In this section, we define our operational framework and its fairness constraints.
We initially demonstrate our terminology for sequential consistency (SC).
In \cref{sec:op_sem_TSO,sec:op_sem_RA,sec:op_sem_SCOH}, we instantiate our framework
to the total store order (TSO), release/acquire (RA), and strong coherence (StrongCOH) models,
and discuss memory fairness in each of these models.

\paragraph{Labeled Transition Systems.}

Our formal development is based on \emph{labeled transition systems} (LTSs),
which we use to represent both programs and operational memory models.
We assume that the transition labels of these systems are split between
\emph{(externally) observable transition labels} and \emph{silent transition labels}.
Using transition labels we define a
\emph{trace} to be a (finite or infinite) sequence of transition labels (of any kind);
whereas an \emph{observable trace} is a
(finite or infinite) sequence of observable transition labels.
Then, LTSs capture sets of traces and observable traces
in the standard way, which is formulated below.

Formally, we define an LTS $\A$ to be a tuple
$\tup{Q,\Sigma,\Theta,\init,\astep{}}$, where $Q$ is a set of \emph{states},
$\Sigma$ is a set of \emph{observable transition labels},
$\Theta$ is a set of \emph{silent transition labels},
$\init\in Q$ is the \emph{initial state},
and ${\astep{}} \suq Q\times (\Sigma\uplus\Theta) \times Q$ is a set of \emph{transitions}.
We denote by $\A.\lQ$, $\A.\lSigma$, $\A.\lTheta$, $\A.\linit$, and $\astep{}_{\A}$ the components of an LTS $\A$.

We denote by $\lSRC(t)$, $\lTLAB(t)$, and $\lTGT(t)$ the three components of a transition $t\in\,\astep{}$.
For $\sigma\in \Sigma\uplus\Theta$,
we write $\astep{\sigma}$ for the relation $\set{\tup{\lSRC(t),\lTGT(t)} \st t\in{} \astep{}, \lTLAB(t)=\sigma}$.
We use $\astep{}$ for the relation $\bigcup_{\sigma\in \Sigma\uplus\Theta} \astep{\sigma}$.
We say that a transition label $\sigma\in \Sigma\uplus\Theta$ is \emph{enabled} in some state $q\in Q$
if $q \astep{\sigma} q'$ for some $q'\in Q$.

A \emph{run} of $\A$ is a (finite or infinite) sequence $\run$ of transitions in $\astep{}_{\A}$ such that
$\lSRC(\run(0))=\A.\linit$ and $\lTGT(\run(k-1))=\lSRC(\run(k))$ for every $k \geq 1$ in $\dom{\run}$.
A run $\run$ of $\A$ \emph{induces} the trace $\trace$ if $\trace(k) = \lTLAB(\run(k))$ for every $k \in \dom{\run}$.
Also, $\run$ induces the \emph{observable} trace $\trace'$ if  $\trace'$ is the restriction
to $\Sigma$ of some trace $\trace$ that is induced by $\run$.

An (observable) trace $\trace$ is called an (observable) trace of $\A$ if it is induced by some run of $\A$.
We write $\OTraces{\A}$ for the set of all observable traces of $\A$
and $\OTracesfin{\A}$ for the set of all finite observable traces of $\A$.

\paragraph{Domains and Event Labels.}

To define programs and their semantics, we fix sets
$\Loc$, $\Tid$, and $\Val$ of \emph{(shared) locations}, \emph{thread identifiers}, and \emph{values} (respectively).
We assume that $\Val$ contains a distinguished value $0$, which serves as the initial value for all locations.
In addition, we assume that $\Tid$ is finite, given by $\Tid=\set{1,2\til N}$ for some $N\geq 1$.
(Our main result below requires $\Tid$ to be finite, see \cref{rem:inf}.) 
We use $\loc,\loca$ to range over $\Loc$;
$\tid,\tida$ to range over $\Tid$;
and $\val$ to range over $\Val$.
Programs interact with the memory using \emph{event labels}, defined as follows.

\begin{definition}
\label{def:event_label}
An \emph{event label} $\lab$ is one of the following:
\begin{itemize}[leftmargin=15pt]
\item Read event label: $\rlab{}{\loc}{\valr}$ where $\loc\in \Loc$ and $\valr\in \Val$.
\item Write event label: $\wlab{}{\loc}{\valw}$ where $\loc\in \Loc$ and $\valw\in \Val$.
\item Read-modify-write label: $\ulab{}{\loc}{\valr}{\valw}$ where $\loc\in \Loc$ and $\valr,\valw\in \Val$.
\end{itemize}
The functions $\lTYP$, $\lLOC$, $\lVALR$, and $\lVALW$
return (when applicable) the type ($\lR/\lW/\lU$), location ($\loc$),
read value ($\val_\lR$), and written value ($\val_\lW$) of a given event label $\lab$.
We denote by $\Lab$ the set of all event labels.
\end{definition}

\begin{remark}
For conciseness, we have not included \emph{fences} in the set of event labels.
In TSO \cite{x86-tso} and RA \cite{Lahav-al:POPL16},
fences can be modeled as read-modify-writes
to an otherwise-unused distinguished location $f$.
\end{remark}

\begin{remark}
Rich programming languages like C/C++ \cite{Batty-al:POPL11}
and Java \cite{Bender-al:OOPSLA19}
as well as the Armv8 multiprocessor \cite{Pulte-al:POPL18}
have multiple kinds of accesses.
This requires us to extend our event labels with additional modifiers.
However, simple event labels as defined above suffice for the purpose of this paper.
\end{remark}

\paragraph{Sequential Programs.}

To keep the presentation abstract, we do not fix a particular programming language,
but rather represent sequential (thread-local) programs as LTSs
with $\Lab$, the set of all event labels, serving as the set of observable transition labels.
For simplicity, we assume that sequential programs do not have silent transitions.\footnote{This
assumption serves us merely to simplify the presentation,
since silent program transitions can be always attached to the next memory access.}
For an example of a toy programming language syntax
and its reading as an LTS, see~\cite{imm}.
In our code snippets throughout the paper, we implicitly assume such a standard interpretation.

We refer to observable traces of sequential programs (\ie sequences over $\Lab$)
as \emph{sequential traces}.

\begin{example}
The simple sequential program 
$\repeatInst{a := x}{(a \neq 0)}$
is formally captured as an LTS with an initial state $\init$ and a state $\mathit{final}$,
and transitions $\tup{\init,\rlab{}{x}{\val},\init}$ for every $\val\in\Val\setminus \set{0}$
and $\tup{\init,\rlab{}{x}{0},\mathit{final}}$. The sequential traces
$\rlab{}{x}{0},\rlab{}{x}{0},\rlab{}{x}{0},\rlab{}{x}{42}$ is an (observable) trace of this program.
The infinite sequential trace
$\rlab{}{x}{0},\rlab{}{x}{0},\dots$ is another (observable) trace of this program.
\end{example}

\paragraph{Concurrent Programs.}

A \emph{concurrent program}, which we also simply call a \emph{program},
is a top-level parallel composition of sequential programs,
defined as a \emph{finite} mapping assigning a sequential program to each thread $\tid\in\Tid$.
A concurrent program $\prog$ induces an LTS
with $\Tid \times \Lab$ serving as the set of observable transition labels
(and no silent transition labels).
This LTS follows the interleaving semantics of $\prog$:
its states are tuples in $\prod_{\tid\in\Tid} \prog(\tid).\lQ$;
the initial state is $\lambda \tid.\; \prog(\tid).\linit$;
and the transitions are given by:
\begin{mathpar}
\inferrule*{
\progstate(\tid) \astep{\lab}_{\prog(\tid)} p \\
}{\progstate \astep{\tstep{\tid}{\lab}}_\prog \progstate[\tid\mapsto p]}
\end{mathpar}
In the sequel, we identify concurrent programs with their induced LTSs.

We refer to observable traces of concurrent programs (\ie sequences over $\Tid \times \Lab$)
as \emph{concurrent traces}.
We denote the two components of a pair $\sigma\in\Tid\times\Lab$
by $\lTID(\sigma)$ and $\lLAB(\sigma)$ respectively.

\paragraph{Behaviors.}

We define a \emph{behavior} to be a function $\beh$ assigning a sequential trace to every thread,
since the events executed by each thread capture precisely what it has observed
about the memory system.

\begin{notation}
The restriction of a concurrent trace $\trace$ to thread $\tid\in\Tid$,
denoted by $\trace\rst{\tid}$, is the sequence obtained from $\trace$
by keeping only the transition labels of the form $\tstep{\tid}{\_}$.
\end{notation}

\begin{definition}
The \emph{behavior induced by a concurrent trace} $\trace$, denoted by $\beh(\trace)$,
is given by
$$\beh(\trace) \defeq \lambda \tid\in\Tid.\; \lambda k \in \dom{\trace\rst{\tid}}.\; \lLAB(\trace\rst{\tid}(k)).$$
This notation is extended to sets of concurrent traces in the obvious way
($\beh(\traces) \defeq \set{\beh(\trace)\st\trace\in\traces}$).
\end{definition}

\begin{notation}
For an LTS $\A$ with $\A.\lSigma=\Tid\times\Lab$,
we denote by $\Beh{\A}$ the set of behaviors induced by observable traces of $\A$
(\ie $\Beh{\A} \defeq \beh(\OTraces{\A})$)
and by $\Behfin{\A}$ the set of behaviors induced by finite observable traces of $\A$
(\ie $\Behfin{\A} \defeq \beh(\OTracesfin{\A})$).
\end{notation}

Since operations of different threads commute in the program semantics, 
the following property easily follows from our definitions.

\begin{proposition}
\label{prop:prog}
For every program $\prog$,
if $\beh(\trace_1)=\beh(\trace_2)$, then
$\trace_1\in \OTraces{\prog}$ iff $\trace_2\in \OTraces{\prog}$.
\end{proposition}
\begin{proofof}{prop:prog}
  W.l.o.g. suppose that $\trace_1\in \OTraces{\prog}$.
  Note that there are no silent transition labels in traces of $P$, so $\trace_1$ is a trace of $\prog$.
  It means that there exists a run $\run_1$ of $P$ inducing $\trace_1$.

  Note that transitions of different sequential programs  are independent of each other.
  That is, they can be reordered in $\run_1$ in such a way that the resulting run $\run_2$ will induce $\trace_2$.
  Since $\beh(\trace_1)=\beh(\trace_2)$, no same thread's transitions reordering are needed to do that.
\end{proofof}

\paragraph{Thread Fairness.}

Not all program behaviors are fair.

\begin{example}
\label{ex:Rloop}
Consider the following program:
\begin{equation}
\tag{\textsf{Rloop}}\label{prog:Rloop}
\inarrII{\assignInst{\loc}{1}}
{L\colon \assignInst{a}{\loc} \\\phantom{L\colon} \ifGotoInst{a=0}{L}}
\end{equation}
The behaviors of this program include the behavior
assigning $\wlab{}{\loc}{1}$ to the first thread
and $\rlab{}{\loc}{1}$ to the second,
but also the (infinite) behavior assigning the empty sequence to the first thread
and the infinite sequence $\rlab{}{\loc}{0},\rlab{}{\loc}{0},\ldots$ to the second.
This behavior occurs if an unfair scheduler only schedules the second thread to run even though the first thread is always
available to execute.\footnote{On this level, without considering a particular memory system (as defined below),
the read values are not restricted whatsoever. Thus,
the behaviors of this program include also any behavior
assigning $\wlab{}{\loc}{1}$ to the first thread
and either $\rlab{}{\loc}{v}$ for some $\val\in\Val\setminus \set{0}$
or the infinite sequence $\rlab{}{\loc}{0},\rlab{}{\loc}{0},\ldots$
to the second thread.
Nonsensical behaviors (with $v\nin\set{0,1}$) are overruled when the program is linked with 
any of the memory systems defined below, with or without ``memory fairness''.}
\end{example}

A natural constraint, which in particular excludes the infinite behavior in the example above,
requires a fair scheduler.
Since our formalism assumes no blocking operations (in particular, locks are implemented using spinloops),
such a scheduler has to ensure that every non-terminated thread is eventually scheduled,
which we formally define as follows.

\begin{definition}
\label{def:thread_fair}
Let $\prog$ be a program.
\begin{itemize}[leftmargin=15pt]
\item A thread $\tid\in\Tid$ is \emph{enabled} in $\progstate \in \prog.\lQ$
if $\tup{\tid,\lab}$ is enabled in $\progstate$ for some $\lab\in\Lab$.
\item A thread $\tid\in\Tid$ is \emph{continuously enabled at index $k$} in an infinite run $\run$ of $\prog$
if it is enabled in $\lSRC(\run(j))$ for every index $j\geq k$.
Thread $\tid$ is \emph{continuously enabled} in $\run$ if it is continuously enabled in $\run$ at some index $k$.
\item A run $\run$ of $\prog$ is \emph{thread-fair} if $\run$ is finite or
for every thread $\tid\in\Tid$ and index $k$ such that $\tid$ is continuously enabled in $\run$ at $k$, there exists $j \geq k$ such that $\lTID(\lTLAB(\run(j)))=\tid$.
\item A \emph{thread-fair observable trace of $\prog$} is any concurrent trace induced
by a thread-fair run of $\prog$.
\item A \emph{thread-fair behavior of $\prog$} is any behavior induced
by a thread-fair observable trace of $\prog$.
We denote by $\Behtf{\prog}$ the set of all thread-fair behaviors of $\prog$.
\end{itemize}
\end{definition}

Returning to \cref{ex:Rloop}, thread-fair behaviors of \ref{prog:Rloop}
are either finite or must assign $\wlab{}{\loc}{1}$ to the first thread.

Again, since operations of different threads commute in the program semantics,
the following property easily follows from our definitions.

\begin{proposition}
\label{prop:prog_tf}
For every program $\prog$,
if $\beh(\trace_1)=\beh(\trace_2)$, then
$\trace_1$ is a thread-fair observable trace of $\prog$
iff
$\trace_2$ is a thread-fair observable trace of $\prog$.
\end{proposition}
\begin{proofof}{prop:prog_tf}
  W.l.o.g. suppose that $\trace_1$ is a thread-fair observable trace of $\prog$ induced by some $\mu_1$.
  By \cref{prop:prog} $\trace_2$ is an observable trace of $\prog$ induced by some $\mu_2$.

  Note that, since $\beh(\trace_1)=\beh(\trace_2)$, every thread's LTS goes via the same sequence of states.
  That is, for each thread $\mu_2$ has the same set of continuously enabled states as $\mu_1$ does.
  Also, since $\beh(\trace_1)=\beh(\trace_2)$, every continuously enabled state is succeeded by the same sequence of labels both in $\mu_1$ and $\mu_2$.
  Thus, if $\trace_1$ is thread-fair then $\trace_2$ should also be.
\end{proofof}

\paragraph{Memory Systems.}

To give operational semantics to programs,
we synchronize them with \emph{memory systems},
which, like programs, are LTSs with $\Tid\times\Lab$ serving as the set of observable transition labels.
In addition, memory systems have silent transition labels, which vary from one system to another.
Intuitively, the set of silent transition labels $\M.\lTheta$ of a memory system $\M$
consists of internal actions that the program cannot observe
(\eg cache-related operations).

The most well-known memory system is that of \emph{sequential consistency}~\cite{lamport-sc},
denoted here by $\M_\SC$,
in which writes by each thread are made immediately visible to all other threads.
$\M_\SC$ tracks the most recent value written to each location.
Its initial state maps each location to zero.
That is, $\M_\SC.\lQ\defeq \Loc \to \Val$ and $\M_\SC.\linit \defeq \lambda \loc.\; 0$.
The system $\M_\SC$ has no silent transitions ($\M_\SC.\lTheta=\emptyset$)
and its transition relation $\astep{}_{\M_\SC}$ is defined as follows:
\begin{mathpar}
\inferrule*{
\m'=\m[\loc \mapsto \val]
}{{\m} \astep{\tstep{\tid}{\wlab{}{\loc}{\val}}}_{\M_\SC} {\m'} }
\and
\inferrule*{
\m(\loc)=v
}{{\m} \astep{\tstep{\tid}{\rlab{}{\loc}{\val}}}_{\M_\SC} {\m} }
\and
\inferrule*{
{\m} \astep{\tstep{\tid}{\rlab{}{\loc}{\val_\lR}}}_{\M_\SC}
\astep{\tstep{\tid}{\wlab{}{\loc}{\val_\lW}}}_{\M_\SC} {\m'}
}{{\m} \astep{\tstep{\tid}{\ulab{}{\loc}{\val_\lR}{\val_\lW}}}_{\M_\SC} {\m'} }
\end{mathpar}
Writing $\val$ to $\loc$ simply updates the value of $\loc$ stored in $\m$.
($\m[\loc \mapsto \val]$ is the function that maps $\loc$ to $\val$
and all other locations $\loca$ to $\m(\loca)$.)
Reading $\val$ from $\loc$ succeeds iff the value stored for $\loc$ in memory is $v$.
The atomic read-modify-write $\ulab{}{\loc}{\val_\lR}{\val_\lW}$
reads location $\loc$ yielding value $\val_\lR$ and immediately writes $\val_\lW$ to it.
Note that $\M_\SC$ is oblivious to the thread that takes the action
($\astep{\tstep{\tid}{\lab}}_{\M_\SC}=\astep{\tida:\lab}_{\M_\SC}$).
The other memory systems below do not have this property.

\paragraph{Linking Programs and Memory Systems.}

By linking programs and memory systems, we
can talk about the behavior of a program $\prog$ under a memory system $\M$.
We say that a certain behavior $\beh$ is a \emph{behavior of a program $\prog$ under a memory system $\M$}
if $\beh$ is both a behavior of $\prog$ and a behavior of $\M$
(\ie $\beh\in \Beh{\prog}\cap\Beh{\M}$).
Similarly, $\beh$ is called a \emph{thread-fair behavior of $\prog$ under $\M$}
if $\beh\in \Behtf{\prog}\cap\Beh{\M}$.

\begin{proposition}
\label{prop:beh_tf_op}
Let $\prog$ be a program, $\M$ be a memory system,
and $\beh$ be a behavior.
\begin{itemize}[leftmargin=15pt]
\item $\beh$ is a behavior of $\prog$ under $\M$ iff
$\beh=\beh(\trace)$ for some $\trace\in\OTraces{\prog}\cap\OTraces{\M}$.
\item $\beh$ is a thread-fair behavior of $\prog$ under $\M$ iff
$\beh=\beh(\trace)$ for some $\trace\in\OTraces{\M}$
that is also a thread-fair observable trace of $\prog$.
\end{itemize}
\end{proposition}
\begin{proofof}{prop:beh_tf_op}
  In the right-to-left implications of both claims the trace $\trace$ belongs both to $\Beh{\M}$ and $\Beh{\prog}$ ($\Behtf{\prog}$ correspondingly) thus satisfying needed conditions.

  Note that for a (thread-fair) behavior $\beh$ under $\M$ there are two traces: one from $\Beh{\prog}$ ($\Behtf{\prog}$) and another from $\Beh{\M}$.
  But since these traces have the same behavior, left-to-right implications can be proved by applying \cref{prop:prog} (\cref{prop:prog_tf} correspondingly) to the trace from $\Beh{\M}$.
\end{proofof}

\begin{example}
Thread-fair behaviors of the program \ref{prog:Rloop} under $\M_\SC$ must be finite.
Indeed, in observable traces of $\M_\SC$, after the first thread performs $\wlab{}{\loc}{1}$,
the second thread will perform $\rlab{}{\loc}{1}$ and terminate its execution.
The behavior $\beh_\text{inf}$ that assigns the empty sequence to the first thread
and the infinite sequence consisting of $\rlab{}{\loc}{0}$ event labels to the second thread
cannot be obtained from a thread-fair run of \ref{prog:Rloop}.
\end{example}

\paragraph{Memory Fairness.}

As we have already discussed, thread-fairness alone is often insufficient
to reason about termination under weak memory models.
For this reason, we introduce \emph{memory fairness} (MF),
which ensures that a thread cannot be lagging behind indefinitely
because the memory system did not propagate certain updates to it. 
We formalize this intuition by having MF require that the memory silent transitions
(responsible for such propagation steps) are scheduled infinitely often.

\begin{definition}
\label{def:memory_fair}
Let $\M$ be a memory system.
\begin{itemize}[leftmargin=15pt]
\item A silent transition label $\theta\in\M.\lTheta$ is \emph{continuously enabled at index $k$} in an infinite run $\run$ of $\M$
if it is enabled in $\lSRC(\run(j))$ for every index $j\geq k$.
The label $\theta$ is \emph{continuously enabled} in $\run$ if it is continuously enabled in $\run$ at some index $k$.
\item A run $\run$ of $\M$ is \emph{memory-fair} if $\run$ is finite or
for every silent memory transition label $\theta\in \M.\lTheta$ and index $k$
such that $\theta$ is continuously enabled in $\run$ at $k$,
there exists $j \geq k$ such that $\lTLAB(\run(j))=\theta$.
\item A \emph{memory-fair observable trace of $\M$} is any concurrent trace induced
by a memory-fair run of $\M$.
\item A \emph{memory-fair behavior of $\M$} is any behavior induced
by a memory-fair observable trace of $\M$.
We denote by $\Behmf{\M}$ the set of all memory-fair behaviors of $\M$.
\end{itemize}
\end{definition}

Linking this definition with programs,
we say that a certain behavior $\beh$ is a \emph{memory-fair behavior of a program $\prog$ under a memory system $\M$}
if $\beh\in \Beh{\prog}\cap\Behmf{\M}$.
Similarly, $\beh$ is called a \emph{thread\&memory-fair behavior of $\prog$ under $\M$}
if $\beh\in \Behtf{\prog}\cap\Behmf{\M}$.

\begin{proposition}
\label{prop:beh_mf_op}
Let $\prog$ be a program, $\M$ be a memory system,
and $\beh$ be a behavior.
\begin{itemize}[leftmargin=15pt]
\item $\beh$ is a memory-fair behavior of $\prog$ under $\M$ iff
$\beh=\beh(\trace)$ for some observable trace $\trace$ of $\prog$
that is also a memory-fair observable trace of $\M$.
\item $\beh$ is a thread\&memory-fair behavior of $\prog$ under $\M$ iff
$\beh=\beh(\trace)$ for some thread-fair observable trace $\trace$ of $\prog$
that is also a memory-fair observable trace of $\M$.
\end{itemize}
\end{proposition}
\begin{proofof}{prop:beh_mf_op}
  In the right-to-left implications of both claims the trace $\trace$ belongs both to $\Behmf{\M}$ and $\Beh{\prog}$ ($\Behtf{\prog}$ correspondingly) thus satisfying needed conditions.

  Note that for a memory-fair (thread\&memory-fair) behavior $\beh$ under $\M$ there are two traces: one from $\Beh{\prog}$ ($\Behtf{\prog}$) and another from $\Behmf{\M}$.
  But since these traces have the same behavior, left-to-right implications can be proved by applying \cref{prop:prog} (\cref{prop:prog_tf} correspondingly) to the trace from $\Behmf{\M}$.
\end{proofof}

Since $\M_\SC.\lTheta=\emptyset$,
every behavior of a program $\prog$ under ${\M_\SC}$ is (vacuously) memory-fair.

\begin{example}
\label{ex:Rloop'}
Consider the following program (assuming that $\loc$ is initialized to $0$):
\begin{equation}
\tag{WWRloop}\label{prog:Rloop'}
\inarrII{L_1\colon \assignInst{\loc}{1} \\
\phantom{L_1\colon} \assignInst{\loc}{0} \\
\phantom{L_1\colon} \gotoInst{L_1}}
{L_2\colon \assignInst{a}{\loc} \\
\phantom{L_2\colon} \ifGotoInst{a=0}{L_2} \\~ }
\end{equation}
The infinite behavior that assigns the infinite sequences
$\wlab{}{\loc}{1},\wlab{}{\loc}{0},\wlab{}{\loc}{1},\wlab{}{\loc}{0},\ldots$,
and  $\rlab{}{\loc}{0},\rlab{}{\loc}{0},\ldots$
to the first and second threads (respectively)
is a thread\&memory-fair behavior of this program under ${\M_\SC}$:
in a corresponding run both threads
are executed infinitely often.
In particular, note that our definitions require that transitions that are \emph{continuously} enabled
are eventually taken, and while the transition $\rlab{}{\loc}{1}$ is infinitely often enabled for the second thread,
it is not continuously enabled.
\end{example}

Next, we demonstrate three weaker memory systems with non-empty sets of silent transitions
that have non-memory-fair traces.
In these systems, whether a program terminates or deadlocks may crucially depend on memory fairness.

\subsection{The Total Store Order Memory System}
\label{sec:op_sem_TSO}

We instantiate memory fairness to the ``Total Store Order'' (TSO) model~\cite{Sewell-al:CACM10,x86-tso}
of the x86 architecture.
This memory system, denoted by $\M_\TSO$, is defined by:
\begin{enumerate}
\item $\M_\TSO.\lQ \defeq (\Loc \to \Val) \times (\Tid \to (\Loc\times \Val)^*)$
\\(Each state consists of a memory and a per-thread store buffer.)
\item $\M_\TSO.\lTheta \defeq \set{\proplab(\tid) \st \tid\in\Tid}$
\\(Silent transitions consist of a propagation label for every thread.)
\item $\M_\TSO.\linit \defeq \tup{\m_0,\b_0}$, where
$\m_0 \defeq \lambda \loc.\; 0$ and $\b_0 \defeq \lambda \tid.\; \epsilon$
(Initially, all buffers are empty.)
\item $\astep{}_{\M_\TSO}$ is given in \cref{fig:tso_transitions}.
\end{enumerate}

\begin{figure*}
\begin{mathpar}
\inferrule*{
\b'=\b[\tid \mapsto \tup{\loc,\val} \cdot \b(\tid)]
}{{\m,\b} \astep{\tid:\wlab{}{\loc}{\val}}_{\M_\TSO} {\m,\b'} }

\inferrule*{
\b(\tid)=\tup{\loc_n,\val_n} \cdottil \tup{\loc_1,\val_1} \\\\
\m[\loc_1 \mapsto \val_1]
{\cdot}\mkern1mu{\cdot}\mkern1mu{\cdot} %
[\loc_n \mapsto \val_n](\loc)=\val
}{{\m,\b} \astep{\tid:\rlab{}{\loc}{\val}}_{\M_\TSO} {\m,\b} }

\inferrule*{
\b(\tid)=\epsilon \\
\m(\loc) = \val_\lR
}{{\m,\b} \astep{\tid:\ulab{}{\loc}{\val_\lR}{\val_\lW}}_{\M_\TSO} {\m[\loc \mapsto \val_\lW],\b} }

\inferrule*{
\b(\tid)= \buff \cdot \tup{\loc,\val}
}{{\m,\b} \astep{\proplab(\tid)}_{\M_\TSO} {\m[\loc \mapsto \val],\b[\tid \mapsto \buff]} }
\end{mathpar}
\caption{Transitions of $\M_\TSO$}
\label{fig:tso_transitions}
\end{figure*}

In addition to the global memory $\m$, states of $\M_\TSO$
include a mapping $\b$ assigning a FIFO \emph{store buffer} to every thread.
Writes are first written to the local buffer and later non-deterministically propagate to memory (in the order in which they were issued).
Reads read the most recent value of the relevant location in the thread's buffer and refer to the memory if such value
does not exist.
RMWs can only execute when the thread's buffer is empty and write their result
in the memory directly.

\begin{example}[Store Buffering]
\label{ex:SB}
The following annotated behavior is \emph{allowed} under $\M_\TSO$ (but not under $\M_\SC$):
\begin{equation}
\tag{SB}\label{prog:SB}
\inarrII{\assignInst{\loc}{1} \\ \assignInst{a}{\loca} \comment{reads 0} }
{\assignInst{\loca}{1} \\ \assignInst{a}{\loc} \comment{reads 0} }
\end{equation}
Indeed, the first thread may run first, but the write of $1$ to $\loc$
may remain in its store buffer.
Then, when the second thread runs, it reads the initial value ($0$) of $\loc$ from the memory.
\end{example}

\begin{example}
Revisiting the \ref{prog:Rloop} program from \cref{sec:op_sem},
unlike under ${\M_\SC}$,
thread-fair behaviors of $\ref{prog:Rloop}$ under ${\M_\TSO}$ include
the (infinite) behavior assigning the $\wlab{}{\loc}{1}$ to the first thread
and the infinite sequence $\rlab{}{\loc}{0},\rlab{}{\loc}{0},\ldots$ to the second.
Indeed, the entry $\tup{\loc,1}$ may indefinitely remain in the first thread's buffer,
so that $\wlab{}{\loc}{1}$ is never executed from the point of view of the second thread.
To disqualify this behavior, we need to further require \emph{memory} fairness.
Indeed, in runs inducing this infinite behavior, the silent memory transition
$\proplab(1)$ is necessarily continuously enabled.
Memory fairness requires that $\proplab(1)$ will be eventually executed,
and from that point on $\M_\TSO$ prohibits the second thread from executing $\rlab{}{\loc}{0}$.
\end{example}

We note that the notion of memory fairness is sensitive to the choice of
silent memory transitions.
For example, consider an alternative memory system, denoted by $\M_\TSO'$,
with less informative silent transition labels that do not record the thread identifier of the propagated write.
(Formally $\M_\TSO'$ is defined just like $\M_\TSO$ except for
 $\M_\TSO'.\lTheta \defeq \set{\proplab}$,
 and the label of the propagation step is $\proplab$ rather than $\proplab(\tid)$.)
Then, $\M_\TSO'$ induces the same set of behaviors as $\M_\TSO$,
but not the same set of \emph{memory fair} behaviors.
In particular, we can extend the \ref{prog:Rloop} program with an additional thread that constantly writes
to some unrelated location $\loca$,
and obtain a memory fair run of $\M_\TSO'$
by infinitely often propagating a write to $\loca$,
but never propagating the $\wlab{}{\loc}{1}$ entry.

\subsection{The Release/Acquire Memory System}
\label{sec:op_sem_RA}

We instantiate our operational framework with a memory system for Release/Acquire (RA),
enriched with silent memory transitions for capturing fair behaviors.
Here we follow an operational formulation of RA from \citet{kaiser17},
based on the Promising Semantics of \citet{Kang-al:POPL17}.

The memory of the RA system records a (finite) set of \emph{messages},
each of which corresponds to some write that was previously executed.
Messages (of the same location) are ordered using \emph{timestamps},
and carry a \emph{view}---a mapping from locations to timestamps.
In turn, the states of this memory system also keep track of the current view of each thread,
and use these views to confine the set of messages that threads may read and write.
In particular, if a thread has observed (either by reading or by writing itself)
a message whose view $\view$ has $\view(\loc)=\ts$,
then it can only read messages of $\loc$ whose timestamp is greater than or equal to $\ts$.

To formally define this system, we let $\Time\defeq\N$ (using natural numbers as {timestamps}),
$\View\defeq\Loc \to \Time$ (the set of {views}),
and $\Msg\defeq\Loc \times \Val\times\Time\times\View$ (the set of messages).
We denote a message $\msga$ as a tuple of the form $\msg{\loc}{\val}{\ts}{\view}$,
where $\loc\in\Loc$, $\val\in\Val$, $\ts\in\Time$, and $\view\in\View$.
We write $\lLOC(\msga)$, $\lVAL(\msga)$, $\lTS(\msga)$, and $\lVIEW(\msga)$
to refer to the components of a message $\msga$.
The usual order $<$ on natural numbers is lifted pointwise to a partial order on views;
 $\sqcup$ denotes the pointwise maximum on views;
 and $\view_0$ is the minimum view
($\view_0 \defeq \lambda \loc.\; 0$).

With these definitions and notations, the RA memory system, denoted here by $\M_\RA$, is defined as follows
(additional silent memory transitions are discussed below):
\begin{enumerate}
\item $\M_\RA.\lQ  \defeq \powerset{\Msg} \times (\Tid \to \View)$.
\item $\M_\RA.\linit \defeq \tup{\m_0,\lambda \tid.\; \view_0}$, where
  the initial memory is $\m_0 \defeq \set{\msg{\loc}{0}{0}{\view_0} \st \loc\in\Loc}$.
\item $\astep{}_{\M_\RA}$ is given in \cref{fig:ra_transitions}.
\end{enumerate}

\begin{figure*}
\begin{mathpar}
\inferrule*{
\nexists \msga \in\m.\; \lLOC(\msga)=\loc \land \lTS(\msga)=\ts \\\\
\Tview(\tid)(\loc) < \ts \\\\
\Tview'=\Tview[\tid \mapsto \Tview(\tid)[\loc \mapsto \ts]] \\\\
\m'=\m \cup \set{\msg{\loc}{\val}{\ts}{\Tview'(\tid)}}
}{\tup{\m,\Tview} \astep{\tstep{\tid}{\wlab{}{\loc}{\val}}}_{\M_\RA} \tup{\m',\Tview'}}
\and
\inferrule*{
\msg{\loc}{\val}{\ts}{\view} \in \m \\\\
\Tview(\tid)(\loc) \leq \ts \\\\
\Tview'=\Tview[\tid \mapsto \Tview(\tid) \sqcup \view]
}{\tup{\m,\Tview} \astep{\tstep{\tid}{\rlab{}{\loc}{\val}}}_{\M_\RA} \tup{\m,\Tview'}}
\and
\inferrule*{
\tup{\m,\Tview} \astep{\tstep{\tid}{\rlab{}{\loc}{\val_\lR}}}_{\M_\RA} \tup{\m,\Tview'}
  \astep{\tstep{\tid}{\wlab{}{\loc}{\val_\lW}}}_{\M_\RA} \tup{\m',\Tview''}
\\ \Tview''(\tid)(\loc)=\Tview'(\tid)(\loc)+1
}{\tup{\m,\Tview} \astep{\tstep{\tid}{\ulab{}{\loc}{\val_\lR}{\val_\lW}}}_{\M_\RA} \tup{\m',\Tview''}}
\end{mathpar}
\caption{Transitions of $\M_\RA$}
\label{fig:ra_transitions}
\end{figure*}

The states of $\M_\RA$ consist of a set $\m$ of all messages added to the memory so far
and a mapping $\Tview$ assigning a view to each thread.
Write steps of thread $\tid$ writing to location $\loc$ pick a timestamp $\ts$ that is fresh
($\nexists \msga \in\m.\; \lLOC(\msga)=\loc \land \lTS(\msga)=\ts$)
and greater than the latest timestamp that $\tid$ has observed for $\loc$
($\Tview(\tid)(\loc) < \ts$);
update the thread's view to include this timestamp
($\Tview'=\Tview[\tid \mapsto \Tview(\tid)[\loc \mapsto \ts]]$);
and add a corresponding message to the memory carrying the (updated) thread view
($\m'=\m \cup \set{\msg{\loc}{\val}{\ts}{\Tview'(\tid)}}$).
Read steps of thread $\tid$ reading from location $\loc$ pick a message from the current memory
($\msg{\loc}{\val}{\ts}{\view} \in \m$) whose timestamp
is greater than or equal to the latest timestamp that $\tid$ has observed for $\loc$
($\Tview(\tid)(\loc) \leq \ts$);
and incorporate the message's view in the thread view
($\Tview'=\Tview[\tid \mapsto \Tview(\tid) \sqcup \view]$).
RMW steps are defined as atomic sequencing of a read step followed by a write step,
with the restriction that the new message's (fresh) timestamp is the successor of the timestamp of the read message
($\Tview''(\tid)(\loc)=\Tview'(\tid)(\loc)+1$).
The latter condition is needed to ensure the atomicity of RMWs: no other write can intervene
between the read part and the write part of the RMW (\ie no message can be placed between
the read and the written messages in the timestamp order).

\begin{example}[Message passing]
\label{ex:MP}
The following annotated behavior is \emph{disallowed} under $\M_\RA$:
\begin{equation}
\tag{MP}\label{prog:MP}
\inarrII{\assignInst{\loc}{1} \\ \assignInst{\loca}{1}}
{\assignInst{a}{\loca} \comment{reads 1} \\ \assignInst{b}{\loc} \comment{reads 0}}
\end{equation}
Indeed, the second thread can read $1$ for $\loca$,
only after the first thread added two messages
$m_\loc=\msg{\loc}{1}{\ts _\loc}{[\loc \mapsto \ts _\loc]}$
and
$m_\loca=\msg{\loca}{1}{\ts _\loca}{[\loc \mapsto \ts _\loc, \loca \mapsto \ts _\loca]}$ to the memory
with $\ts _\loc, \ts _\loca >0$.
When reading $m_\loca$, the second thread increases its view of $\loc$ to be $\ts _\loc$.
Since $\ts _\loc >0$, it is then unable to read the initial message of $\loc$,
and must read $m_\loc$.
\end{example}

\begin{example}
\label{ex:2RMW}
By forcing RMWs to use the successor of the read message as the timestamp of the written message,
$\M_\RA$ forbids different RMWs to read the same message.
To see this, consider the following example (where $\faddInstn$ denotes an atomic fetch-and-add instruction
that returns its read value):
\begin{equation}
\tag{2RMW}\label{prog:2RMW}
\inarrII{\incInst{a}{\loc}{1} \comment{reads 0}}{\incInst{b}{\loc}{1} \comment{reads 0}}
\end{equation}
W.l.o.g., if  the first runs first, it reads from the initialization message $\msg{\loc}{0}{0}{\view_0}$
(it is the only message of $\loc$ in $\m_0$),
and it is forced to add a message \emph{with timestamp $1$}.
When the second thread runs, it may \emph{not} read from the initialization message: that would again require adding a message of $\loc$ with timestamp $1$, but that timestamp is no longer available.
Thus, it may only read from the message that was added by the first thread.
\end{example}

\begin{example}
\label{ex:SB+RMWs}
Fences (modeled as RMWs to an otherwise unused distinguished location $f$)
can be used to recover sequential consistency when needed.
The following outcome is forbidden by \RA.
\begin{equation}
\tag{SB+RMWs}\label{prog:SB+RMWs}
\inarrII{\assignInst{\loc}{1} \\ \incInstP{f}{0} \\ \assignInst{a}{\loca} \comment{reads 0}}
{\assignInst{\loca}{1} \\ \incInstP{f}{0} \\ \assignInst{b}{\loc} \comment{reads 0}}
\end{equation}
Due to the RMWs in both threads, $\M_\RA$ forbids the annotated program behavior.
Indeed, suppose, w.l.o.g., that the first thread executes its $\incInstP{f}{0}$ first,
it will read from the initialization message to $f$ and will add to memory
a message of the form $\msg{f}{0}{1}{\view}$ with $\view(\loc) > 0$.
When the second thread executes its $\incInstP{f}{0}$,
it will necessarily read that message and incorporate the view $\view$ in its thread view,
so that its view of $\loc$ will be increased.
Then, when it reads $\loc$ it may not pick the initial message.
\end{example}

The RA memory system defined so far (with no silent transitions) allows non-fair executions.
In particular, it allows messages added by some thread to never propagate to other threads,
so that other threads may forever read a message with a lower timestamp, and thus, allows, \eg
a thread-fair \emph{infinite} behavior for the \ref{prog:Rloop} program from \cref{sec:op_sem}.

To address this problem, we include silent memory transitions in $\M_\RA$,
labeled with tuples of the form $\proplab(\tid,\msga)$, where $\tid\in\Tid$ and $\msga\in\Msg$
(\ie $\M_\RA.\lTheta \defeq \set{\proplab(\tid,\msga) \mid \tid \in \Tid, \msga \in \Msg}$).
Then, we include in $\M_\RA$ the following silent memory step:
\begin{equation*}
\inferrule[RA-propagate]{
\msga\in\m  \\ \Tview(\tid)(\lLOC(\msga)) < \lTS(\msga)
}{\tup{\m,\Tview} \astep{\proplab(\tid,\msga)}_{\M_\RA} \tup{\m,\Tview[\tid \mapsto \Tview(\tid)[\lLOC(\msga) \mapsto \lTS(\msga)]]}}
\end{equation*}
For a given thread $\tid$ and message $\msga$ that has not been yet observed by thread $\tid$
($\Tview(\tid)(\lLOC(\msga)) < \lTS(\msga)$),
this step increases $\tid$'s view to include $\msga$'s timestamp.
Intuitively speaking, it ensures that %
every thread $\tid$ eventually advances
its view so that it cannot keep reading an old message indefinitely.

\begin{example}
\label{ex:Rloop_RA}
While thread-fair behaviors of $\ref{prog:Rloop}$ under ${\M_\RA}$ include
an infinite behavior (in which the second thread indefinitely read the initialization message),
memory fairness forbids this behavior.
Indeed, in runs inducing this infinite behavior, a silent label
$\proplab(2,\msg{\loc}{1}{\ts}{[\loc \mapsto \ts]})$ (where $\ts$ is a timestamp of a message added by instruction $\assignInst{\loc}{1}$ of $\ref{prog:Rloop}$)
is necessarily continuously enabled.
Memory fairness ensures that the corresponding transition is eventually executed,
and from that point on, $\M_\RA$ prohibits the second thread from executing $\rlab{}{\loc}{0}$.
\end{example}

We emphasize again that
memory fairness is sensitive to the choice of silent memory transitions.
For instance, the system obtained from $\M_\RA$ by discarding the message $\msga$ from the labels of silent memory steps
 induces the same set of behaviors as $\M_\RA$,
but not the same set of \emph{memory fair} behaviors.
In the next sections, we present the declarative approach for defining the semantics of memory systems,
which uniformly captures  memory fairness, and does not require the technical ingenuity needed for ensuring fairness in
operational memory systems.

\subsection{The Strong-Coherence Memory System}
\label{sec:op_sem_SCOH}

We consider a memory system for Strong-Coherence (StrongCOH), \ie the relaxed fragment of RC11. 
Similar to RA, we follow an operational formulation of StrongCOH following the relaxed and 
promise-free fragment of the Promising Semantics of \citet{Kang-al:POPL17}.
Since this operational formulation is very close to RA's one discussed above,
we describe only the difference between them.

The states of $\M_\SCOH$ are the same as of $\M_\RA$, and transitions are similar,
where the only difference is in the read transition (note the crossed out ``$\sqcup \view$''):
\begin{equation*}
\inferrule*{
\msg{\loc}{\val}{\ts}{\view} \in \m \\
\Tview(\tid)(\loc) \leq \ts \\
\Tview'=\Tview[\tid \mapsto \Tview(\tid)[\loc \mapsto \ts] \cancel{\sqcup \view}]
}{\tup{\m,\Tview} \astep{\tstep{\tid}{\rlab{}{\loc}{\val}}}_{\M_\SCOH} \tup{\m,\Tview'}}
\end{equation*}
That is, when a thread reads from a message, it does not update its view by the message's view
but just by its timestamp.%
\footnote{In this model one may change messages to not store views at all
since they are never used. We keep the message views only in order to be 
as close as possible to RA.}
This change makes the semantics weaker:
StrongCOH allows weak behavior of~\ref{prog:MP} and~\ref{prog:SB+RMWs} from~\cref{ex:MP,ex:SB+RMWs}.

We include the same silent memory transitions in $\M_\SCOH$ as we do for $\M_\RA$,
which is enough to guarantee termination of memory-fair executions of \ref{prog:Rloop} for the same reason
as for RA.

%% file: dec_sem.tex
\section{Preliminaries on Declarative Semantics}
\label{sec:dec_sem}

In this section, we review the declarative (\aka axiomatic) framework
for assigning semantics to concurrent programs
and present the well-known declarative models for the four operational models presented above.
Later, we will extend the framework and the existing correspondence results with fairness guarantees
that account for infinite behaviors.

\paragraph{Relations}

Given a binary relation (in particular, a function) $R$, $\dom{R}$ and $\codom{R}$ denote its domain and co\-domain.
We write $R^?$, $R^+$, and $R^*$ respectively to denote its reflexive, transitive, and reflexive-transitive closures.
The inverse relation is denoted by $R^{-1}$.
We denote by $R_1 \seq R_2$ the (left) composition of two relations $R_1,R_2$,
and assume that $\seq$ binds tighter than $\cup$ and $\setminus$.
We denote by $[A]$ the identity relation on a set $A$.
In particular, $[A] \seq R \seq [B] = R\cap (A\times B)$.
For $n\geq 0$ and a relation $R$ on a set $A$, $R^n$ is recursively defined by $R^0\defeq [A]$
and $R^{n+1}\defeq R \seq R^n$.
We write $R^{\leq n}$ for the union $\bigcup_{1\leq i\leq n} R^n$.

\paragraph{Events}

Events represent individual memory accesses in a run of a program.
They consist of a thread identifier, an event label, and a serial number used to 
uniquely identify events and order the events inside each thread.

\begin{definition}
\label{def:event}
An \emph{event} $e$ is a tuple $\ev{k}{\tid}{\lab}$ where
$k\in\N \cup \set{\bot}$ is a serial number inside each thread ($\bot$ for initialization events),
$\tid\in\Tid \uplus\set{\bot}$ is a thread identifier
($\bot$ for initialization events),
and $\lab\in\Lab$ is an event label (as defined in \cref{def:event_label}).
The functions $\lSN$, $\lTID$, and $\lLAB$
return the serial number, thread identifier, and the event label of an event.
The functions $\lTYP$, $\lLOC$, $\lVALR$, and $\lVALW$ are lifted to events in an obvious way.
We denote by $\Event$ the set of all events, and 
use $\lR,\lW$, and $\lU$ to denote the following subsets:
\begin{align*}
 \lR &\defeq \set{e\in\Event\st \lTYP(e)=\lR \lor \lTYP(e)=\lU} \\
 \lW &\defeq \set{e\in\Event\st \lTYP(e)=\lW \lor \lTYP(e)=\lU}\\
 \lU &\defeq \set{e\in\Event\st \lTYP(e)=\lU } 
\end{align*}
We use subscripts and superscripts to restrict sets of events to certain location
and thread 
(\eg $\lW_\loc=\set{w\in \lW \st \lLOC(w)=\loc}$
and $E^\tid= \set{e\in E \st \lTID(e)=\tid}$).
The set of \emph{initialization events}
is given by $\Init \defeq \set{\ev{\bot}{\bot}{\wlab{}{\loc}{0}} \st \loc \in \Loc}$.
\end{definition}

\begin{notation}
Given a relation $R$ on events, we denote by $R\rst{\lLOC}$ the restriction of $R$ to events of
the same location:
$$R\rst{\lLOC} = \set{\tup{e_1,e_2}\in R \st \exists \loc\in\Loc.\; \lLOC(e_1)=\lLOC(e_2)=\loc}$$
\end{notation}

Our representation of events induces a \emph{sequenced-before} partial order on events given by:
$$e_1 < e_2 \;\defiff\; 
\inarr{ (e_1 \in \Init \land e_2 \nin \Init) \lor {}
 (\lTID(e_1) =\lTID(e_2) \land \lSN(e_1) < \lSN(e_2)) } $$
Initialization events precede all non-initialization events,
while events of the same thread are ordered according to their serial numbers.

Behaviors (\ie mappings from threads to sequential traces)
are associated with sets of events in an obvious way:

\begin{definition}
\label{def:gen}
The set of events \emph{extracted from a behavior} $\beh$,
denoted by $\Event(\beh)$, is given by 
$\Event(\beh) \defeq\Init \cup \set{\ev{k}{\tid}{\beh(\tid)(k)} \st \tid\in\Tid, k\in\dom{\beh(\tid)}}$.
\end{definition}

It is easy to see that for every behavior $\beh$,
$\Event(\beh)$ satisfies certain ``well-formedness'' properties:

\begin{definition}
\label{def:wf}
A set $E\suq\Event$ is \emph{well-formed} if the following hold:
\begin{itemize}[leftmargin=15pt]
\item $\Init\suq E$.
\item $\lTID(e)\neq \bot$ and $\lSN(e)\neq \bot$ for every $e\in E\setminus \Init$.
\item If $\lTID(e_1)=\lTID(e_2)$ and $\lSN(e_1)=\lSN(e_2)$, then $e_1 = e_2$ for all $e_1, e_2 \notin \Init$.
\item For every $e\in E \setminus \Init$ and $0 \leq k <  \lSN(e)$, there exists $\lab\in\Lab$
such that $\ev{k}{\lTID(e)}{\lab}\in E$.
\end{itemize}
\end{definition}

\paragraph{Execution Graphs}

An execution graph consists of a set of events, a \emph{reads-from} mapping that determines the write event from which
each read reads its value, and a \emph{modification order} which totally orders the writes to each location.

\begin{definition}
\label{def:execution}
An \emph{execution graph} $G$ is a tuple $\tup{E,\rf,\mo}$ where:
\begin{enumerate}
\item $E$ is a well-formed (possibly, infinite) set of events.
\item $\rf$, called \emph{reads-from}, is a relation on $E$ satisfying:
\begin{itemize}[leftmargin=15pt]
\item If $\tup{w,r}\in \rf$ then $w\in\lW$, $r\in\lR$, $\lLOC(w)=\lLOC(r)$, and $\lVALW(w)=\lVALR(r)$.
\item $w_1=w_2$ whenever $\tup{w_1,r},\tup{w_2,r}\in \rf$ (that is, $\rf^{-1}$ is functional).
\item $E \cap \lR \suq \codom{\rf}$ (every read should read from some write).
\end{itemize}
\item $\mo$, called \emph{modification order}, is a disjoint union of relations $\set{\mo_\loc}_{\loc\in\Loc}$,
such that each $\mo_\loc$ is a strict total order on $E \cap \lW_\loc$.
\end{enumerate}
We denote the components of $G$ by $G.\lE$, $G.\lRF$, and $G.\lMO$,
and write $G.\lPO$ (called \emph{program order}) for the restriction of sequenced-before to $G.\lE$
(\ie $G.\lPO \defeq [G.\lE]; < ; [G.\lE]$).
For a set $E'\suq \Event$, we write $G.E'$ for $G.\lE \cap E'$ (\eg $G.\lW=G.\lE\cap \lW$).
The set of all execution graphs is denoted by $\Exec$.
\end{definition}

A \emph{declarative memory system} is simply a set $\G$ of execution graphs
(often formulated using a conjunction of several constraints).
We refer to execution graphs in a declarative memory system $\G$
as \emph{$\G$-consistent} execution graphs.

We can now define the behaviors allowed by a given declarative memory system.

\begin{definition}
A behavior $\beh$ is \emph{allowed by a declarative memory system} $\G$
if $\Event(\beh)=G.\lE$ for some execution graph $G\in\G$.
We denote by $\Beh{\G}$ ($\Behfin{\G}$) the set of all (finite) behaviors that are allowed by $\G$.
\end{definition}

The linking with programs is defined as follows.

\begin{definition}
\label{def:beh_tf_dec}
Let $\prog$ be a program, $\G$ be a declarative memory system,
and $\beh$ be a behavior.
\begin{itemize}[leftmargin=15pt]
\item $\beh$ is a 
\emph{behavior of $\prog$ under $\G$}  if $\beh\in\Beh{\prog}\cap\Beh{\G}$.
\item $\beh$ is a 
\emph{thread-fair behavior of $\prog$ under $\G$}  if $\beh\in\Behtf{\prog}\cap\Beh{\G}$.
\end{itemize}
\end{definition}

\subsection{A Declarative Memory System for SC}

To provide a declarative formulation of \SC, following \citet{herding-cats}, 
we use the standard ``from-read'' relation (\aka ``reads-before'').
In this relation a read $r$ is ordered before a write $w$ 
if $r$ reads from a write $w'$ that is earlier than $w$ in the modification order.

\begin{definition}
\label{def:fr}
The \emph{from-read} relation for an execution graph $G$,
denoted by $G.\lFR$, is defined by:
$$G.\lFR \defeq (G.\lRF^{-1} \seq G.\lMO) \setminus [G.\lE].$$
\end{definition}

Note that we have to explicitly subtract the identity relation from $ G.\lRF^{-1} \seq G.\lMO$
for making sure that $\lU$ events are not $G.\lFR$-ordered before themselves.

Having defined $\lFR$, the ``\SC-happens-before'' relation is given by:
$$G.\lHBSC \defeq (G.\lPO \cup G.\lRF \cup G.\lMO \cup G.\lFR)^+$$

In turn, \SC consistency requires that $G.\lHBSC$ is irreflexive:
$$\G_\SC\defeq \set{G\in\Exec \st G.\lHBSC \text{ is irreflexive}}$$

Intuitively speaking, every trace of $\M_\SC$ induces an execution graph $G$
with irreflexive $G.\lHBSC$; and, conversely, every total order on $G.\lE$ that extends
$G.\lHBSC$ is essentially a trace of $\M_\SC$.
The following standard theorem formalizes these claims for \emph{finite} executions:

\begin{theorem}[\cite{herding-cats}]
\label{thm:sc-equiv_basic}
$\Behfin{\M_\SC}=\Behfin{\G_\SC}$.
\end{theorem}

\begin{example}
\label{ex:SB_dec}
$\G_\SC$ forbids the 
annotated outcome of the \ref{prog:SB} program from \cref{ex:SB}
because the following graph is $\G_\SC$-inconsistent
($\evlab{\lW}{}{\loc}{0}$ and $\evlab{\lW}{}{\loca}{0}$ are the implicit initialization writes):
\begin{equation*}
\small
\begin{tikzpicture}[yscale=0.9,xscale=1.5]
  \node (0x)  at (-2,0.8) {$\evlab{\lW}{}{\loc}{0}$};
  \node (0y)  at (-2,0.2) {$\evlab{\lW}{}{\loca}{0}$};
  \node (11)  at (0,1) {$\evlab{\lW}{}{\loc}{1}$ };
  \node (12)  at (1,1) {$\evlab{\lR}{}{\loca}{0}$ };
  \node (21)  at (0,0) {$\evlab{\lW}{}{\loca}{1}$ };
  \node (22)  at (1,0) {$\evlab{\lR}{}{\loc}{0}$ };
  \draw[po] (11) edge (12);
  \draw[po] (21) edge (22);
  \draw[po] (0x) edge (11) edge (21);
  \draw[po] (0y) edge (21) edge (11);
  \draw[mo,transform canvas={yshift=0.5ex}] (0x) edge node[above] {\smaller$\lMO$} (11);
  \draw[rf] (0y) edge node[below,pos=0.74] {\smaller$\lRF$} (12);
  \draw[mo,transform canvas={yshift=-0.5ex}] (0y) edge node[below] {\smaller$\lMO$} (21);
  \draw[rf] (0x) edge (22);
\end{tikzpicture}
\end{equation*}
Indeed, to get the desired behavior, the $\lRF$-edges are forced because of the read values.
Since $\lMO$ cannot contradict $\lPO$ (they are both included in $\lHBSC$),
the $\lMO$-edges are also forced as depicted above.
We obtain $\lFR$-edges from $\evlab{\lR}{}{\loc}{0}$ to $\evlab{\lW}{}{\loc}{1}$
and from $\evlab{\lR}{}{\loca}{0}$ to $\evlab{\lW}{}{\loca}{1}$,
which, in turn, imply a $\lHBSC$-cycle composed of two $\lPO$ and two $\lFR$ edges.
\end{example}

\subsection{A Declarative Memory System for TSO}\label{sec:declarative TSO memory system}

Following \citet{herding-cats}, a declarative formulation for \TSO
is easily obtained from the one of \SC, by removing from the transitive closure
in $\lHBSC$ the program order edges from writes to reads
that are not necessarily ``preserved'' in \TSO.
Indeed, because writes are buffered in \TSO,
roughly speaking, the effect of a write in \TSO may be delayed
\wrt subsequent reads.
By contrast, it cannot be delayed \wrt subsequent writes,
since entries in the \TSO buffers propagate in a FIFO fashion.

When removing the write to read program order edges,
we need to explicitly enforce ``\SC per-location'' (\aka coherence),
which takes care of intra-thread write-read pairs (a read $r$ from $\loc$ that is later in program order than a write $w$ to $\loc$
may not read from a write that is $\lMO$-earlier than $w$).
To achieve this, the model employs the following derived relations:
\begin{align*}
G.\lRFE &\defeq G.\lRF \setminus G.\lPO  \tag{\emph{external reads-from}} \\
G.\lPPO &\defeq G.\lPO \setminus (\lWs \times \lRs)  \tag{\emph{preserved program order}} \\
G.\lHBTSO &\defeq (G.\lPPO \cup G.\lRFE \cup G.\lMO \cup G.\lFR)^+ \tag{\emph{\TSO-happens-before }} \\
G.\lSC_{\lLOC}  &\defeq (G.\lPO\rst{\lLOC} \cup G.\lRF \cup G.\lMO \cup G.\lFR)^+ \tag{\emph{SC-per-location order}} 
\end{align*}

Then, \TSO consistency requires that 
$G.\lHBTSO$ and $G.\lSC_{\lLOC}$ are irreflexive:
$$\G_\TSO\defeq \{ G\in\Exec \st G.\lHBTSO \text{\  and\ } G.\lSC_{\lLOC} \text{\ are irreflexive}\}$$

\begin{theorem}[\cite{herding-cats}]
\label{thm:tso-equiv_basic}
$\Behfin{\M_\TSO}=\Behfin{\G_\TSO}$.
\end{theorem}

The execution graph for the \ref{prog:SB} program in \cref{ex:SB_dec}
is $\G_\TSO$-consistent. In particular, the two
$\lPO$ edges that participate in the $G.\lHBSC$ cycle
are from a write to a read, so none of them is included in $G.\lHBTSO$.

\subsection{A Declarative Memory System for RA}

The declarative model for \RA is obtained by strengthening the SC per-location requirement
to use \RA's happens-before relation instead of the program order:
\begin{align*}
G.\lHBRA & \defeq (G.\lPO \cup G.\lRF)^+ \tag{\emph{\RA-happens-before}} \\
G.\lRA_{\lLOC}  &\defeq (G.\lHBRA\rst{\lLOC} \cup G.\lRF \cup G.\lMO \cup G.\lFR)^+ \tag{\emph{\RA-per-location order}} 
\end{align*}

Then, \RA consistency requires that $G.\lRA_{\lLOC}$ is irreflexive:
$$\G_\RA\defeq \set{G\in\Exec \st G.\lRA_{\lLOC} \text{\ is irreflexive}}$$

\begin{example}
\label{ex:RA-MP}
The annotated outcome of the \ref{prog:MP} program from \cref{ex:MP}
is disallowed by $\G_\RA$ because the following (partially depicted) execution graph is 
$\G_\RA$-inconsistent:
\begin{equation*}
\small
\begin{tikzpicture}[yscale=0.9,xscale=1.5]
  \node (0x)  at (-2,0.8) {$\evlab{\lW}{}{\loc}{0}$};
  \node (0y)  at (-2,0.2) {$\evlab{\lW}{}{\loca}{0}$};
  \node (12)  at (0,1) {$\evlab{\lW}{}{\loc}{1}$ };
  \node (13)  at (1,1) {$\evlab{\lW}{}{\loca}{1}$ };
  \node (21)  at (0,0) {$\evlab{\lR}{}{\loca}{1}$ };
  \node (22)  at (1,0) {$\evlab{\lR}{}{\loc}{0}$ };
  \draw[po] (0x) edge (12) edge (21);
  \draw[po] (0y) edge (21) edge (12);
  \draw[po] (12) edge (13);
  \draw[po] (21) edge (22);
  \draw[rf] (13) edge node[above] {\smaller$\lRF\ $} (21);
  \draw[rf] (0x) edge  (22);
  \draw[mo,transform canvas={yshift=0.5ex}] (0x) edge node[above]{$\lMO$} (12);
\end{tikzpicture}
\end{equation*}
An execution graph for this outcome must have $\lRF$ and $\lMO$-edges as depicted above.
Since $\lMO$ goes from $\evlab{\lW}{}{\loc}{0}$ to $\evlab{\lW}{}{\loc}{1}$,
and $\evlab{\lR}{}{\loc}{0}$ reads from $\evlab{\lW}{}{\loc}{0}$,
we have an $\lFR$ edge from $\evlab{\lR}{}{\loc}{0}$ to $\evlab{\lW}{}{\loc}{1}$.
Due to the $\lHBRA$ from $\evlab{\lW}{}{\loc}{1}$ to $\evlab{\lR}{}{\loc}{0}$,
we obtain a $\lRA_{\lLOC}$-cycle, rendering this graph $\G_\RA$-inconsistent.
\end{example}

\begin{example}
\label{ex:2RMW-dec}
Similarly, the annotated outcome of \ref{prog:2RMW} from \cref{ex:2RMW}
is disallowed by $\G_\RA$ because the following execution graph is 
$\G_\RA$-inconsistent for any choice of $\lMO$:
\begin{equation*}
\small
\begin{tikzpicture}[yscale=0.8,xscale=1]
  \node (0x)  at (-2,0.5) {$\evlab{\lW}{}{\loc}{0}$};
  \node (11)  at (0,1) {$\evulab{}{\loc}{0}{1}$ };
  \node (21)  at (0,0) {$\evulab{}{\loc}{0}{1}$ };
  \draw[rf] (0x) edge node[above left=-2pt] {\smaller$\lRF$} (11);
  \draw[rf] (0x) edge node[below left=-2pt] {\smaller$\lRF$} (21);
\end{tikzpicture}
\end{equation*}
To see this, note that in $\G_\RA$-consistent executions,
 $\lMO$ cannot contradict $\lPO$.
 Hence, we must have $\lMO$ from the initial write to the two RMWs.
 This implies an $\lFR$ edge in both directions between the two RMWs,
 so that $\lRA_{\lLOC}$ must be cyclic.
\end{example}

Equivalence to the operational RA model for \emph{finite} behaviors follows from \cite{Kang-al:POPL17}:

\begin{theorem}
\label{thm:ra-equiv_basic}
$\Behfin{\M_\RA}=\Behfin{\G_\RA}$.
\end{theorem}

\subsection{A Declarative Memory System for StrongCOH}

The declarative model for \SCOH is obtained by requiring ``SC per-location'' and irreflexivity of \RA's happens-before, $(G.\lPO \cup G.\lRF)^+$:
$$\G_\SCOH \defeq \set{G\in\Exec \st G.\lHB_\RA \text{ and } G.\lSC_\lLOC \text{\ are irreflexive}}$$

Similarly to \RA, equivalence to the operational \SCOH model for \emph{finite} behaviors follows from the results of \citet{Kang-al:POPL17}:
\begin{theorem}
\label{thm:scoh-equiv_basic}
$\Behfin{\M_\SCOH}=\Behfin{\G_\SCOH}$.
\end{theorem}

%% file: fair_dec_sem.tex
\section{Making Declarative Semantics Fair}
\label{sec:fair-dec-sem}

\begin{toappendix}
\subsection{Proofs for \Cref{sec:fair-dec-sem}}
\end{toappendix}

In this section, we introduce memory fairness
into declarative models in a model-agnostic fashion.

To define fairness of execution graphs, we require that the partial ordering of events
in the graph is, like the ordering of natural numbers, \emph{prefix-finite}.
From an operational point of view,
an event preceded by an infinite number of events is never executed.

\begin{definition}
 \label{def:prefix_finite}
A relation $R$ on a set $A$ is \emph{prefix-finite}
if $\set{a\st \tup{a,b}\in R}$ is finite for every $b\in A$.
\end{definition}

Concretely, we require the modification order and the from-read relation
to be prefix-finite.%
\footnote{Note that the \emph{program order} and the \emph{reads-from} relation are prefix-finite
in a well-formed execution graph.
The former--by construction, the latter--since its reverse relation is functional.}

\begin{definition}
\label{def:fair_dec}
An execution graph $G$ is \emph{fair} if $G.\lMO$ and $G.\lFR$ are prefix-finite.
We denote by $\G^\fair$ the set of all fair execution graphs,
and let $\G_X^\fair \defeq \G_X\cap\G^\fair$ for $X \in \set{\SC,\TSO,\RA,\SCOH}$.
\end{definition}

\begin{example}
The following program illustrates our definition of fairness:
\begin{equation*}
  \tag{SCDeclUnfair} \label{ex:sc_decl_unfair}
\inarrII
{ \phantom{L_1\colon} \assignInst{\loc}{1} \,; \\
  L_1\colon           \assignInst{a}{\loc} \comment{\text{only} 1} \\
  \phantom{L_1\colon} \gotoInst{L1} }
{ L_2\colon \assignInst{\loc}{2} \,; \\ \phantom{L_2\colon} \gotoInst{L_2} }
\end{equation*}
Thread-fair executions of this program cannot produce the annotated outcome
with the \SC memory system.
With the declarative \SC memory system, however, there are two ways
in which every read can read from the write of $1$.

\newcommand{\unVertices}{
 
  \node[label=left:{Thread 1:\quad}] (W1)   at (0, 1) {$\evlab{\lW}{}{x}{1}$};
  \node (R11)  at (1, 1) {$\evlab{\lR}{}{x}{1}$};
  \node (R12)  at (2, 1) {$\evlab{\lR}{}{x}{1}$};
  \node (inf1) at (2.5, 1) {$\cdots$};

  \node[label=left:{Thread 2:\quad}] (W21)  at (0, 0) {$\evlab{\lW}{}{x}{2}$};
  \node (W22)  at (1, 0) {$\evlab{\lW}{}{x}{2}$};
  \node (W23)  at (2, 0) {$\evlab{\lW}{}{x}{2}$};
  \node (inf2) at (2.5, 0) {$\cdots$};
}

  First, the write of $1$ to $x$ may have infinitely many $\lMO$-predecessors,
  as illustrated below.
  \begin{center}
    \begin{tikzpicture}[yscale=1,xscale=2]
      \unVertices
      \dpo{W21}{W22}; \dpo{W22}{W23};
      \dmo[transform canvas={yshift=-0.5ex}][below]{W21}{W22}; \dmo[transform canvas={yshift=-0.5ex}][below]{W22}{W23};
      \dmo[][left]{W21}{W1}; \dmo[][below left=-2pt]{W22}{W1}; \dmo[][above right=-2pt]{W23}{W1};
      \dpo{W1}{R11}; \dpo{R11}{R12};
      \drf[transform canvas={yshift=0.5ex}]{W1}{R11}; \drf[bend left=30][below,yshift=0.6ex,opacity=0.9]{W1}{R12};
    \end{tikzpicture}
  \end{center}

  Otherwise, the write of $1$ may have finitely many $\lMO$-predecessors but infinitely many $\lMO$-successors.
  Then, each of the $\lMO$-successors will have infinitely many $\lFR$-predecessors.
  \begin{center}
    \begin{tikzpicture}[yscale=1,xscale=2]
      \unVertices
      \dpo{W21}{W22};
      \dpo{W22}{W23};
      \dmo[][left]{W21}{W1}; \dmo[][below]{W1}{W22}; \dmo[transform canvas={yshift=-0.5ex}][below]{W22}{W23};
      \dpo{W1}{R11}; \dpo{R11}{R12};
      \drf[transform canvas={yshift=0.5ex}]{W1}{R11}; \drf[bend left=30][below,yshift=0.6ex,opacity=0.9]{W1}{R12};
      \dfr[left]{R11}{W22}; \dfr{R12}{W22};
      \dfr{R11}{W23}; \dfr[right]{R12}{W23};
    \end{tikzpicture}
  \end{center}

In both cases, the execution graph is unfair.
(As we prove below, this is not a coincidence.)
\end{example}

\begin{example}
On the converse, one should avoid unnecessary prefix-finiteness constraints.
In particular, requiring prefix-finiteness of cyclic relations, such as 
$[G.\lE \setminus \Init] \seq \lHBSC$ under \TSO, \RA, or \SCOH, is too strong.
Doing so would forbid the annotated behavior of the following example,
since the corresponding execution graph contains
an infinite $\lPO\cup\lFR$ descending chain.
Yet, \TSO, \RA, and \SCOH allow the annotated behavior,
as every write may be delayed past 1 or 2 reads.
\begin{gather*}
  \tag{HbAcyclic}\label{ex:hb_acyclic}
  \inarrII
  {L_1\colon           \assignInst{k}{k+1} \\
   \phantom{L_1\colon} \assignInst{x}{k} \\
   \phantom{L_1\colon} \assignInst{a}{y} \comment{0,0,1,2\ldots} \\
   \phantom{L_1\colon} \gotoInst{L_1} 
  }{L_2\colon          \assignInst{m}{m+1} \\
   \phantom{L_2\colon} \assignInst{y}{m} \\
   \phantom{L_2\colon} \assignInst{b}{x} \comment{0,1,2,3\ldots}  \\
   \phantom{L_2\colon} \gotoInst{L_2} 
  }
\end{gather*}
\begin{gather*}
\begin{tikzpicture}[yscale=0.8,xscale=1.4, every node/.style={inner sep=2pt}]
  \node[label=left:{Thread 1:\quad}]  (W11)  at (0, 1) {$\evlab{\lW}{}{x}{1}$};
  \node (R11)  at (1, 1) {$\evlab{\lR}{}{y}{0}$};
  \node (W12)  at (2, 1) {$\evlab{\lW}{}{x}{2}$};
  \node (R12)  at (3, 1) {$\evlab{\lR}{}{y}{0}$};
  \node (W13)  at (4, 1) {$\evlab{\lW}{}{x}{3}$};
  \node (R13)  at (5, 1) {$\evlab{\lR}{}{y}{1}$};
  \node[label=left:{Thread 2:\quad}]  (W21)  at (0, -0.2) {$\evlab{\lW}{}{y}{1}$};
  \node (R21)  at (1, -0.2) {$\evlab{\lR}{}{x}{0}$};
  \node (W22)  at (2, -0.2) {$\evlab{\lW}{}{y}{2}$};
  \node (R22)  at (3, -0.2) {$\evlab{\lR}{}{x}{1}$};
  \node (W23)  at (4, -0.2) {$\evlab{\lW}{}{y}{3}$};
  \node (R23)  at (5, -0.2) {$\evlab{\lR}{}{x}{2}$};
  \node (inf) at ($(R13)!0.5!(R23)+(0.5,0)$) {\huge\ldots};
  \dpo{W11}{R11}; \dpo{R11}{W12}; \dpo{W12}{R12}; \dpo{R12}{W13}; \dpo{W13}{R13}; 
  \dpo{W21}{R21}; \dpo{R21}{W22}; \dpo{W22}{R22}; \dpo{R22}{W23}; \dpo{W23}{R23}; 
  \draw[fr] (R12) to [in=45,out=225,looseness=0.3] (W21);
  \draw[fr] (R13) to [in=45,out=225,looseness=0.3] (W22);
  \draw[fr] (R21) to [out=135,in=315,looseness=0.3] (W11);
  \draw[fr] (R22) to [out=135,in=315,looseness=0.3] (W12);
  \draw[fr] (R23) to [out=135,in=315,looseness=0.3] (W13);
\end{tikzpicture}
\end{gather*}
\end{example}

Our main result extends
\cref{thm:sc-equiv_basic,thm:tso-equiv_basic,thm:ra-equiv_basic,thm:scoh-equiv_basic}
for \emph{infinite} traces by imposing memory fairness
on the operational systems (\cref{def:memory_fair})
and execution graph fairness on the
declarative systems (\cref{def:fair_dec}).

\begin{theorem}\label{thm:equiv_proof}
For $X \in \set{\SC,\TSO,\RA,\SCOH}$,
$$\Behmf{\M_X}=\Beh{\G_X^\fair}.$$
\end{theorem}

As a corollary, it easily follows from our definitions that the set
of (thread\&) memory-fair behaviors of a program $\prog$ under 
$\M_X$ coincides with the set of (thread\&) memory-fair behaviors of a program $\prog$ under $\G_X^\fair$.

The full proof of \cref{thm:equiv_proof} is included in appendix (\cite{appendix})
and its Coq mechanization in \cite{artifact}.
Here, we outline the proof
starting with the easier direction.%

\subsection{$\Behmf{\M_X} \suq \Beh{\G_X^\fair}$}

Given a memory-fair behavior $\beh$ of $\M_X$,
we let $\trace$ be a memory-fair observable trace of $\M_X$
such that $\beh(\trace)=\beh$.
Then, using $\trace$, we construct a fair execution graph $G\in \G_X$.
Its events are determined by $\beh$ ($G.\lE = \Event(\beh)$),
and its relations are defined differently for every system:

\paragraph{SC}

The $\lRF$ and $\lMO$ relations are determined by the trace order:
for each read $\lRF$ assigns the latest write of the same location,
while $\lMO$ corresponds to the trace order restricted to writes to the same location.
It follows that $\lFR$ is included in the trace order,
and since the trace order is prefix-finite, $\lMO$ and $\lFR$ are prefix-finite as well.

\paragraph{TSO}

We define $\lMO$ to be the order in which writes to the same location are propagated (unbuffered) to memory.
For each read, $\lRF$ maps it either to the $\lMO$-maximal write to the same location
that was propagated before it in $\trace$ (if the read reads from memory)
or to the $\lPO$-maximal one by the same thread (if it reads from the buffer).
Since every write is eventually propagated to memory,
and once propagated no thread can read from an $\lMO$-prior write,
it follows that both $\lMO$ and $\lFR$ are prefix-finite.

\paragraph{RA and StrongCOH}

The $\lMO$ component of $G$ follows
the order induced by timestamps of messages in the operational run.
Prefix-finiteness of $\lMO$ follows from the facts that
a location and a timestamp uniquely identify the corresponding message
(and the write event in $G$ respectively)
and that timestamps are natural numbers---that is,
each write event $w$ representing a message with a timestamp $\ts$
has at most $\ts$ $\lMO$-prior writes.

The $\lRF$ component of $G$ connects
an event related to a read/RMW transition of $\trace$ with
a write event representing the message read by the transition.

Prefix-finiteness of $\lFR$ follows from the fact that in the fair operational run
every message is eventually propagated to every thread.
That is, for any given write event $w$ to a location $\loc$ in $G$ representing a message with a timestamp $\ts$,
there cannot be infinitely many reads from $\loc$ in $G$ reading from
write events that correspond to messages with timestamps smaller than $\ts$.

\subsection{$\Beh{\G_X^\fair} \suq \Behmf{\M_X}$}

The converse direction is more challenging.
Given a fair $\G_X$-consistent execution graph $G$,
we have to find a memory-fair observable trace $\trace$ of $\M_X$
such that $\Event(\beh(\trace))=G.\lE$.

Put differently, we need a total order over $G.\lE\setminus\Init$ that extends $G.\lPO$,
so that some memory-fair run of $\M_X$ executes according to this order.
Existing proofs of correspondence between declarative and operational definitions
of SC, RA, and StrongCOH pick an arbitrary total order extending $G.\lHBSC$ (for SC) and $G.\lHBRA$ (for RA and StrongCOH).
(Assuming the axiom of choice, any partial order $R$ on a set $A$ can be
extended to a total order on $A$.)
It is then not difficult to show that executing the program following that order yields the labels appearing in the execution graph.
For infinite graphs, however, an arbitrary extension of $G.\lHBSC$ (or $G.\lHBRA$ respectively)
does not necessarily correspond to a (memory-fair) run of the program.
For this, we need an \emph{enumeration} of $G.\lE\setminus\Init$, as defined next.

\begin{definition}
An \emph{enumeration} of a set $A$ is a (finite or infinite) injective (\ie without repetitions) sequence $\en$
covering all the elements in $A$ (\ie $A=\set{\en(i) \st i\in\dom{\en}}$).
An enumeration $\en$ of $A$ \emph{respects} a partial order $R$ on $A$
if $i<j$ whenever $\tup{\en(i),\en(j)} \in R$.
\end{definition}

Prefix-finiteness of a partial order ensures that a suitable enumeration exists
(our proof employs classical, non-constructive, reasoning):

\begin{proposition}\label{prop:respect-enum}
Let $R$ be a prefix-finite partial order on a countable set $A$.
Then, there exists an enumeration of $A$ that respects $R$.
\end{proposition}

\begin{proofof}{prop:respect-enum}
Let $\nu$ be an enumeration on $A$.
For $i\in\N$, let $S^i = \set{\nu(i)} \cup \set{a \st \tup{a,\nu(i)} \in R}$
be the set containing $\nu(i)$ and all its $R$-predecessors.
Note that this set is finite because $R$ is prefix-finite.

Now let $T^i = S^i \setminus \bigcup \set{ S^j \st j < i}$ be the set of new
elements in $S^i$ (\ie those not included in the $S$-set for previous indices).
Since $T^i$ is a finite set, it can be topologically ordered with respect to $R$
yielding a sequence $t^i_1 \ldots t^i_{|T^i|}$ respecting $R$.

We will now define an enumeration $\mu$ on $A$ that respects $R$.
Let $\mu(i)$ to be the $i^{\textrm{th}}$ element of the sequence
$$t^1_1 \, t^1_2 \ldots t^1_{|T^1|} \,
 t^2_1 \, t^2_2 \ldots t^2_{|T^2|} \ldots
 t^i_1 \, t^i_2 \ldots t^i_{|T^i|}.$$
Note that this sequence has at least $i$ elements, so $\mu(i)$ is well defined.
Moreover, by construction, this sequence has no duplicates, respects $R$,
and the sequence for $i$ is a prefix of the sequence for $i+1$.
From these facts, it follows that $\mu$ is an enumeration of $A$ that respects $R$.
\end{proofof}

However, we do not yet have that the ``happens-before'' relation of each model is prefix-finite;
we only know that $G.\lMO$ and $G.\lFR$ are prefix-finite.
Next, we show that prefix-finiteness of $G.\lMO$ and $G.\lFR$
suffices for prefix-finiteness of the other relations,
as long as the program in question has a bounded number of threads.
(Recall that we assume that the set $\Tid$ is finite.)

First, note that every relation on a finite set is prefix-finite,
and prefix-finiteness is preserved by (finite) composition.
\begin{lemma} \label{lem:compose-pref-fin}
Let $R$ and $R'$ be prefix-finite relations and $n \in \N$.
Then $R \cup R'$,  $R \seq R'$ and $R^{\leq n}$ are also prefix-finite.
\end{lemma}

For transitive closures, we need an auxiliary property.

\begin{definition}
\label{def:n_totality}
A relation $R$ on a set $A$ is \emph{$n$-total} if for every $n+1$ distinct
elements $a_1 \til a_{n+1} \in A$, we have $\tup{a_i,a_j}\in R$ for some $1 \leq i,j \leq n+1$.
\end{definition}

For an execution graph $G$ with $n$ threads, 
$G.\lPO$ is $n$-total
(as a relation on $G.\lE \setminus \Init$).
By the pigeonhole principle, any set of $n+1$ events in $G.\lE \setminus \Init$ contain two elements
belonging to the same thread, and those two events are ordered by $G.\lPO$.

Now, if a relation $R$ is $n$-total and acyclic, its transitive closure $R^{+}$ has
bounded length, which entails that $R^+$ is prefix-finite provided $R$ is
prefix-finite.

\begin{lemma} \label{lem:ct-pref-fin}
Let $R$ be an acyclic, $n$-total, prefix-finite relation.
Then, $R^+$ is prefix-finite.
\end{lemma}
\begin{proofof}{lem:ct-pref-fin}
It suffices to show that $R^{2n+1} \suq R^{\leq 2n}$.
The reason is that then $R^+\suq R^{\leq 2n}$,
and so by \cref{lem:compose-pref-fin}, $R^+$ is prefix-finite.
Consider therefore a set of $2n+1$ distinct elements $a_1 \til a_{2n+1} \in A$
such that $\tup{a_k,a_{k+1}}\in R$ for all $1 \leq k \leq 2n$.
Consider the set of $n+1$ elements $a_1,a_3,a_{2n+1}$.
By $n$-totality, there exist $1 \leq i,j \leq n+1$ such that $\tup{a_{2i+1},a_{2j+1}}\in R$.
Since $R$ is acyclic, it follows that $i < j$,
and so there is a shorter path from $a_1$ to $a_{2n+1}$,
\ie $\tup{a_1,a_{2n+1}}\in R^{\leq 2n}$, as required.
(In more detail, we have $\tup{a_1,a_{2i+1}}\in R^{2i+1}$,
$\tup{a_{2i+1},a_{2j+1}}\in R$,
$\tup{a_{2j+1},a_{2n+1}}\in R^{2n-2j}$,
and so $\tup{a_1,a_{2n+1}}\in R^{2i+1+1+2n-2j}=R^{2(n+i+1-j)} \suq R^{\leq 2n}$.)
\end{proofof}

As a corollary, we obtain that the prefix-finiteness of the ``happens-before''
relation in fair execution  graphs.

\begin{corollary}\label{cor:hb-pf}
For $X \in \set{\SC,\TSO,\RA,\SCOH}$, let $G$ be a fair $\G_X$-consistent execution graph.
Then $[G.\lE \setminus \Init] \seq G.\lHB_X$ is prefix finite.%
\footnote{We define $G.\lHB_\SCOH$ to be equal to $G.\lHB_\RA$.}
\end{corollary}

\begin{proofof}{cor:hb-pf}
In all cases, by applying \cref{lem:ct-pref-fin},
it suffices to show that the non-transitive versions of these relations
are acyclic, prefix-finite, and $n$-total.
Acyclicity follows by the consistency conditions,
prefix-finiteness follows by \cref{lem:compose-pref-fin} as they are unions of prefix-finite relations.

$n$-totality follows immediately for SC, RA, and StrongCOH, because $\lPO$ is $n$-total where $n$ is the number of threads.

In the case of TSO, $\lPO \setminus (\lWs \times \lRs)$ is $2n$-total (where $n$ is the number of threads)
because for every three events in the same thread at least two of them are related by $\lPO \setminus (\lWs \times \lRs)$.
Namely, let the three events be $a$, $b$, and $c$, and without loss of generality
assume that $\tup{a,b}\in\lPO$ and $\tup{b,c}\in\lPO$.
If $\tup{a,b}\notin \lWs \times \lRs$, then $\tup{a, b}$ satisfies the required condition. 
Otherwise, depending on the type of $c$, either $\tup{b, c}$ or $\tup{a, c}$ satisfy it.
\end{proofof}

From \cref{prop:respect-enum}, there is an enumeration $\nu$ that respects $\lHB_X$.
We use $\nu$ to construct a program trace $\trace$:

\paragraph{SC}
The trace $\trace$ follows $\nu$ exactly.
Since $\M_\SC$ has no silent memory transitions, $\trace$ is trivially memory fair. 
    
\paragraph{TSO}
The trace $\trace$ is incrementally constructed by following the order of events
in $\nu$ and appending an appropriate sequence of transitions.
If the next event in $\nu$ is a read,
we append to $\trace$ all unexecuted $\lPO$-prior writes and then the read.
If the next event in $\nu$ is a write, we append the write to the trace
if it has not already been included in the trace.
In addition, when the next event is a write, we append its propagation action to memory.
By construction, every write in $\trace$ is eventually propagated to memory.

\paragraph{RA and StrongCOH}
The trace $\trace$ is the enumeration $\nu$ interleaved with silent \RA/\SCOH transition labels.
Namely, for each write $w$ and thread $\tid$, we compute an index $i$ in the enumeration 
such that it is \emph{safe} to propagate $w$ to $\tid$ at that index:
for each event in $\tid$ with index greater than $i$, there is no $\lX$-following (where $\lX=\lHB_\RA$ for \RA and $\lX=\lRF^?\mathbin{;}\lPO^?$ for \SCOH)
(i) write that $\lMO$-precedes $w$
and (ii) read that reads from a write $\lMO$-preceding $w$.
Since $G$ is fair, such an index is defined for all (non-initialization) writes. 
Then, after the event with an index corresponding to some write has been enumerated, we execute a propagation transition for the write.
In that way, every write is eventually propagated to every thread, so the resulting trace is memory fair.

\begin{remark}
\label{rem:inf}
\Cref{cor:hb-pf} relies on having a bounded number of threads.
With infinite number of threads, generated, \eg by thread spawning,
prefix-finiteness of $\lMO$ and $\lFR$ is not enough to rule out unfair behaviors.
To see this, consider the annotated behavior of the following program
and the corresponding execution graph:
\begin{gather*}
    {\inarr{
      L\colon \assignInst{i}{i+1} \\
      \phantom{L\colon} \kw{spawn} \left\{\;\inarr{
      x_{i+1} := 1 \\
      a := x_{i} \comment{only 0} \\
      }\;\right\} \\
      \phantom{L\colon} \gotoInst{L}
    }}
\qquad
\inarr{\begin{tikzpicture}[yscale=0.8,xscale=0.8]
  \node (W1)  at (-2, 2) {$\evlab{\lW}{}{x_2}{1}$};
  \node (R1)  at (-2, 1) {$\evlab{\lR}{}{x_1}{0}$};
  \node (W2)  at (0, 2) {$\evlab{\lW}{}{x_3}{1}$};
  \node (R2)  at (0, 1) {$\evlab{\lR}{}{x_2}{0}$};
  \node (W3)  at (2, 2) {$\evlab{\lW}{}{x_4}{1}$};
  \node (R3)  at (2, 1) {$\evlab{\lR}{}{x_3}{0}$};
  \node (inf) at (5, 1.5) {\Huge $\cdot\! \cdot\! \cdot$};
  \node (R4)  at (4, 1) {$\evlab{\lR}{}{x_4}{0}$};
  \draw[po] (W1) edge (R1); \draw[po] (W2) edge (R2); \draw[po] (W3) edge (R3);
  \draw[fr] (R2) edge node[right] {\smaller$\lFR$} (W1);
  \draw[fr] (R3) edge node[right] {\smaller$\lFR$} (W2);
  \draw[fr] (R4) edge node[right] {\smaller$\lFR$} (W3);
\end{tikzpicture}}
\end{gather*}
While $\lMO$ and $\lFR$ are trivially prefix-finite,
$\lHBSC$ has an infinite descending chain, and indeed there is
no \SC execution of the program leading to the annotated behavior
(where spawn adds a thread to the current pool,
and a thread from the pool is non-deterministically chosen at each step).
\end{remark}

%% file: rc11.tex
\subsection{Making RC11 Fair}
\label{sec:rc11}

Having established evidence for the adequacy of the declarative fairness condition,
we may apply this condition in other (and richer) declarative models.
In particular, we propose to adopt this condition into the C/C++ memory model.
Next, we discuss this proposal in the context of the RC11 model~\cite{scfix},
a repaired version of the C/C++11 specification~\cite{Batty-al:POPL12}
that fixes certain issues involving sequentially consistent accesses
and works around the ``thin-air'' problem by completely forbidding $\lPO\cup\lRF$ cycles.
A full definition of RC11 is obtained by carefully
combining the key concepts of \SC, \RA, and \SCOH.
It requires us to include  in the declarative framework
access modes (\aka ``memory orderings''---the consistency level required from every memory access),
and several types of fences.
For simplicity, we elide these definitions and keep the discussion more abstract.
Indeed, there is nothing special about RC11 in this context---the declarative fairness
condition could be added to any model requiring  $\lPO\cup\lRF$ acyclicity.

Generally speaking, when proposing a strengthening of a programming language memory model,
one has to make sure that the mapping schemes to multicore architectures are not broken,
and that source-to-source compiler transformations are still validated.
In our case, the mapping of RC11 to x86-TSO trivially remains sound.
Indeed, as we saw in~\cref{thm:equiv_proof}, the natural operational characterization of liveness in TSO
corresponds to the declarative condition requiring that the $\lMO$ and $\lFR$ relations
are prefix-finite. Since the same condition is applied both in the source level (RC11)
and in the target level (x86-TSO), and mappings of source graphs to target ones
keep $\lMO$ and $\lFR$ intact, we maintain the soundness of the known mappings.\footnote{See
\url{http://www.cl.cam.ac.uk/~pes20/cpp/cpp0xmappings.html} [accessed July-2021].}
We note that for establishing the soundness of the mappings to other architectures, one first needs a formal fairness
condition of the architecture. While this may be more difficult in architectures
weaker than x86-TSO (see \cref{sec:discussion}), it is likely that no hardware
will allow that a write is placed after infinitely many other writes in the coherence order
(non-prefix-finite $\lMO$), or that infinitely many reads do not observe a later
write (non-prefix-finite $\lFR$).

Considering compiler transformations, one has to show that
every behavior of the target program explained by a consistent graph $G_{\text{tgt}}$
is also obtained by a consistent graph $G_{\text{src}}$ of the source program.
It is not hard to see that the constructions of \citet{c11comp} and \citet{scfix}
work as-is for the RC11 model strengthened with fairness.
First, the constructions of $G_{\text{src}}$ for \emph{reordering transformations},
which reorder two memory accesses under certain conditions,
keep the same $\lMO$ and $\lFR$ relations of $G_{\text{tgt}}$;
so their prefix-finiteness trivially follows.

Second, we consider \emph{elimination transformations}
that eliminate a redundant memory access.
In this case,
$G_{\text{src}}$ is obtained from $G_{\text{tgt}}$ by adding one additional event
$e_{\text{new}}$ that corresponds to the eliminated instruction.
For read elimination (read-after-read or read-after-write),
$e_{\text{new}}$ is a read event, and the construction of
ensures that $G_{\text{src}}.\lMO=G_{\text{tgt}}.\lMO$.
In turn, $G_{\text{tgt}}.\lFR \suq G_{\text{src}}.\lFR$, but since only one event
is added to $G_{\text{src}}$, prefix-finiteness of $\lFR$ is again trivially preserved.

Finally, we consider write-after-write elimination.
Let $w_0$ denote the immediate $G_{\text{src}}.\lPO$-successor of $e_{\text{new}}$.
Then, to construct $G_{\text{src}}.\lMO$, one places $e_{\text{new}}$ as the
immediate predecessor of $w_0$.
Then, consistency of $G_{\text{src}}$ follows the argument of \cite{scfix},
and it remains to show that fairness of $G_{\text{src}}$ follows from the fairness of $G_{\text{tgt}}$.
The latter is easy: write events in $G_{\text{src}}$ other than $e_{\text{new}}$
all have at most one more incoming $G_{\text{src}}.\lMO$ edge (from $e_{\text{new}}$),
and the same set of  incoming $G_{\text{src}}.\lFR$ edges.
In turn, For $e_{\text{new}}$ itself, we have:
$\set{e \in G_{\text{src}}.\lE \st \tup{e,e_{\text{new}}}\in G_{\text{src}}.\lMO \cup G_{\text{src}}.\lFR} \suq
\set{e \in G_{\text{tgt}}.\lE \st \tup{e,w_0}\in G_{\text{tgt}}.\lMO \cup G_{\text{tgt}}.\lFR}.$

In the next section, we demonstrate that adding fairness to RC11 as proposed above
provides the necessary underpinnings allowing one to formally reason about termination under RC11.

%% file: robustness.tex
\subsection{From Finite to Infinite Robustness}
\label{sec:robustness}

Common advice given to programmers of multi-threaded software
is to follow a programming discipline that hides the effects of that weak memory model,
\eg to use exclusively sequentially consistent accesses.
Programs that follow such a discipline are \emph{robust}, 
meaning that they have only sequentially consistent behaviors on the underlying weak memory model.
While there is a rich literature on programming disciplines that imply robustness
and verification techniques for robustness~ \cite{tso-robustness2,tso-robustness,Derevenetc:2014,Bouajjani18,Margalit21,Lahav:pldi19,jonastso}, most work only considers \emph{finite} behaviors, i.e., they leave open whether programs following the discipline have only sequentially consistent \emph{infinite} behaviors.
This means that any correctness properties that only concern infinite behaviors, such as starvation-freedom, might be lost on the weak memory model despite its (finite) robustness.
In this section, we show that this cannot happen as long as the weak memory model satisfies our declarative memory fairness condition and its consistency predicate is $\lPO\cup\lRF$-prefix closed.
This is the case for all models studied in this paper.
Thus our unified definition of memory fairness lifts all existing robustness results for these models from the literature to infinite behaviors.

First, we observe that the consistency predicates based on acyclicity (\SC-consisterncy, in particular)
enjoy a ``compactness property''---if they hold for all finite prefixes of a graph, then they also hold for the full graph.
Below, by \emph{finite} execution graph, we mean a graph $G$ with $G.\lE\setminus \Init$ being finite
(the set $\Init$ of initialization events may be infinite if $\Loc$ is infinite).

\begin{definition}
An execution graph $G'$ is a \emph{$\lPO\cup\lRF$-prefix} of an execution graph $G$
if we have $\dom{(G.\lPO \cup G.\lRF) \seq [G'.\lE]} \suq G'.\lE$,
$G'.\lRF = [G'.\lE] \seq G.\lRF \seq [G'.\lE]$,
and $G'.\lMO = [G'.\lE] \seq G.\lMO \seq [G'.\lE]$.
\end{definition}

\begin{proposition}[$\G_\SC$ compactness]
\label{prop:SC_compact}
Let $G$ be an execution graph with prefix-finite $([G.\lE \setminus \Init] \seq G.\lPO \cup G.\lRF)^+$.
If every finite $\lPO\cup\lRF$-prefix of $G$ is $\G_\SC$-consistent, then so is $G$.
\end{proposition}
\begin{proof}
Suppose that $G$ is $\G_\SC$-inconsistent,
and let $a_1\til a_n \in G.\lE$ such that 
$\tup{a_i,a_{i+1}} \in G.\lPO \cup G.\lRF \cup G.\lMO \cup G.\lFR$ for every $1\leq i \leq n-1$,
and $\tup{a_n,a_1} \in G.\lPO \cup G.\lRF \cup G.\lMO \cup G.\lFR$.
Let $E' = \Init  \cup \dom{(G.\lPO \cup G.\lRF)^* \seq [\set{a_1 \til a_n}]}$,
and let $G'= \tup{E', [E'] \seq G.\lRF \seq [E'], [E'] \seq G.\lMO \seq [E']}$.
Since $([G.\lE \setminus \Init] \seq G.\lPO \cup G.\lRF)^+$ is prefix-finite,
$G'$ is a finite  $\lPO\cup\lRF$-prefix of $G$.
However, we have 
$\tup{a_i,a_{i+1}} \in G'.\lPO \cup G'.\lRF \cup G'.\lMO \cup G'.\lFR$ for every $1\leq i \leq n-1$,
and $\tup{a_n,a_1} \in G'.\lPO \cup G'.\lRF \cup G'.\lMO \cup G'.\lFR$,
so $G'$ is $\G_\SC$-inconsistent.
\end{proof}

\begin{definition}[Robustness]
Let $\prog$ be a program and $\G$ be a declarative memory system.
\begin{itemize}
\item $\prog$ is \emph{finitely execution-graph robust} against $\G$
if for every finite behavior $\beh\in \Beh{\prog}$ and $G\in \G$ with $\Event(\beh)=G.\lE$,
we have $G \in \G_\SC$.
\item $\prog$ is \emph{strongly execution-graph robust} against $\G$
if for every (finite or infinite) behavior $\beh\in \Beh{\prog}$ and $G\in \G$ with $\Event(\beh)=G.\lE$,
we have $G \in \G_\SC$.
\end{itemize}
\end{definition}

\begin{theorem}
\label{prop:robustness}
Let $\G$ be a declarative memory system such that:
\begin{itemize}
\item $\G$-consistency is $\lPO\cup\lRF$-prefix closed (\ie if $G\in \G$ then $G'\in \G$ for every $\lPO\cup\lRF$-prefix $G'$ of $G$).
\item $G \in \G$ implies that $([G.\lE \setminus \Init] \seq G.\lPO \cup G.\lRF)^+$ is prefix-finite.
\end{itemize}
Then, if a program $\prog$ is finitely execution-graph robust against $\G$,
then it is also strongly execution-graph robust against $\G$.
\end{theorem}
\begin{proof}
Suppose that $\prog$ is finitely execution-graph robust against $\G$.
Let $G\in \G$ such that $G.\lE = \Event(\beh)$  for some behavior $\beh\in\Beh{\prog}$.
From finite execution-graph robustness, it follows that every finite $\lPO\cup\lRF$-prefix of $G$ is $\G_\SC$-consistent.
By \cref{prop:SC_compact}, $G$ is $\G_\SC$-consistent as well.
\end{proof}

We note that the declarative $\TSO$, $\RA$, \SCOH, and \RC models satisfy the premises
of \cref{prop:robustness}.
The Coq mechanization includes the formal proof of the statement below. 

\begin{corollary}
\label{cor:robustness}
Suppose that a program $\prog$ is finitely execution-graph robust against $\G_X$ for $X\in\set{\TSO,\RA,\SCOH,\RC}$.
Then, the set of (thread\&) memory-fair behaviors of $\prog$ under $\M_X$ 
coincides with the set of (thread\&) memory-fair behaviors of $\prog$ under $\M_\SC$.
\end{corollary}
\begin{proof}
One direction is obvious since $\M_\SC$ is stronger than $\M_X$.
For the converse, let $\beh$ be a  memory-fair behavior of $\prog$ under $\M_X$.
Then, by \cref{thm:equiv_proof}, we have that 
$\beh$ be a memory-fair behavior of $\prog$ under $\G_X^\fair$.
By definition, we have that $\beh\in\Beh{\prog}\cap\Beh{\G_X^\fair}$.
Let $G\in\G_X$ such that $\Event(\beh)=G.\lE$.
Then, since $G\in \G_X$, by \cref{prop:robustness}, we have that $G \in \G_\SC$.
Since the declarative fairness condition is the same in all four models, we have 
$G \in \G_\SC^\fair$.
Hence, we have $\beh\in\Beh{\prog}\cap\Beh{\G_\SC^\fair}$,
and so by \cref{thm:equiv_proof}, it follows that 
 $\beh$ be a memory-fair behavior of $\prog$ under $\M_\SC$.
Finally, to deal with thread fairness, one has to use $\Behtf{\prog}$ instead of $\Beh{\prog}$ in this argument.
\end{proof}

As a simple application example, the \ref{ex:spinlock-client} program in \cref{sec:spinlock} below is
(finitely) execution-graph robust because the program employs only a single location  (the location $l$ for the lock implementation).
Then, \cref{cor:robustness} entails that this program may diverge under the weak memory models studied in this paper
iff it diverges under $\SC$, and that the same also holds when assuming thread fairness.

%% file: lockexamples.tex
\section{Proving Deadlock Freedom for Locks}
\label{sec:lockexamples}

In this section, we prove the termination and/or fairness of spinlock, ticket lock, and MCS lock clients.
The key to doing so is \cref{thm:spinloop termination basic} below,
which reduces proving termination of spinloops under fair weak memory models
to reasoning about a single specific iteration of the loop.

\input{spinloops.tex}

\input{spinlock.tex}

\input{ticketlock.tex}

\input{hmcslock.tex}

%% file: spinloops.tex
For simplicity, we henceforth assume that the sequential programs composing the concurrent programs are 
deterministic, as defined below. (The thread interleaving itself still makes the 
concurrent program semantics non-deterministic.)

\begin{definition}
A program $\prog$ is \emph{deterministic} if
$\progstate \astep{\tstep{\tid}{\lab_1}}_\prog \progstate_1$
and  $\progstate \astep{\tstep{\tid}{\lab_2}}_\prog \progstate_2$
imply that $\lTYP(\lab_1) = \lTYP(\lab_2)$ and $\lLOC(\lab_1)=\lLOC(\lab_2)$,
and, moreover, if $\lab_1 = \lab_2$, then $\progstate_1 = \progstate_2$ also holds.
\end{definition}

For a behavior $\beh$ of a deterministic program $\prog$
and $\tid\in\Tid$,
we denote by $\run_\tid(\beh)$ the unique run of $\prog(\tid)$
that induces the sequential trace $\beh(\tid)$.

\begin{definition}\label{def:spinloop iteration}
A \emph{spinloop iteration} of thread $\tid$ in a behavior $\beh$
is a range of event serial numbers $[n,n']$ such that
the sequence of corresponding program steps:
\begin{enumerate}
\item performs only reads:
$\lTYP(\lTLAB(\run_\tid(\beh)(i))) = \lR$ for $n \leq i \leq n'$; and
\item returns the program to the starting state of the loop:
$\lSRC(\run_\tid(\beh)(n)) = \lTGT(\run_\tid(\beh)(n'))$.
\end{enumerate}
An \emph{infinite spinloop} of thread $\tid$ in a behavior $\beh$ is an infinite sequence $s$ of consecutive spinloop iterations of thread $\tid$ (\ie $s(i)=[n_i,n_i'] \implies \exists n_{i+1}'.\, s(i+1)=[n_i',n_{i+1}']$).
\end{definition}

If infinite spinloops are the only source of unbounded behavior in programs
(\ie their individual iterations are of bounded length
and there are boundedly many writes to each memory location),
then because of fairness, an infinite spinloop has to eventually read from the
$\lMO$-maximal writes.

\begin{theorem} \label{thm:spinloop termination basic}
Let $\beh$ be a behavior of a deterministic program
and $G$ be a fair execution graph with $G.\lE = \Event(\beh)$
and $G.\lSC_{\lLOC}$ (see \cref{sec:declarative TSO memory system}) irreflexive.
For every infinite spinloop $s$ of a thread $\tid$ in $\beh$
whose iterations have bounded length
and read only from locations that are written to by finitely many writes in $G$,
there is a loop iteration $s(i)$ whose reads all read from $\lMO$-maximal writes.
\end{theorem}
\begin{proofof}{thm:spinloop termination basic}
Assume for the sake of contradiction that
$N>0$ is the smallest number for which no iteration $s(i)$ reads from $\lMO$-maximal stores in its first $N$ steps.
Consider a spinloop iteration $s(i)$ in which the first $N-1$ steps read $\lMO$-maximal stores,
and let the location that its $N^{\mathrm{th}}$ step reads be $x$.

Because of determinism and $G.\lSC_{\lLOC}$ irreflexivity,
in all subsequent iterations $s(j)$ with $j>i$,
the first $N-1$ transitions are identical,
while the $N^{\mathrm{th}}$ transitions read location $x$.

Since by assumption the spinloop only reads from locations with a finite set of writes,
$x$ must have an $\lMO$-maximal write $w$.
By assumption, step $N$ of the spinloop never reads from $w$,
so all of these infinitely many reads must read from stores that are $\lMO$-ordered before $w$.
Thus we have infinitely many reads that are $\lFR$-ordered before $w$.
But because $G$ is fair, $\lFR$ is prefix-finite, which is a contradiction.
\end{proofof}

This theorem provides a sufficient condition for establishing termination of
spinloops. In the supplementary material, we also establish the other direction:
whenever a deterministic program has a behavior where all non-terminated
threads end with a loop iteration reading from $\lMO$-maximal writes,
then it has an infinite memory-fair behavior.

%% file: spinlock.tex
\subsection{Spinlock}
\label{sec:spinlock}

Consider the following spinlock implementation:
  $$\inarrT{
    \assignInst{\kw{int} \; l}{0} \\
    \funcDecl{\kw{void}}{lock}{}{
      \kw{int}\; r \\
      \repeatml{
        \repeatInst{ \readInst{r}{l} }{(r = 0)}
      }{(\casInstn(l, 0, 1))}
    }
    \funcDeclsl{\kw{void}}{unlock}{}{
      \writeInst{l}{0}
    }
  }$$

\begin{theorem} \label{thm:spinlock-termination}
All thread-fair behaviors of the following client program 
under $\G^\fair_{\set{\SC, \TSO, \RA, \SCOH}}$ are finite:
\begin{equation*}
\tag{SpinLock-Client}\label{ex:spinlock-client}
\inarrIV{
  lock() \\
  unlock()
}{
  lock() \\
  unlock()
}{
  \dots
}{
  lock() \\
  unlock()
}
\end{equation*}
\end{theorem}
\begin{proof}
Assume for the sake of contradiction that the program has an infinite
thread-fair behavior $\beh$, which is induced by a fair execution graph $G$.
By inspection, since $G$ is infinite, $\beh$ must contain an infinite spinloop.
The number of write events to the location $l$ in $G$ is finite since each thread makes at most two writes to $l$.
Fix the $\lMO$-maximal one among them and denote it $w$.
Due to thread fairness of $\beh$, 
the value written by $w$ has to be $0$.
(Otherwise, it could have been only the value $1$ produced by the $\kw{CAS}$ instruction,
which is followed by a store writing $0$, and 
the write event produced by the store would have been $\lMO$-following for $w$
by $\set{\SC, \TSO, \RA, \SCOH}$-consistency of $G$.)
By \cref{thm:spinloop termination basic},
there is a spinloop iteration that reads from $w$,
which is a contradiction, since reading $0$ from location $l$ exits the loop.
\end{proof}

%% file: ticketlock.tex
\subsection{Ticket Lock}
\label{sec:tickelock}

Consider the following ticket lock implementation:
$$\inarrT{
  \assignInst{\kw{int} \; serving := 0, \; ticket}{0} \\
  \funcDecl{\kw{void}}{lock}{}{
    \assignInst{\kw{int}\; s := 0, \; r}{\faddInstn(ticket, 1)} \\
    \repeatInst{\readInst{s}{serving}}{(s = r)}
  }
  \funcDeclsl{\kw{void}}{unlock}{}{
    \writeInst{serving}{serving+1}
  }
}$$

\begin{theorem}
  \label{thm:tickelock}
In every thread-fair behavior
of the following program under $\G^\fair_{\set{\SC, \TSO, \RA, \SCOH}}$,
$r_1 \til r_N$ all grow unboundedly:
\begin{equation*}
\label{ex:tickelock-client}
\inarrIV{
  L_1: lock() \\
  r_1 := r_1 + 1 \\
  unlock() \\
  \gotoInst{L_1}
}{
  L_2: lock() \\
  r_2 := r_2 + 1 \\
  unlock() \\
  \gotoInst{L_2}
}{
  \dots
}{
  L_N: lock() \\
  r_N := r_N + 1 \\
  unlock() \\
  \gotoInst{L_N}
}
\end{equation*}
\end{theorem}
\begin{proof}
For any thread-fair behavior $\beh$ of this program
and a fair execution graph $G$ inducing $\beh$,
it can be shown that each call to $lock$ reads a unique value from $\mathit{ticket}$,
and that whenever a certain $lock$ call reads ticket value $v$ (and the spinloop exits),
the corresponding $unlock$ writes to $\mathit{serving}$ value $v+1$.
Moreover, the values written to $\mathit{ticket}$ and to $\mathit{serving}$ are strictly increasing along $G.\lMO$.
(These are standard safety properties, so we elide details of their proofs.)

By means of contradiction, now 
assume that there is a fair execution graph $G$ inducing $\beh$ 
where $r_i$ for some $1 \leq i \leq N$ is incremented only a finite number of times.

Due to thread-fairness of $\beh$, the only way this can happen is if thread $i$ has an infinite spinloop.
There may well be multiple threads with infinite spinloops,
so among those threads let us consider the thread $\tid$ that reads the smallest value for $\mathit{ticket}$, say $k$,
just before going into the infinite spinloop.
So, for all $0 \leq j < k$, some $lock$ has incremented $\mathit{ticket}$ to value $j$ and subsequently $\mathit{serving}$ to value $j+1$.
In particular, the $\lMO$-maximal among those sets $\mathit{serving}$ to value $k$.
Note that there cannot be any writes to $\mathit{serving}$ with larger values
because they all require $\mathit{serving}$ to first be set to $k+1$ (which does not happen since $\tid$ is stuck in a spinloop).

Because of thread-fairness and \cref{thm:spinloop termination basic},
the infinite spinloop must have an iteration that reads from the $\lMO$-maximal write to $\mathit{serving}$,
\ie reading value $k$.
This is a contradiction, because reading $k$ exits the loop.
\end{proof}

%% file: hmcslock.tex
\subsection{MCS lock}
\label{sec:mcslock}

\newcommand{\codecomment}[1]{\textcolor{teal}{//\; \textit{#1}}}
\newcommand{\Kmark}{\textcolor{blue!80!black}{\textbf{A}}\xspace}
\newcommand{\Lmark}{\textcolor{red!80!black}{\textbf{B}}\xspace}

As a third example, we study the MCS lock \cite{mcs},
which is the basis of the \texttt{qspinlock} currently used in the Linux kernel
and the highly scalable NUMA-aware HMCS lock \cite{hmcs}.
For the latter, \citet{hmcs-verification} observe that ``the fences necessary
for the HMCS lock on systems with processors that use weak ordering'' presented
in the original HMCS paper~\cite[p.~218]{hmcs} result in non-terminating behaviors
under RC11, which do in fact occur in practice when running the HMCS lock on a
Kunpeng 920 Arm server.
Non-termination is due to a missing release fence (or store-release) in the MCS lock used in that algorithm.
For simplicity, we therefore limit our discussion to the MCS lock, whose code follows.
\begin{equation*}
\inarrTwo{
     \assignInst{\kw{QNode} \; \mln{Lock}}{\kw{null}} \\[5pt]
         \funcDeclml{\kw{void}}{lock}{\kw{QNode} \;n}{
           \writeInst{n.\mln{locked}}{1} \\
           \writeInst{n.\mln{next}}{\kw{null}}\\
           \codecomment{$\kw{fence}^\kw{rel}$ missing in HMCS paper}\\
           \writeInst{\kw{QNode}\; \mln{pred}}{\kw{SWAP}^\kw{acqrel}(\mln{Lock}, n)} \\
           \itneml{\mln{pred} \not= \kw{null}}{
		\writeInst{\mln{pred}.\mln{next}}{n} \\
		\while{ n.\mln{locked} = 1}{}	\\
		\kw{fence}^\kw{acq}
		}
         }
         }{
     \funcDeclml{\kw{void}}{unlock}{\kw{QNode}\;n}{
       \kw{fence}^\kw{rel} \\
       \writeInst{\kw{QNode}\; \mln{succ}}{n.\mln{next}} \\
       \kw{fence}^\kw{acq} \quad\codecomment{can be elided on ARM} \\
       \itneml{\mln{succ} = \kw{null}}
       { \iteml{\kw{CAS}^\kw{acqrel}(\mln{Lock}, n, \kw{null})}
       		{\kw{return}}
       		{\repeatml{\writeInst{\mln{succ}}{n.\mln{next}}}{\mln{succ} \not= \kw{null}}}
    	}
       \\
       \writeInst{\mln{succ}.\mln{locked}}{0}
     }
   }
\end{equation*}
The MCS lock uses a FIFO queue to ensure fairness.
Therefore, the \mln{lock} and \mln{unlock} functions take a \mln{QNode} argument to identify the calling thread.
A thread $T$ can enter the critical section (after calling the \mln{lock} function) either if the queue is empty
or after its predecessor in the queue lowers the \mln{locked} bit in $T$'s \mln{QNode}.
To release the lock, a thread $T$ lowers the \mln{locked} bit of the next thread in the queue,
or if no such thread exists, empties the queue.

Consider now the following client program, in which two threads enter the critical section once.
\begin{equation*}
\tag{MCS-Client}\label{ex:MCS-client}
\inarrII{
  \assignInst{a}{\kw{new QNode}()} \\
  \mln{lock}(\mln{Lock},a) \\
  \mln{unlock}(\mln{Lock},a)
}{
  \assignInst{b}{\kw{new QNode}()} \\
  \mln{lock}(\mln{Lock},b) \\
  \mln{unlock}(\mln{Lock},b)
}
\end{equation*}
Suppose we want to show that this program terminates
and, in particular, that the \kw{while} loops in $\mathit{lock}$ terminate if ever reached.
Due to symmetry, we only consider the loop for $n=a$.
By \cref{thm:spinloop termination basic}, it suffices to consider the iteration in which the loop reads from the $\lMO$-maximal store.
We can now construct all candidate $\lMO$s and attempt to show for each one that
either the $\lMO$-maximal store allows the loop to terminate or any graph with that $\lMO$ is not RC11-consistent.

It is easy to show that in every execution of this program in which that loop is reached,
there are exactly two non-initial stores to $a.\mln{locked}$,
generated by the calls $\mathit{lock}(\mln{Lock},a)$ and $\mathit{unlock}(\mln{Lock},b)$, respectively.
For brevity's sake, we call these stores \Kmark and \Lmark respectively.
Since \Lmark writes $a.\mln{locked} = 0$, reading from it allows the loop to terminate.
Consequently, the loop may only diverge in execution graphs in which \Kmark is $\lMO$-maximal.
Such a graph is shown below.
\begin{center}\small
\begin{tikzpicture}[yscale=1,xscale=1]
  \node[label=above:{\Kmark}] (11) at (0,0) {$\evlab{\lW}{}{a.\mln{locked}}{1}$ };
  \node[right=0.3 of 11] (12) {$\evlab{\lW}{}{a.\mln{next}}{\kw{null}}$ };
  \node[right=0.3 of 12] (13) {$\ulab{}{\mln{Lock}}{b}{a}$ };
  \node[right=0.3 of 13] (14) {$\evlab{\lW}{}{b.\mln{next}}{a}$ };
  \node[right=0.3 of 14] (15) {$\evlab{\lR}{}{a.\mln{locked}}{1}$ };
  \node[right=0.3 of 15] (16) {\ldots};

  \node[below=1.0 of 11] (21)  {$\evlab{\lW}{}{b.\mln{locked}}{1}$};
  \node[right=0.3 of 21] (22) {$\evlab{\lW}{}{b.\mln{next}}{\kw{null}}$ };
  \node[right=0.3 of 22] (23) {$\ulab{}{\mln{Lock}}{\kw{null}}{b}$ };
  \node[right=0.3 of 23] (25) {\phantom{ABC}};
  \node[anchor=south] at (25.south) {$\kw{F}^\kw{rel}$};
  \node[right=0.3 of 25] (26) {$\evlab{\lR}{}{b.\mln{next}}{a}$ };
  \node[right=0.3 of 26] (27) {\phantom{ABC}};
  \node[anchor=south] at (27.south) {$\kw{F}^\kw{acq}$};
  \node[right=0.3 of 27,label=above:{\Lmark}]  (28) {$\evlab{\lW}{}{a.\mln{locked}}{0}$ };

  \draw[po] (11) edge (12);
  \draw[po] (12) edge (13);
  \draw[po] (13) edge (14);
  \draw[po] (14) edge (15);
  \draw[po] (15) edge (16);

  \draw[po] (21) edge (22);
  \draw[po] (22) edge (23);
  \draw[po] (23) edge (25);
  \draw[po] (25) edge (26);
  \draw[po] (26) edge (27);
  \draw[po] (27) edge (28);

  \draw[mo] (22) edge (14);
  \draw[mo,transform canvas={xshift=0.5ex}] (23) edge node[pos=0.8,right] {\smaller$\lMO$} (13);
  \draw[rf,transform canvas={xshift=-0.5ex}] (23) edge node[pos=0.8,left] {\smaller$\lRF$} (13);
  \draw[mo,in=340,out=160] (28) edge node[below] {\smaller$\lMO$} (11);

  \draw[rf] (14) edge node[pos=0.4,below] {\smaller$\lRF$} (26);
  \draw[rf,bend left=8] (11) edge node[above] {\smaller$\lRF$} (15);

  \begin{pgfonlayer}{background}
    \draw[ultra thick,draw=yellow,in=340,out=160] (28) edge (11);
    \draw[ultra thick,draw=yellow] (14) edge (26);
    \draw[ultra thick,draw=yellow] (26) edge (27);
    \draw[ultra thick,draw=yellow] (27) edge (28);
    \draw[ultra thick,draw=yellow] (11) edge (12);
    \draw[ultra thick,draw=yellow] (12) edge (13);
    \draw[ultra thick,draw=yellow] (13) edge (14);
  \end{pgfonlayer}
\end{tikzpicture}
\end{center}
The graph is in fact RC11-consistent, and therefore the client program does not always terminate.
Once, however, we add back the commented-out $\kw{fence}^\kw{rel}$ in the \mln{lock} function, then 
the highlighted {\setlength{\fboxsep}{0pt}\colorbox{yellow!20!white}{$\lPO \seq \lRF \seq \lPO \seq \lMO$}
cycle in the execution graph above is forbidden.
Similarly, the release fence also rules out all other graphs in which \Kmark is the $\lMO$-maximal store,
and we can thus prove the following theorem.
(Our Coq proof generalizes this theorem to an arbitrary finite number of threads.)
\begin{theorem} \label{thm:hmcs-termination}
If the $\kw{fence}^\kw{rel}$ in the MCS lock is uncommented, \ref{ex:MCS-client}'s thread-fair behaviors
under $G^\fair_{\set{\SC, \TSO, \RA, \RC}}$ are all finite.
\end{theorem}

%% file: discussion.tex
\section{Related Work and Discussion}
\label{sec:discussion}

We have investigated fairness in $(\lPO\cup\lRF)$-acyclic weak memory models,
both operationally and declaratively, established four equivalence results,
and showed how the declarative formulations can be used for reasoning about program termination.

Several papers, \eg \cite{cerone,Bouajjani14,gotsman_et_al:Sequencing}, have studied
declarative formulations of transactional consistency with prefix-finiteness
constraints to ensure that a transaction is never preceded by an infinite set
of other transactions.
In particular,
\citet{gotsman_et_al:Sequencing} established a connection
between declarative presentations that include fairness constraints
and operational presentations for models in their ``Global Operation
Sequencing'' framework. The TSO model can be expressed in this framework.
Their declarative specifications require prefix-finiteness of the global visibility order,
while we derive this property from prefix finiteness of more local relations ($\lMO$ and $\lFR$).
Thus, our formulation is easily applicable for model checking based on partial order reduction
in the style of \citet{rcmc,genmc}.
To the best of our knowledge, this is the first work to make a connection
between liveness in declarative models formulated in the widely used framework of \citet{herding-cats}
and in operational models.

Termination of the MCS lock was previously studied by \cite{vsync};
however, due to the lack of a formal definition of fairness, \citet{vsync}
assumed a highly technical consequence of fairness in their proofs.
Our unified definition of fairness and \cref{thm:spinloop termination basic}
bridge the gap left in their arguments and allow us to obtain the first
complete formal termination proof for the MCS lock.

We note that our approach for establishing termination of spinloops is not only
useful for manually proving deadlock-freedom and related progress properties as
shown in \cref{sec:lockexamples}, but can also be used to automatically establish
termination of programs whose only potentially unbounded behavior is due to
spinloops.  One can use \cref{thm:spinloop termination basic} to reason about
the termination of such programs by examining only a finite number of finite
execution graphs.  This approach has actually been implemented in the
\textsc{GenMC} model checker \cite{genmc-tool}, and thus termination of
the example programs in the paper (for a bounded number of threads) can
also be shown automatically.

We outline two directions for future work,
which concern extending our results to more complex models.

\paragraph{Fairness under non-$(\lPO\cup\lRF)$-acyclic models}

Some low-level hardware memory models, such as Arm \cite{arm8-model} and
POWER~\cite{herding-cats}, and hardware-inspired memory models,
such as LKMM~\cite{lkmm} and IMM~\cite{imm},
record syntactic dependencies between instructions
so as to allow certain executions with cycles in $\lPO \cup \lRF$.
In these models, prefix-finiteness of $\lMO$ and $\lFR$ alone does not suffice
for prefix-finiteness of the appropriate ``happens-before'' relation.
For instance, under Arm (version 8)~\cite{arm8-model},
assuming prefix-finiteness of $\lMO$ and $\lFR$ does not forbid
the out-of-thin-air read of the value $5$ in the following example
(with an unbounded address domain):
\begin{gather*}
\inarrII{
L_1\colon           y_i := x_i \comment{5}  \\
\phantom{L_1\colon} i := i + 1 \\
\phantom{L_1\colon} \textbf{goto} \ L_1
}
{
L_2\colon           x_j := y_{j+1} \comment{5} \\
\phantom{L_2\colon} j := j + 1  \\
\phantom{L_2\colon} \textbf{goto} \ L_2
}
\quad
\inarr{\begin{tikzpicture}[yscale=0.8,xscale=1.4]
  \node[inner sep=2pt] (R11)  at (0, 1) {$\evlab{\lR}{}{x_0}{5}$};
  \node[inner sep=2pt] (W11)  at (1, 1) {$\evlab{\lW}{}{y_0}{5}$};
  \node[inner sep=2pt] (R12)  at (2, 1) {$\evlab{\lR}{}{x_1}{5}$};
  \node[inner sep=2pt] (W12)  at (3, 1) {$\evlab{\lW}{}{y_1}{5}$};
  \node[inner sep=2pt] (R13)  at (4, 1) {$\evlab{\lR}{}{x_2}{5}$};
  \node[inner sep=2pt] (W13)  at (5, 1) {$\evlab{\lW}{}{y_2}{5}$};
  \node[inner sep=2pt] (R21)  at (0, -0.2) {$\evlab{\lR}{}{y_1}{5}$};
  \node[inner sep=2pt] (W21)  at (1, -0.2) {$\evlab{\lW}{}{x_0}{5}$};
  \node[inner sep=2pt] (R22)  at (2, -0.2) {$\evlab{\lR}{}{y_2}{5}$};
  \node[inner sep=2pt] (W22)  at (3, -0.2) {$\evlab{\lW}{}{x_1}{5}$};
  \node[inner sep=2pt] (R23)  at (4, -0.2) {$\evlab{\lR}{}{y_3}{5}$};
  \node[inner sep=2pt] (W23)  at (5, -0.2) {$\evlab{\lW}{}{x_2}{5}$};
  \node[inner sep=2pt] (inf) at ($(W13)!0.5!(W23)+(0.5, 0)$) {\huge\ldots};
  \dpo{R11}{W11}; \dpo{W11}{R12}; \dpo{R12}{W12}; \dpo{W12}{R13}; \dpo{R13}{W13};
  \dpo{R21}{W21}; \dpo{W21}{R22}; \dpo{R22}{W22}; \dpo{W22}{R23}; \dpo{R23}{W23};
  \draw[rf] (W12) to [out=215,in=35] (R21);
  \draw[rf] (W13) to [out=215,in=35] (R22);
  \draw[rf] (W21) -- (R11);
  \draw[rf] (W22) -- (R12);
\end{tikzpicture}}
\end{gather*}
We conjecture that the appropriate liveness condition for Arm is to require
prefix-finiteness of the ``ordered-before'' ($\mathtt{ob}$) relation.
We leave adapting the operational Arm model to ensure fairness and establishing
correspondence between the two models for future work.

Similarly, there are a number of more advanced memory models for programming
languages that aim to admit write-after-read reorderings (and thus have to
allow $(\lPO \cup \lRF)$ cycles) such as JMM~\cite{Manson-al:POPL05},
Promising~\cite{Kang-al:POPL17}, Pomsets with Preconditions~\cite{pomsets-pre},
and Weakestmo~\cite{weakestmo}.  Integrating liveness requirements in such
memory models is left for future work.

\paragraph{Weak RMWs}
Besides ordinary (``strong'') CAS instructions,
C11 supports ``weak'' CASes,\footnote{See \url{https://en.cppreference.com/w/cpp/atomic/atomic/compare_exchange} [accessed November-2020].}
which may fail spuriously,
\ie even when they read the expected value,
since on some architectures---namely, POWER and Arm---weak CASes are
more efficient than strong ones.
A strong CAS can be implemented by repeatedly performing a weak CAS
in a loop as long as it fails spuriously.
Termination of such loops depends upon the weak CASes not always failing spuriously, which constitutes an additional fairness requirement.
Since this requirement is orthogonal to the notion of memory fairness introduced in this paper,
we leave it for future work.

%% file: equiv_proof.tex
\section{Equivalence proofs of operational and declarative representations}
\label{sec:equiv_proof}
In this section, we present the proof of \cref{thm:equiv_proof}.
We split it to eight lemmas: forward and backward directions for
$\SC$ (\cref{thm:sc-equiv-forward,thm:sc-equiv-backward}),
$\RA$ (\cref{thm:ra-equiv-forward,thm:ra-equiv-backward}),
$\SCOH$ (\cref{thm:scoh-equiv-forward,thm:scoh-equiv-backward}),
and $\TSO$ (\cref{thm:tso-equiv-forward,thm:tso-equiv-backward}) models.

\input{SC.tex}
\input{RA.tex}
\input{SCOH.tex}

\input{TSO.tex}

%% file: SC.tex
\begin{lemma}\label{thm:sc-equiv-forward}
$\Behmf{\M_\SC}\suq \Beh{\G_\SC^\fair}$
\end{lemma}
\begin{proof}
Let $\beh=\beh(\trace)$ where $\trace$ is a trace induced by a run $\run$ of $\M_\SC$.

We construct an \SC-consistent fair execution graph $G$ inducing $\beh$ such that $G.\lE$ represents labels of $\trace$. %
We start by defining a partial
function $\toevents{\cdot} : \mathbb{N} \rightharpoonup \Tid \times \N \times \Lab$
which converts elements
of trace $\trace$ to execution graph events:
$$
\toevents{k} \defeq
\begin{cases}
  \ev{n}{\tid}{\lab} &
  \text{where } k \in \dom{\trace},\; \trace(k) = \tstep{\tid}{\lab}, \; \lab \in \Lab, \\
     &\quad \text{and } n \defeq | \set{i \le k \st \exists \lab' \in \Lab. \; \trace(i) = \tstep{\tid}{\lab'}} | \\
  \bot &  \text{otherwise} %
\end{cases}
$$
We define a set $E$, which consists of events constructed from trace $\trace$ via function $\toevents{\cdot}$
supplemented with initial events, and a strict partial order $\evorder$ over it:
$$
\begin{array}{c@{\;}l}
E & \defeq \Init \cup \set{\toevents{k} \st k \in \dom{\trace}} \\
\evorder & \defeq \Init \times (E \setminus \Init) \cup \set{\tup{\toevents{i}, \toevents{j}} \in E \times E \st i < j}
\end{array}
$$

We define $G$ to be a tuple $\tup{E,\lRF,\lMO}$ where:
\begin{itemize}
\item $\lRF \defeq ([\lW]\seq {\evorder}|_{\lLOC} \seq[\lR]) \setminus ({\evorder}|_{\lLOC} \seq [\lW] \seq {\evorder}|_{\lLOC})$
 relates a read and the previous $\evorder$-latest write to the same location;

\item $\lMO\defeq [\lW]\seq {\evorder}|_{\lLOC} \seq[\lW]$ relates $\evorder$-ordered writes to the same location: $\lMO $. 
\end{itemize}
Finally, we prove the following. 
\begin{itemize}
\item $G$ is an execution graph.
  \begin{itemize}
  \item Requirements on the event set hold by construction. 
  \item During a read transition the value currently stored in memory is observed and this value is written by the last write to the same location. Also, $\lRF^{-1}$ is functional and is defined for all reads because for each read there is the unique previous write to the same location. 
  \item For each $l \in \Loc$ the relation $\lMO_l$ is a strict total order because ${\evorder}|_{\lLOC}$ is a strict total order. 
  \end{itemize}

\item $\lFR \suq {\evorder}$.
Suppose that there are $r\in \lR$, $w\in \lW$ such that $\lFR(r, w) \land w \evorder r$.
By the definition of $\lFR$, there is a $w'\in \lW$ such that $\lRF(w', r) \land \lMO(w', w)$.
By the definition of $G.\lMO$ it follows that $w' \evorder w$.
But then it follows that $r$ reads from non-latest write which contradicts the definition of $G.\lRF$.

\item $G$ is fair holds because both $\lMO$ and $\lFR$ are subsets of $\evorder$, which is prefix-finite.

\item $G \in \G_\SC$ holds because $G.\lHBSC$ is a subset of $\evorder$, which is a strict partial order.

\item $\beta(G) = \beta$. By the definition of $\toevents{}$ the sequence of labels of events belonging to a thread $\tid$ is exactly a restriction of $\trace$ to $\tid$. 

\end{itemize}
\end{proof}

\begin{lemma}\label{thm:sc-equiv-backward}
$\Behmf{\M_\SC}\supseteq\Beh{\G_\SC^\fair}$
\end{lemma}
\begin{proof}
Let $\beh\in \Beh{\G_\SC^\fair}$.
Let $G$ be a fair SC-consistent execution graph with $G.\lE$ that induces the behavior $\beh$.
We'll show that there exists a memory-fair trace $\trace$ of $\M_\SC$ that induces behavior $\beh$.

Since $G.\lE$ represents $\trace$, $G.\lE \setminus \Init = \bigcup_{\tid \in Tid, k \in \dom{\trace_\tid}} \set{(k, \tid : \trace_\tid(k))}$
where $\trace_\tid$ is a restriction of $\trace$ to thread $\tid$.
That is, for each $\tid \in Tid$ there exists $\run_\tid$.

By \cref{cor:hb-pf,prop:respect-enum}, there exists an enumeration $\set{e_i}_{i\in\N}$ of $G.\lE\setminus \Init$ that respects $G.\lHBSC$.

Now, we define the run of the memory system.
We build a function $\run \defeq \N \to \M.\lQ \times (\Tid \times \Lab) \times \M.\lQ$ recursively. Consider an arbitrary $i \in \N$.
\begin{itemize}
\item Let $(k, \tid : \trace_\tid(k)) = e_i$.
\item Then let $(src, lbl, tgt) = \run_\tid(k)$. Note that $lbl = \trace_\tid(k)$ by construction. 
\item If $i = 0$ let $M_\mathit{prev} = \M.\linit$, else let $M_\mathit{prev} = \lTGT(\run (i - 1))$. 
\item If $\lTYP(lbl) = \lR\ \lor\ \lTYP(lbl) = \lF$, let $M_\mathit{next}= M_\mathit{prev}$. \\
  Otherwise let $M_\mathit{next}= M_\mathit{prev}[\lLOC(lbl) \to \lVALW(lbl)]$. 
\item Finally, we define $\run (i) = (M_\mathit{prev}, lbl, M_{next})$. 
\end{itemize}

Suppose that on each step the \SC memory subsystem allows the transition between the $M_\mathit{prev}$ and $M_\mathit{next}$.
Then $\run$ is a sequence of transitions.
Also, $\lSRC(\run(0))= \M_\SC.\linit$ and adjacent states agree.
That is, $\run$ is a run of $\M_\SC$.

Let $\trace = \lTLAB \circ \run$ be a trace of $\M_\SC$.
For each $\tid \in \Tid$ the restriction $\trace\rst{\tid}$ is equal to $\trace_\tid$ by construction, so $\beh(\trace) = \beh(G) = \beh$.
Also, since $\SC$ doesn't have silent memory transitions, $\trace$ is memory-fair.

It remains to show that on each step the \SC memory subsystem allows the transition between the $M_\mathit{prev}$ and $M_\mathit{next}$.
To do that we need to prove that during transitions with $\lR$ and $\lU$ labels a thread observes the value currently stored in memory at some address $\loc$.
Note that by the moment of that transition's execution the value stored in memory is written by the current $\lMO\rst{\loc}$-latest write $w$ because the events enumeration order respects $\lHBSC \supseteq \lMO$.
So it would be sufficient to prove that $\lRF(w, r)$ where $r$ is the event that corresponds to the aforementioned transition.

Suppose by contradiction that there is $w'$ such that $w' \neq w \land \lRF(w', r) \land \lMO(w', w)$.
Then $\lFR(r, w)$ and, since $\lFR \suq \lHBSC$, the $r$ transition should have been executed before $w$ transition which contradicts the choice of $w$. 
\end{proof}

%% file: RA.tex
\begin{lemma}\label{thm:ra-equiv-forward}
$\Behmf{\M_\RA}\suq \Beh{\G_\RA^\fair}$.
\end{lemma}
\begin{proof}
  \input{RAforward}
\end{proof}

\begin{lemma}\label{thm:ra-equiv-backward}
$\Behmf{\M_\RA}\supseteq\Beh{\G_\RA^\fair}$.
\end{lemma}
\begin{proof}
  \input{RAbackward}
\end{proof}

%% file: RAforward.tex
Let $\beh$ be in $\Behmf{\M_\RA}$ and 
$\trace$ and $\run$ be a finite or infinite and fair trace and run $\M_\RA$ inducing $\beh$ correspondigly.

Then, we define a partial function $\toevents{\cdot} : \mathbb{N} \rightharpoonup \Tid \times \N  \times \Lab$,
a set $E$, and a partial order $\evorder$ over it as in \cref{thm:sc-equiv-forward} for \SC.

We define a function $\smap : E \to \lQ_\RA$: %
$$
\smap(e) \defeq
\begin{cases}
  q & \text{where } e = \toevents{k} \text{ for some } k \in \N \text{ and } \tup{\_, q} = \run(k) \\
  \linit_\RA & \text{where }  e \in \Init 
\end{cases}
$$
its projection $\vmap : E \to \View$ which maps non-initial events to their thread views:
$$
\vmap(e) \defeq 
\begin{cases}
\Tview(\lTID(e)) & \text{ where } e \in E\setminus\Init \text{ and } \tup{\_, \Tview, \_} = \smap(e) \\
\viewinit & \text{ where } e \in \Init
\end{cases}
$$
and its projection $\tmap : E \rightharpoonup \Time$ which maps events to timestamps of related messages:
$$
\tmap(e) \defeq
\begin{cases}
  \vmap(e)(\lLOC(e))  & \text{where } e \in E \text{ and } \lLOC(e) \text{ is defined} \\
  \bot & \text{otherwise} %
\end{cases}
$$

We denote $\set{\tup{e,e'} \st f(e) = f(e') }$ by $=_{f}$
and $\set{\tup{e,e'} \st f(e) < f(e') }$ by $<_{f}$ for $f \in \set{\tmap, \vmap}$.
Obviously, both $<_{\tmap}$ and $<_{\vmap}$ are strict partial orders.

We define $G$ to be the execution graph $\tup{E,\lRF,\lMO}$ where:
\begin{itemize}
\item  $\lRF \defeq [\lW]\seq =_{\tmap} \rst{\lLOC} \seq[\lR]$;
\item  $\lMO \defeq [\lW]\seq <_{\tmap} \rst{\lLOC} \seq[\lW]$;
\end{itemize}
Consequently, $G.\lFR$ is equal to $[\lR]\seq <_{\tmap} \rst{\lLOC} \seq[\lW]$.

Finally, we prove the following. 
\begin{itemize}
\item $G$ is an execution graph.
  \begin{itemize}
  \item Requirements on the events set hold by construction.
  \item During a read transition, a message is read, and this message is written by some write to the same location.
    Also, $\lRF^{-1}$ is functional and is defined for all reads because for each read there is the unique message to the same location with
    the same timestamp.
  \item For each $l \in \Loc$ the relation $\lMO_l$ is a strict total order because $<_{\tmap}\rst{\lLOC}$ is a strict total order by properties of $\tmap$.
  \end{itemize}
  
\item $G$ is fair.
  First, we show that $G.\lMO$ is prefix-finite.
  Fix $w\in\lE\cap\lW$. We know that $\tmap(w)$ is defined and belongs to $\N$.
  As a consequence, $|\dom{G.\lMO; [w]}| \le \tmap(w)$ by definition of $G.\lMO$.
  
  Now, we show that $G.\lFR$ is prefix-finite.
  Suppose that it is not the case.
  Then, there exists $w \in \lE\cap\lW$ and an infinite set of unique read events $\set{r_i}_{i \in \N} \suq \lE\cap\lR$ \emph{s.t.} $\forall i \in \N. \; \tup{r_i, w} \in G.\lFR$.
  Since a number of threads in the program is bounded, we may assume that $\lTID(r_i) = \tid$ for all $i$ and some $\tid$.
  Also, since $\set{r_i}_{i \in \N} \suq \lE\setminus\Init$, we know that there exists $\set{a_i}_{i \in \N} \suq \N$ \emph{s.t.} $\forall i \in \N. \; r_i = \toevents{a_i}$.
  We can deduce that $\forall i \in \N. \; G.\lRF^{-1}(r_i) <_{\tmap} w$.

  For each $i$,  $\run_i.\lQ.\Tview(\tid)(\loc)$ is less than $\tmap(w)$.
  $\tmap(w)$ is a timestamp of some message $\msga = \msg{\lLOC(w)}{\lVAL(w)}{\tmap(w)}{\_} \in \run_i.\lQ.\m$.
  That is, $\astep{\proplab(\tid,\msga)}_{\RA}$ is continuously enabled in $\run$ and never taken.
  It contradicts fairness of $\run$.

\item $G \in \G_\RA$, that is, the relation $G.\lHBRA\rst{loc} \cup \lMO \cup \lFR$ is acyclic.
      We start by showing that $G.\lHBRA$ is acyclic. It follows from the fact that $\lPO \cup \lRF \suq {} \evorder {}$
      and $\evorder$ is a strict partial order:
      \begin{itemize}
      \item $G.\lPO \suq {}\evorder{}$ holds by construction of $E$.
      \item $\lRF \suq {}\evorder$.
        Fix an edge $\tup{e, e'}$ in $\lRF$.
        Since $e' \in \lR$, there exists $j$ \emph{s.t.} $e' = \toevents{j}$. Also, $e' \in G.\lE \setminus \Init$.
        If $e \in \Init$, then $e \evorder e'$.
        If $e \in G.\lE \setminus \Init$, then there exists $i$ \emph{s.t.} $e = \toevents{i}$.
        Also, it means that $\tup{\_, \tstep{\lTID(e)}{\lLAB(e)}, \smap(e)} = \run(i)$.
        Since $e \in \lW$, the $\run(i)$ transition is a write step, \ie on this step
        a message $\msga = \msg{\lLOC(e)}{\_}{\tmap(e)}{\_}$ is added.
        The transition $j$ reads the message, \ie $i < j$. Consequently, $e \evorder e'$.

      \end{itemize}

      Acyclicity of $G.\lHBRA$ means that a $G.\lHBRA\rst{loc} \cup \lMO \cup \lFR$ cycle has to contain at least one $\lMO \cup \lFR$ edge.
      Then, it is enough to show that $G.\lHBRA\rst{loc} \suq {} \le_{\tmap} {}$ since
      $\lMO \cup \lFR \suq {} <_{\tmap} {}$ ($\lMO$ by definition, and $\lFR$ as being equal to $[\lR]\seq <_{\tmap} \rst{\lLOC} \seq[\lW]$)
      and $<_{\tmap}$ is a strict partial order.

      To show that $G.\lHBRA\rst{loc} \suq {} \le_{\tmap} {}$, it is enough to prove that $G.\lHBRA \suq {} \le_{\vmap} {}$.
      For the latter, it is enough to show $G.\lPO \cup \lRF \suq {} \le_{\vmap} {}$ since $\le_{\vmap}$ is transitive.
      \begin{itemize}
        \item $G.\lPO \suq {} \le_{\vmap}$. Fix an edge $\tup{e, e'}$ in $G.\lPO$.
        By construction of $E$, we know that $\smap(e) \astep{}_\RA^* \smap(e')$.
        Consequently, $\vmap(e) \le \vmap(e')$ since a view of a specific thread only growths during the RA run.

        \item $\lRF \suq {} \le_{\vmap}$. Fix an edge $\tup{e, e'}$ in $\lRF$.
        If $e \in \Init$, then $\vmap(e) = \viewinit$, that is, $\vmap(e) \sqsubseteq \vmap(e')$.
        If $e \in G.\lE \setminus \Init$, then there exists $i$ \emph{s.t.} $e = \toevents{i}$.
        Also, it means that $\tup{\_, \tstep{\lTID(e)}{\lLAB(e)}, \smap(e)} = \run(i)$.
        Since $e \in \lW$, the $\run(i)$ transition is a write step, \ie on this step
        a message $\msga = \msg{\lLOC(e)}{\_}{\tmap(e)}{\vmap(e)}$ is added.

        Since $e' \in \lR$, there exists $j$ \emph{s.t.} $e' = \toevents{j}$. On the $\run(j)$ transition,
        message $\msga$ is read by thread $\lTID(e')$. That is, the thread's view is updated by message's view $\vmap(e)$.
        As a consequence, $\vmap(e) \sqsubseteq \vmap(e')$.

      \end{itemize}

\end{itemize}

%% file: RAbackward.tex
Let $\beh\in \Beh{\G_\RA^\fair}$ be induced by a fair execution graph $G$ \emph{s.t.}
$G$ is in $\G_\RA^\fair$.
Let $\set{e_i}_{i \in \N}$ be an enumeration of $G.\lE\setminus\Init$ events%
\footnote{Here we assume that $G.\lE$ is infinite to be specific. However, the similar argument works for the finite case.}
which respects $G.\lPO \cup G.\lRF$. %
Such enumeration exists since $G.\lPO \cup G.\lRF$ is acyclic.
We also define corresponding sequences of event sets $\set{E_i = \Init \cup \bigcup_{j < i}e_{j}}_{i \in \N}$ and
(partial) execution graphs:
$$
\set{G_i = \tup{E_i, [E_i]\seq\lRF\seq[E_i], [E_i]\seq\lMO\seq[E_i]}}_{i \in \N}
$$
For the sequence of execution graphs, we will construct a sequence of RA memory subsystem states
$\set{q_i}_{i \in \N} \suq \lQ_\RA$ \emph{s.t.}, for all $i \in \N$,
$q_i$ \emph{simulates} $G_i$ %
and the following conditions hold:
\begin{itemize}
\item $q_0 = \linit_\RA$;
\item $\forall i \in \N. \; q_i \astep{\tstep{\lTID(e_{i})}{\lLAB(e_{i})}}_\RA \astep{\proplab(\cdot,\cdot)}^*_\RA q_{i+1}$.
\end{itemize}
The latter condition states that $q_i$ and $q_{i+1}$ are related via a step with the same label as $e_{i}$ and a finite number of \RA's silent transition,
\ie propagation of some messages. The propagation steps are required for constructing a fair run.

Since for each transition $q_i \astep{\tstep{\lTID(e_{i})}{\lLAB(e_{i})}}_\RA \astep{\proplab(\cdot,\cdot)}^*_\RA q_{i+1}$
the silent transitions' number is finite, we can construct a run of the RA operational machine
from  $\set{q_i}_{i \in \N}$ with the same behavior $\beh$ induced by execution graph $G$.
In the remainder of the proof, we define a number of auxiliary constructions, build $\set{q_i}_{i \in \N}$,
then show that the aforementioned conditions on the sequence hold,
and prove the related run of the RA machine is fair.

Since $G.\lMO$ is acyclic and prefix-finite, there exists a function $\tmap : G.\lE \to \Time$ satisfying the following requirements:
\begin{itemize}
  \item $e \in \Init \Rightarrow \tmap(e) = 0$;
  \item $\tup{w, w'} \in G.\lMO \Rightarrow \tmap(w) < \tmap(w')$;
  \item $\tup{w, w'} \in G.\lMO \setminus (G.\lMO \seq G.\lMO) \Rightarrow \tmap(w') = \tmap(w) + 1$;
  \item $\tup{w, r} \in G.\lRF\seq[G.\lR] \Rightarrow \tmap(w) = \tmap(r)$.
\end{itemize}
Consequently, we know that if $\tup{r, w} \in G.\lFR$, then $\tmap(r) < \tmap(w)$.

We define a set of \emph{safe points for propagation of} a write event $w$, denoted $\safepoints(w)$,
to contain events which are not $\lHBRA^?$-followed by an event $e$ to the same location as $w$ with a smaller timestamp or to event $w$ itself:
$$
\inarr{
\safepoints(w) \defeq {} \\
\quad G.\lE \setminus \dom{G.\lHBRA^? \seq [\{ e \st \lLOC(e) = \lLOC(w) \land \tmap(e) < \tmap(w) \}\cup\{w\}]} 
}
$$

We define the following partial function $\tslot : \Tid \times ((G.\lE \cap \lW)\setminus \Init) \rightharpoonup \N$:%
$$
\tslot(\tid, e_i) \defeq
\begin{cases}
\bot & \text{if } G.\lE_{\tid} \suq \emptyset \\
\min\set{j \st i \le j \land
           \forall k > j. \;
           e_k \in G.\lE_{\tid} \Rightarrow  e_k \in \safepoints(e_i)} & \text{otherwise}
\end{cases}
$$
Note that $G.\lE\setminus\Init \suq \set{e_i}_{i \in \N}$,
thus we can define function for elements of the sequence.

We use $\tslot(\tid, w)$ to point to a transition in $\set{q_i}_{i \in \N}$ which includes
propagation step $\astep{\proplab(\tid,\msga)}_\RA$ for a message $\msga$ representing $w$ in the RA memory,
\ie $\msga = \msg{\lLOC(w)}{\lVAL(w)}{\tmap(w)}{\_}$.

We need to show that $\tslot$ is defined for all non-empty threads and non-initializing write events in $G$
as a consequence of $G$'s fairness.

\begin{lemma}\label{lem:tslot-wf}
$\tslot(\tid, w) \not = \bot$ for all $w \in (G.\lE \cap \lW)\setminus \Init$ and $\tid$ \emph{s.t.} $G.\lE_\tid \not = \emptyset$.
\end{lemma}
\begin{proof}
Fix $\tid$ and $w$.
If thread $\tid$ has a finite number of events then
either they all are enumerated before $w$, or the $\lPO$-latest one among them is in $G.\lE_{\tid} \cap \safepoints(w)$.
Now, we consider the case of infinite number of events in $\tid$.
To prove the lemma, we need to show that $G.\lE_{\tid} \cap \safepoints(w)$ is not empty.

Since $w \in G.\lE \setminus \Init$, there exists $i$ \emph{s.t.} $w = e_i$.
Let's pick $e_{j_0}$ for some $j_0$ \emph{s.t.} $e_{j_0} \in G.\lE_\tid$ and $j_0 \ge i$.
If $e_{j_0} \in \safepoints(w)$, then the proof is completed.
Otherwise, there exists some $k_0$ \emph{s.t.}
$$e_{k_0} \in \codom{[e_{j_0}] \seq G.\lHBRA^?} \land \lLOC(e_{k_0}) = \lLOC(w) \land \tmap(e_{k_0}) < \tmap(w)$$
We know that $k_0$ is bigger than $j_0$ since $\tup{e_{j_0}, e_{k_0}} \in G.\lHBRA^?$ and enumeration $\set{e_i}_{i \in \N}$
respects $G.\lHBRA$.
Then, we can pick $j_1 > k_0$ \emph{s.t.} $e_{j_1} \in G.\lE_\tid$. Again, we have either found an element of $\safepoints(w)$, or
we could pick $k_1 > j_1$.
By iterating the process, we either find an element of $\safepoints(w)$, or construct an infinite sequence $\set{e_{k_n}}_{n \in \N}$ \emph{s.t.}
$$\lLOC(e_{k_n}) = \lLOC(w) \land \tmap(e_{k_n}) < \tmap(w)$$
All elements of the sequence are either read or write events since they all operate on the same location as the write event $w$.
That is, $\set{e_{k_n}}_{n \in \N}$ has an infinite subsequence containing either only read events or only write events.
In the first case, there is an infinite number of $G.\lFR$-predecessors of $w$, in the second case---$G.\lMO$-predecessors of $w$.
In both cases, it contradicts fairness of $G$.
\end{proof}

Now, we construct $\set{q_i = \tup{\m_i, \Tview_i}}_{i \in \N}$.
As we mentioned earlier, we want $q_i$ to simulate $G_i$:
(i) there should be a message in $\m_i$ for each write event in $G_i$ and
(ii) $\Tview_i(\tid)$ has to represent $\tmap$-timestamps of write events from $\dom{G.\lHBRA\seq[G_i.\lE_{\tid}]}$ for each thread $\tid$.
However, all components of $q_i$ have to account for message propagation steps (which we determine by $\tslot$) also.

We define an auxiliary function $\settoview : \powerset{G.\lE} \to \View$ which assigns views to
a set of write events from $G.\lE$:
$$
\settoview(E) \defeq \bigsqcup\set{[\lLOC(w):\tmap(w)] \st w \in \lW \cap E}
$$
Then, we define a function $\vmap : G.\lE \to \View$ which assigns views representing $\lHBRA$ paths,
a function $\vmapprop : G.\lE\setminus\Init \to \View$ which assigns views representing observed propagated messages
according to $\tslot$,
and a combination of the functions $\vmapfull : G.\lE\setminus\Init \to \View$:
$$
\begin{array}{l c l}
\vmap(e)     & \defeq & \settoview(\dom{G.\lHBRA^?\seq[e]}) \\
\vmapprop(e_i) & \defeq & \settoview \; \set{w \st \tslot(\lTID(e_i), w) < i} \\
\vmapfull(e) & \defeq & \vmap(e) \sqcup \vmapprop(e) \\
\end{array}
$$
All three functions are monotone on $G.\lHBRA$ paths:
$$\forall \tup{e,e'} \in G.\lHBRA, d \in \set{-,\proplab,\mathsf{full}}. \; \vmap_d(e) \sqsubseteq \vmap_d(e')$$
The function $\etom : G.\lE \cap \lW \to \Msg$ constructs a message from a write event:
$$
\etom(e) \defeq \msg{\lLOC(e)}{\lVAL(e)}{\tmap(e)}{\vmapfull(e)}
$$
Now, we can construct all components of $q_i$:
$$
\begin{array}{l c l}
\m_i     & \defeq & \memoryinit \cup \bigcup \set{\etom(e_j) \st j < i \land e_j \in \lW} \\
\Tview'_i & \defeq & \lambda \tid. \; \bigsqcup \set{\vmapfull(e_j) \st j < i
  \land e_j \in \lE_{\tid}} \sqcup \settoview \set{w \st \tslot(\tid, w) < i - 1} \\
\Tview_i & \defeq & \lambda \tid. \; \Tview'_i(\tid) \sqcup \settoview \set{w \st \tslot(\tid, w) < i} \\
q_i & \defeq & \tup{\m_i, \Tview_i}
\end{array}
$$
Next, we need to prove that
$$
q_i \astep{\tstep{\lTID(e_{i})}{\lLAB(e_{i})}}_\RA \tup{\m_{i+1}, \Tview'_{i+1}} \astep{\proplab(\cdot,\cdot)}^*_\RA q_{i+1}
$$
holds for all $i \in \N$.

Fix $i \in \N$. For brevity, we denote $e_{i}$ by $e$ and $\lTID(e_{i})$ by $\tid$.
First, we show that $q_i \astep{\tstep{\tid}{\lLAB(e)}}_\RA \tup{\m_{i+1}, \Tview'_{i+1}}$ holds.
Consider possibilities for $\lLAB(e)$:
\begin{itemize}

\item $\lLAB(e) = \lR(\loc, \val)$.
      In this case, $\m_{i+1} = \m_i$
      and $\Tview'_{i+1} = \Tview_i[\tid \mapsto \Tview_i(\tid) \sqcup \vmapfull(e)]$ by construction.
      Since $e \in G.\lE \cap \lR$ and enumeration $\set{e_i}_{i \in \N}$ respects $G.\lRF$, 
      there exists $k < i$ \emph{s.t.} $\tup{e_k, e} \in G.\lRF$, $\tmap(e_k) = \tmap(e)$, $\lVAL(e_k)=\val$,
      and $\etom(e_k) \in \m_i$.
      
      For the transition to hold,
      we need to show that $\Tview(\tau)(\loc) \le \tmap(e_k)$, \ie
      $$
      (\bigsqcup \set{\vmapfull(e_j) \st j < i \land e_j \in \lE_{\tid}} \sqcup \settoview \set{w \st \tslot(\tid, w) < i})(\loc) \le \tmap(e_k)
      $$
      It could be split to three statements to show:
      \begin{enumerate}
      \item $(\bigsqcup \set{\vmap(e_j) \st j < i \land e_j \in \lE_{\tid}})(\loc) \le \tmap(e_k)$. \\
      Suppose it does not hold. Then, there exists some $e_j$ and $w \in \lW$ \emph{s.t.}
      $\lLOC(w) = \loc$, $\tmap(e_k) < \tmap(w)$, $\tup{w, e_j} \in G.\lHBRA^?$, and $\tup{e_j,e} \in G.\lPO$.
      That is, there is a cycle $[w]\seq G.\lHBRA^?\seq[e_j]\seq G.\lPO\seq[e]\seq G.\lRF^{-1}\seq[e_k]\seq\lMO\seq[w]$.
      Since $[w]\seq G.\lHBRA^?\seq[e_j]\seq G.\lPO\seq[e]\suq G.\lHBRA\rst{loc}$ and $G.\lRF^{-1}\seq[e_k]\seq G.\lMO\suq G.\lFR$,
      it contradicts the \RA consistency predicate.

      \item $(\bigsqcup \set{\vmapprop(e_j) \st j < i \land e_j \in \lE_{\tid}})(\loc) \le \tmap(e_k)$. \\
      Suppose it does not hold. Then, there exists some $j$, $p$, and $w \in \lW$ \emph{s.t.} $\lLOC(w) = \lLOC$, $\tmap(e_k) < \tmap(w)$,
      $p = \tslot(\tid, w)$, $p < j$, and $\set{e_p, e_j} \suq \lE_{\tid}$. That is, $\tup{e_p, e} \in G.\lPO \suq G.\lHBRA^?$.
      The fact that $\tmap(e) = \tmap(e_k) < \tmap(w)$ contradicts definition of $\tslot$.

      \item $(\settoview \set{w \st \tslot(\tid, w) < i})(\loc) \le \tmap(e_k)$. \\
      Suppose it does not hold. Then, there exists some $w \in \lW$ \emph{s.t.} $\lLOC(w) = \lLOC$, $\tmap(e_k) < \tmap(w)$, and $\tslot(\tid, w) < i$.
      It contradicts definition of $\tslot$.
      \end{enumerate}

\item $\lLAB(e) = \lW(\loc, \val)$.
      In this case, $\m_{i+1} = \m_i \cup \set{\etom(e)}$ and
      $\Tview'_{i+1} = \Tview_i[\tid \mapsto \Tview_i(\tid) \sqcup \vmapfull(e)]$
      by construction.
      By $G.\lHBRA$-monotonicity of $\vmapfull$, $\vmapfull(e) = \Tview_{i}(\tid)[\loc \mapsto \tmap(e)]$ and,
      consequently, $\Tview_i(\tid) \sqcup \vmapfull(e) = \vmapfull(e)$.
      
      Absence of a message to $\loc$ with the same timestamp $\tmap(e)$ in $\m_i$ follows from the construction of $\m_i$ and properties of $\tmap$.
      Also, $\Tview_i(\tid)(\loc) < \tmap(e)$ holds by the same reason as in the read case.
      
\item $\lLAB(e) = \lU(\loc, \valr, \valw)$.
      This case is similar to the read and write ones.
\end{itemize}
Now, we need to construct $\tup{\m_{i+1}, \Tview'_{i+1}} \astep{\proplab(\cdot,\cdot)}^*_\RA q_{i+1}$.
For that, we take a set of pairs $$X = \set{\tup{\tau, w} \st \exists \tau, w. \; i = \tslot(\tau, w)}$$
Since $\tslot(\tau,e_k) \ge k$ for all $k\in\N$ and $\tid \in \Tid$, $|X|$ is smaller than $|\Tid| * i$,
\ie it is finite. Thus, there is a finite number of message propagation steps which lead to $q_{i+1}$ from
$\tup{\m_{i+1}, \Tview'_{i+1}}$.
Note that the number of steps may be smaller than the size of $X$,
\eg if $\set{\tup{\tid, w}, \tup{\tid, w'}} \suq X$ and $\lLOC(w) = \lLOC(w')$---in this case, it is enough to
make a propagation step only for a message to the location with the biggest timestamp.
We construct a run $\run$ of the RA model inducing $\beh$ from $\set{q_i}_{i \in \N}$.

Lastly, we need to show that $\run$ is fair.
For that, it is enough to show that there do not exist $\tid \in \Tid$, $\msga \in \Msg$, and $k \in \N$ \emph{s.t.}
for all $j > k$ there exists $q'$ \emph{s.t.} $q_j \astep{\proplab(\tid,\msga)}_\RA q'$.
Suppose there are such $\tid$, $\msga$, and $k$. It means that $\msga \in \M_j$ for all $j > k$ and there is $w \in \set{e_i}_{i \le k} \cap (G.\lE \cap \lW) \setminus \Init$.
\emph{s.t.} $\etom(w) = \msga$ and $\lTID(w) = \tid$.
By~\cref{lem:tslot-wf}, $\tslot(\tid,w)$ is defined.
That is, $\tmap(w) \le \Tview_j(\tid)(\loc)$ for $j > \tslot(\tid,w)$.
It contradicts that for all $j > k$ there exists $q'$ \emph{s.t.} $q_j \astep{\proplab(\tid,\msga)}_\RA q'$. \qedhere

%% file: SCOH.tex
\begin{lemma}\label{thm:scoh-equiv-forward}
$\Behmf{\M_\SCOH}\suq \Beh{\G_\SCOH^\fair}$.
\end{lemma}
\begin{proof}
  The proof follows the same pattern as the proof of~\cref{thm:ra-equiv-forward}.
  The only difference is in showing that the constructed (in the same way) execution graph $G$ is \SCOH-consistent, \ie $G \in \G_\SCOH$,
  instead of $G \in \G_\RA$.
  
  That is, we need to show that $G.\lHBRA$ and $G.\lPO\rst{loc} \cup \lRF \cup \lMO \cup \lFR$ are acyclic.
  $G.\lHBRA$ is acyclic for the same reason as it was in~\cref{thm:ra-equiv-forward}.

  Acyclicity of $G.\lHBRA = (G.\lPO \cup \lRF)^+$ means that a $G.\lPO\rst{loc} \cup \lRF \cup \lMO \cup \lFR$ cycle has to contain at least one $\lMO \cup \lFR$ edge.
  Then, it is enough to show that $G.\lPO\rst{loc} \cup \lRF \suq {} \le_{\tmap} {}$ since
  $\lMO \cup \lFR \suq {} <_{\tmap} {}$ ($\lMO$ by definition, and $\lFR$ as being equal to $[\lR]\seq <_{\tmap} \rst{\lLOC} \seq[\lW]$)
  and $<_{\tmap}$ is a strict partial order.
  
  $G.\lPO\rst{loc} \cup \lRF \suq {} \le_{\tmap} {}$ holds for the similar reason as in~\cref{thm:ra-equiv-forward}.
\end{proof}

\begin{lemma}\label{thm:scoh-equiv-backward}
$\Behmf{\M_\SCOH}\supseteq\Beh{\G_\SCOH^\fair}$.
\end{lemma}
\begin{proof}
  The proof follows the same pattern as the proof of~\cref{thm:scoh-equiv-backward} with
  the main difference being that instead of $G.\lHBRA^?$ in definition of $\safepoints$
  it uses $G.\lRF^?\mathbin{;}G.\lPO^?$ since the read transition of \SCOH does not update thread's view with the read message's view.
\end{proof}

%% file: TSO.tex
\begin{lemma}
\label{thm:tso-equiv-forward}
$\Behmf{\M_\TSO}\suq \Beh{\G_\TSO^\fair}$.
\end{lemma}
\begin{proof}
Let $\beh\in \Behmf{\M_\TSO}$. Then, $\beh=\beh(\trace)$ where $\trace$ is a trace induced by a memory-fair run $\run$ of $\M_\TSO$.
We'll construct a \TSO-consistent fair execution graph $G$ inducing $\beh$ \emph{s.t.} $G.\lE$  represents labels of $\trace$.

First we'll relate write and propagation transitions. Since $\trace$ is memory-fair, every buffered write eventually gets propagated. That is, there exists a bijective function $write2prop: \set{i \st \trace(i) = (\_, \lW(\_, \_))} \to \set{i \st \trace(i) = (\_, \proplab)}$ that maps the trace index of write transition to the trace index of transition that propagates it.

\newcommand{\totrans}{{\toevents{}}^{-1}}

We start by constructing the set of graph events $\lE$. It is made of $\Init$ set and trace-induced events constructed with the function $\toevents{}$ like in the \SC case. Note that $\toevents{}$ has an inverse $\totrans$ on the resulting set of non-initializing events. 

The program order on the resulting events set (as defined in \cref{thm:sc-equiv-forward}) under \TSO doesn't necessarily represent the order in which the shared memory is accessed. To establish that order, we introduce the visibility function $vis: \lE \to \mathbb{Z}$. For all $e$ such that $\lTYP(e)\in \set{\lR \cup \lU}$ let $vis(e)=\totrans(e)$, that is, the trace index of the transition corresponding to that event. On the other hand, the visibility of a non-RMW write is determined by its propagation rather than the write transition itself. That is, for all $e$ such that $\lTYP(e)=\lW \land e \notin \Init$ let $vis(e)=write2prop(\totrans(e))$. Also let $vis(e)=-1$ for all $e \in \Init$.

The relations $\lRF$ and $\lMO$ will represent the events' visibility in the graph.

\begin{itemize}
\item The modification order is determined by the order in which events become visible:

  $\forall x\ y.\ \lMO(x, y) \iff x\in \lW \land y\in \lW \land \lLOC(x)=\lLOC(y) \land vis(x) < vis(y)$.
  
\item The $\lRF$ relation formalizes the rule according to which a write to a location is observed: it is either the latest non-propagated write to that location from reading thread's buffer or, if this buffer is empty, the latest propagated write to that location.
  Formally, $\lRF(w, r)$ where both $w$ and $r$ access the location $l$ if either:
  \begin{itemize}
  \item $\lTID(r) = \lTID(w)$ and $\totrans(w)$ is the index of the latest transition of a write to $l$ such that $\totrans(w) < vis(r) < vis(w)$;
  \item the buffer of $\lTID(r)$ is empty by the moment of $vis(r)$th transition and $vis(w)$ is the index of the latest $l$-propagation transition less than $vis(r)$.
  \end{itemize}
\end{itemize}

\begin{lemma} \label{thm:tso-op-dec}
  $G = \tup{E,\lRF,\lMO}$ is a \TSO-consistent fair execution graph $G$ inducing $\beh$ \emph{s.t.} $G.\lE$  represents labels of $\trace$. 
\end{lemma}
\begin{proof}
  The proof of $G$ inducing $\beh$ and having $G.\lE$ representing labels of $\trace$ is the same as in \cref{thm:sc-equiv-forward}.
  The requirements on $G$ being an execution graph hold by the construction of $G.\lE$, $G.\lRF$ and $G.\lMO$.
  Fairness and two conditions of \TSO-consistency are proved in \cref{tso-g-fair}, \cref{thm:scpl} and \cref{thm:hbtso-acyclic} correspondingly. 
\end{proof}

To prove fairness and \TSO-consistency of the aforementioned graph we'll first state some auxillary claims. 

\begin{proposition} \label{lem:w2p-order}
  The index of write transition is less than the index of its propagation: $\forall i_w.\ i_w < write2prop(i_w)$.
\end{proposition}

\begin{proposition} \label{w-prec-prop}
  An event transition is executed at the same moment or before it becomes visible: $\forall e.\ \totrans(e) \leq vis(e)$. 
\end{proposition}
\begin{proof}
  If $e$ is not a write, it becomes visible in the same transition. If $e$ is a write event, it can be propagated only after the write transition itself. 
\end{proof}

\begin{proposition} \label{po-implies-co}
  Program order on writes and RMWs respects the propagation order: $\forall x\ y.\ ([\lW];\lPO;[\lW])(x, y) \implies vis(x) < vis(y)$. 
\end{proposition}
\begin{proof}
  Writes of the same thread are propagated in the FIFO order. Also, $vis$ function is monotone w.r.t. event index for same thread events.
\end{proof}

\begin{proposition} \label{rf-implies-prec}
  $\forall w\ r.\ \lRF(w, r) \implies \totrans(w)< vis(r)$. Moreover, $\lRFE(w, r) \implies vis(w) < vis(r)$. 
\end{proposition}
\begin{proof}
  By construction the read event can only observe a write that has been executed before.
  If the read observes the write from another thread, then it should has been already propagated. 
\end{proof}

\begin{lemma} \label{fr-implies-prec}
  $\lFR(r, w) \implies vis(r) < vis(w)$. 
\end{lemma}
\begin{proof}
  Suppose that there exists such read $r$ that $\lFR(r, w) \land vis(w) < vis(r)$. It reads from some $w'$ so that $\lMO(w', w)$. By the definition of $\lMO$ it follows that $vis(w') < vis(w)$. But $r$ must read at least from the most recent propagated write, that is, at least from $w$.  
\end{proof}

\begin{lemma} \label{tso-g-fair}
  $G$ is fair.
\end{lemma}
\begin{proof}
  $\lMO$ is prefix-finite by construction.
  Suppose, by contradiction, that there exists $w$ for which there are infinitely many reads that read before it.
  Note that there is only a finite amount of read transitions before $vis(w)$ in the trace order.
  That is, there are infinitely many of them after $vis(w)$.
  But, as \cref{fr-implies-prec} shows, there are no such reads. 
\end{proof}

\begin{proposition} \label{po-fr-irreflexive}
  Relation $\lPO ; \lFR$ is irreflexive. 
\end{proposition}
\begin{proof}
  Suppose that $\lPO(w, r)$ and $\lFR(r, w)$. Then there exists a $w'$ such that $\lRF(w', r)$ and $\lMO(w', w)$. 

  Consider the $vis(w)$. By \cref{fr-implies-prec} $vis(r) < vis(w)$. Then the write buffer is not empty by the moment of read so $r$ should have observed $w$ or more recent write from the buffer. If $w'$ is from another thread, it cannot be observed. If $w'$ is from the same thread, because of $\lMO(w', w)$ it's true that $\lPO(w', w)$ and thus $w'$ cannot be observed as well. 
\end{proof}

\begin{proposition} \label{w-fr-co}
  $[\lW];\lPO_{loc};[\lR];\lFR;[\lW] \subseteq \lMO$.
\end{proposition}
\begin{proof}
  Suppose that $w_1$ and $w_2$ are related in the way shown above. Then either $\lMO(w_1, w_2)$ or $\lMO(w_2, w_1)$. Suppose the latter is the case. Then $\lFR(r, w_1)$ by the definition of $\lFR$. That is, there exists a cycle of form $\lPO ; \lFR$. According to the \cref{po-fr-irreflexive} it's not possible. 
\end{proof}

\begin{proposition} \label{po-implies-prec}
  Program order respects precedence relation except for write-read pair: 
  	\[ \forall x\ y.\ (\lPO \setminus (\lWs \times \lRs))(x, y) \implies vis(x) < vis(y). \] 
\end{proposition}
\begin{proof}
  If the first event is not a write, it becomes visible immediately: $vis(x)=\totrans(x)$. Also $\totrans(x) < \totrans(y)$. Finally, by \cref{w-prec-prop} $\totrans(y) \leq vis(y)$. 

  Suppose that the first event is a write. Then the second event is a write or an RMW, so $vis(x) < vis(y)$ by \cref{po-implies-co}.
\end{proof}

\begin{proposition}\label{rels-respect-vis}
  $\forall r \in \set{\lRFE, \lMO, \lFR}\ x\ y.\ r(x, y) \implies vis(x) < vis(y)$. 
\end{proposition}
\begin{proof}
  By the definition of $\lMO$ together with \cref{rf-implies-prec} and \cref{fr-implies-prec}. 
\end{proof}

\begin{lemma} \label{thm:scpl}
  The relation $\lPO_{loc} \cup \lRF \cup \lMO \cup \lFR$ is acyclic. 
\end{lemma}
\begin{proof}
  Suppose, by contradiction, that there exists such cycle. 
  We represent it as $\lPO$ chains (possibly made of one event only) alternated with $\lRFE \cup \lMOE \cup \lFRE$ edges. With that, we'll show that for every $e_i$ -- the first event in the cycle's $i$th $\lPO$ chain -- $vis(e_i) < vis(e_{i+1})$. Since it's a cycle, it'll result in $vis(e_0) < vis(e_0)$. 

Consider an arbitrary $i$th $\lPO$ chain and the corresponding $e_i$ (its $\lPO$-first event). Also consider $r_i$ (the next $\lRFE \cup \lMOE \cup \lFRE$ edge) which relates $t_i$ (the $\lPO$-last event in $i$th chain) and $e_{i+1}$.

If $e_i = t_i$, then $vis(e_i)=vis(t_i)$. Also, by \cref{rels-respect-vis}, $vis(t_i) < vis(e_{i+1})$.

Suppose $e_i \neq t_i$. If $\lTYP(e_i) = \lW \land \lTYP(t_i) = \lR$, then $r_i \in \lFR$. Then by \cref{w-fr-co} $\lMO(e_i, e_{i+1})$ so $vis(e_i) < vis(e_{i+1})$. If, on the other hand, $t_i \in \lW$, then by \cref{w-fr-co} $vis(e_i) < vis(t_i)$ otherwise, if $e_i \in \lR$, \cref{po-implies-prec} applies, so again $vis(e_i) < vis(t_i)$. Finally $vis(t_i) < vis(e_{i+1})$ by \cref{rels-respect-vis}. 

In all cases it's true that $vis(e_i) < vis(e_{i+1})$. 
\end{proof}

\begin{lemma} \label{thm:hbtso-acyclic}
  The relation $\lPPO \cup \lRFE \cup \lMO \cup \lFR$ is acyclic. 
\end{lemma}
\begin{proof}
  The argument from \cref{thm:scpl} applies. By definition of $\lPPO$ we don't consider the case of $\lPO \cap (\lWs \times \lRs)$ --- a special case of $\lPO \cap (\lW \times \lR)$ which was the only one in the proof of \cref{thm:scpl} where the restriction to the same location was needed. 
\end{proof}

This concludes the proof of~\cref{thm:tso-equiv-forward}.
\end{proof}

\begin{lemma}\label{thm:tso-equiv-backward}
$\Behmf{\M_\TSO}\supseteq\Beh{\G_\TSO^\fair}$.
\end{lemma}
\begin{proof}

\begin{lemma} \label{thm:tso-dec-op}
Let $G$ be a fair \TSO-consistent execution graph of $P$ where the amount of threads is finite. 
Then there exists a fair \TSO execution trace of $P$ with the same behavior. 
\end{lemma}
\begin{proof}
  The proof structure is the same as in \cref{thm:sc-equiv-backward}: we construct an enumeration of graph events that respects an acyclic relation on it, then obtain a run of program according to that enumeration and finally prove that \TSO memory subsystem allows that run's transitions. 

  We'll consider a weaker relation as a base for the enumeration, namely, $\lHBTSO \defeq (\lPPO \cup \lRFE \cup \lMO \cup \lFR)^+$.  
  It is is prefix-finite by \cref{cor:hb-pf}.
  
  By \cref{prop:respect-enum} there exists an enumeration $\set{e_i}_{i\in\N}$ of $G.\lE \setminus \Init$ that respects $\lHBTSO$.

  According to that enumeration we build a run of a concurrent system.
  A corresponding sequence of transitions is obtained with the following algorithm. 
  We start with the initial state.
  Then we enumerate (w.r.t. $\lHBTSO$) graph events and for each event extend the current run prefix with a block of transitions.

\begin{itemize}
\item If the current event is a RMW, the block will consist only of a RMW transition. 
\item If the current event is a non-RMW read, the block will consist of write transitions for all write events that $\lPO$-precede the current event but haven't been enumerated or processed in such way yet.
  Such events may appear because $\lHBTSO$ doesn't respect the program order between non-RMW writes and reads in general.
  Additionally, the last block event will be the read transition itself. 
\item If the current event is a non-RMW write and it hasn't been already included in some read's block the way described above, the current block will include the corresponding write transition.
  Additionally, whether the write has been processed before or not, the block will include the propagation transition. 
\end{itemize}

It remains to show that the resulting labels sequence is actually the sought trace. 
\end{proof}

\begin{lemma} \label{lem:tso-alg}
  The sequence $\mu$ of transitions obtained by the algorithm above is a fair \TSO execution trace. 
\end{lemma}
\begin{proof}
  Suppose that on each step \TSO memory subsystem allows the transition obtained by the algorithm.
  Note that $\lSRC(\mu(0))= \M_\TSO.\linit$ and adjacent states agree.
  That is, $\mu$ is a run of ${\M_\TSO}$.

  Let $\trace = \lTLAB \circ \mu$ be a trace of ${\M_\TSO}$.
  For each $\tid \in \Tid$ the restriction $\trace\rst{\tid}$ is equal to $\trace_\tid$ by construction: even if the enumeration outputs events out of program order, due to the specific form of read blocks they appear in the trace according to $\lPO$. So $\beh(\trace) = \beh(G) = \beh$.
  Also, for each write transition in the run there exists a propagation transition, so silent memory transition labels $(\tid, \proplab)$ are eventually taken.
  That is, $\trace$ is memory-fair.

  It remains to show that on each step the \TSO memory subsystem allows the transition.
  Namely, we need to show that propagation transitions occur on non-empty thread buffer, RMW transitions occur on empty thread buffer and read/RMW labels agree with current buffer and memory contents.
  The first condition holds because propagation occurs only after write transition that puts an item into a buffer, and propagation is the only way to clear a buffer.
  The second condition holds because enumeration respects the program order between writes and RMW, so by the moment of such transition's execution all previous thread's writes will be already propagated.

  To prove the latter condition we'll consider the write event that a given read event reads from in $G$ and show that this write coincides with the one determined by the \TSO memory subsystem\footnote{Here we reason about events in the context of operational semantics. Formally, it only deals with event labels. But every value stored in memory or buffer is placed there due to some transition execution which in turn is determined by some graph event, since the trace under consideration is built from $G$. So we implicitly assume such mapping between memory values and graph events that produce them. }. 
  Let $r$ be the read event and $loc$ the location it reads from, $w_{tr}$ the write that has been observed during the transition and $w_{\lRF}$ the $\lRF$-predecessor of $r$.
  Suppose by contradiction that $w_{tr} \neq w_{\lRF}$. 
  \begin{itemize}
  \item If the thread buffer has non-propagated writes to the location $loc$, the $w_{tr}$ will be the latest non-propagated write to $loc$ of the same thread. 
    \begin{itemize}
    \item If $\lMO(w_{\lRF}, w_{tr})$, then $\lFR(r, w_{tr})$. Therefore, there is an $\lPO ; \lFR$ cycle prohibited by \TSO.
    \item Suppose that $\lMO(w_{tr}, w_{\lRF})$.
      \begin{itemize}
      \item If $w_{\lRF}$ belongs to the same thread, it can only appear after $r$ in the program order, since $w_{tr}$ is the latest write.
        So there is a $\lPO; \lRF$ cycle prohibited by \TSO. 
      \item If $w_{\lRF}$ belongs to the other thread, it has been already been propagated since $\lRFE$ is respected by the enumeration. But it means that $w_{tr}$ is also propagated. 
      \end{itemize}
    \end{itemize}
  \item If all previous writes to $loc$ from the same thread have been already propagated, the $w_{tr}$ will be the $\lMO$-latest (in the enumeration) write to $loc$.
    \begin{itemize}
    \item If $\lMO(w_{\lRF}, w_{tr})$, then $\lFR(r, w_{tr})$.
      Then, since the enumeration respects $\lFR$, $w_{tr}$ will be preceded by $r$ in the trace order and thus cannot be observed during the transition.
    \item Suppose that $\lMO(w_{tr}, w_{\lRF})$.
      Then $w_{\lRF}$ is not propagated yet.
      The only case when $r$ can still read from it is when $w_{\lRF}$ belongs to the same thread and is buffered. 
      But it contradicts the condition on the empty buffer.  \qedhere
    \end{itemize}
  \end{itemize} 
\end{proof}

This concludes the proof of~\cref{thm:tso-equiv-backward}.
\end{proof}

%% file: termination-checking.tex
\section{Termination Checking}\label{sec:termination}

Just like termination, non-termination can be proven by considering only the iteration of a spinloop which reads only from $\lMO$-maximal stores.
If in a finite graph some threads are blocked by spinloop iterations which read from only $\lMO$-maximal stores, and all other threads have terminated, then the graph can be extended to a fair infinite graph with infinite spinloops.
\begin{definition}
A thread $\tid$ is terminated in execution graph $G$ if in the (unique) behavior $\beh$ with $G.\lE = \Event(\beh)$, $\mu_\tid(\beh)$ is terminating.
 
An execution graph $G$ is a non-termination witness if there is a thread which is not terminated in $G$ and all threads $\tid$ which are not terminated in $G$ have finite runs $\mu_\tid(\beh)$ where $\beh$ is the (unique) behavior with $G.\lE = \Event(\beh)$, and these runs end in a spinloop iteration reading from $\lMO$-maximal stores, i.e., for all $\tid$ there is a spinloop iteration $[n_\tid,{n_\tid}']$ of $\tid$ in $\beh$ where ${n_\tid}'$ is the length of $\mu_\tid(\beh)$ and
\[ [ \set{ e \in G.\lE^\tid  \st n_\tid \le \lSN(e) \le {n_\tid}'  } ]; G.\lFR = \emptyset \]
\end{definition}

Every non-termination witness proves the existence of a fair graph with infinite spinloops.
\begin{theorem} \label{lem:non-termination}
Let $G$ be a non-termination witness with $G.\lE = \Event(\beh)$ where $\beh$ is a behavior of a deterministic program.
Then there is a fair execution graph $G'$ with $G'.\lE = \Event(\beh')$ where $\beh'$ is a behavior of the program and every thread which is not terminated in $G$ has an infinite spinloop in $\beh'$.
\end{theorem}
\begin{proof}
We first define $\beh'$. For each $\tid$ we defined $\beh'(\tid)$ by case analysis on whether $\tid$ is terminated in $G$.
If $\tid$ is terminated, then $\beh'(\tid) = \beh(\tid)$. Otherwise, since $G$ is a non-termination witness, the run $\mu_\tid(\beh)$ of $\tid$ is finite. Let $K_\tid$ be the length of this run. For steps $k < K_\tid$ inside this run, the behavior is unchanged
\[ k < K_\tid \ \implies \ \beh'(\tid,k) = \beh(\tid,k) \]
For later steps, we repeat the last spinloop iteration infinitely often. Let $[n_\tid,{n_\tid}']$ be the final spinloop iteration of $\tid$ in $\beh$ which reads from $\lMO$-maximal stores, and $L_\tid={n_\tid}'-n_\tid+1$ the length of that iteration (i.e., the number of transitions in the iteration). Then the infinite repetition of this spinloop iteration is defined by
\[ k \ge K_\tid \ \implies \ \beh'(\tid,k) = \beh(\tid,K_\tid-L_\tid + (k-K_\tid) \mod L_\tid) \]

Due to determinism, each of these spinloop iterations in the program reads exactly the same values from the same locations and hence traverses exactly the same states. This shows that $\beh'$ is a behavior of the program.

We proceed to define $G'$. The set of events is defined by $G'.\lE =  \Event(\beh')$. Modification order is unchanged ($G'.\lMO = G.\lMO$). For the reads-from relation we let each of the new events read from the $\lMO$-maximal store that the final spinloop iteration read from.
We first define a mapping $l(\tid, k)$ which maps each serial number $k$ of an event in $G'$ to the serial number of the corresponding event in $G$
\[ l(\tid,k) = \begin{cases} k & \tid \ \text{is terminated} \ \lor \ k < K_\tid \\ K_\tid-L_\tid + (k-K_\tid) \mod L_\tid & \text{o.w.} \end{cases} \]
With this definition, we have $\beh'(\tid,k) = \beh(\tid,l(\tid,k))$.
We lift this mapping to events in the obvious way:
\[ l(e) = \ev{l(\lTID(e),\lSN(e))}{\lTID(e)}{\lLAB(e)} \]
Observe that we have $l(e) \in G.\lE$ for all $e \in G'.\lE$. 
Furthermore $G.\lE \subseteq G'.\lE$; hence $l$ is a relation on $G'.\lE$. 
It relates new events to the corresponding old event. 
Its inverse $l^{-1}$ relates old events to the corresponding new events.
This allows defining the reads-from relation of $G'$ as
\[ G'.\lRF \ =  \ G.\lRF;l^{-1} \]

Since $G'.\lMO$ is $G.\lMO$, it remains prefix-finite.
Furthermore, the $\lMO$-maximal stores in $G'$ are exactly the $\lMO$-maximal stores in $G$.
Since all the new reads added read from $\lMO$-maximal stores, they have no $\lFR$-successors. 
Thus $\lFR$ remains also prefix finite.
Therefore $G'$ is fair, and the claim follows.
\end{proof}

This infinite extension of the graph does not in general preserve $\G$-consistency.
For example, consider an (absurd) declarative memory system $\G$ which only admits finite execution graphs. 
We capture declarative memory systems which allow the infinite extension by the following predicate.
\begin{definition}
A declarative memory system $\G$ is infinitely-read-extensible if adding infinitely many reads to $\lMO$-maximal elements preserves $\G$-consistency, i.e., if $G \in \G$ then so is every $G'$ with 
\[ G'.\lE= G.\lE \cup E \]
with $E \subseteq \lR$ satisfying
\[ [E];G.\lFR = \emptyset \] 
\end{definition}
Practically relevant declarative memory systems are all infinitely-read-extensible. This holds in particular for $\G_{\set{\SC, \TSO, \RA}}$.
\begin{lemma} \label{lem:infinitely read extensible}
\[ \G_{\set{\SC, \TSO, \RA}} \quad \text{are infinitely-read-extensible} \]
\end{lemma}
\begin{proof}
Omitted.
\end{proof}

\begin{theorem} \label{lem:non-termination-MM}
If $\G$ is infinitely-read-extensible and $G$ is $\G$-consistent, so is $G'$ (where $G$, $G'$ are the graphs from \cref{lem:non-termination}).
\end{theorem}
\begin{proof}
By assumption, the events in the final spinloop iterations read from $\lMO$-maximal events. The newly added events
\[ E = G'.\lE - G.\lE \]
are by definition all reading from those same events. Observe first that the newly added events are all mapped by $l$ to events in spinloop iterations
\[ [E];l = [ \set{ e \in G.\lE  \st \exists \tid .\; n_\tid \le \lSN(e) \le {n_\tid}'  } ] \]
Since $G$ is a non-termination witness by assumption, these events read from $\lMO$ maximal stores
\[ [ \set{ e \in G.\lE  \st \exists \tid .\; n_\tid \le \lSN(e) \le {n_\tid}'  } ] ;G.\lFR = \emptyset \]

We conclude by definition of $G'$ that the newly added events $E$ read from $\lMO$-maximal stores:
\begin{align*}
 [E];G'.\lFR \ &=\ [E];G'.\lRF^{-1};G'.\lMO 
 \\ &= \ [E];(G.\lRF;l^{-1})^{-1};G.\lMO 
 \\ &= \ [E];l;G.\lRF^{-1};G.\lMO 
 \\ &= \ [E];l;G.\lFR 
 \\ &= \ [ \set{ e \in G.\lE  \st \exists \tid .\; n_\tid \le \lSN(e) \le {n_\tid}'  } ] ;G.\lFR
 \\ &= \ \emptyset 
\end{align*} 
Since $\G$ is infinitely-read-extensible, the claim follows.
\end{proof}

\begin{lemma}
If $G$ from \cref{lem:non-termination} is $\G_{\set{\SC, \TSO, \RA}}$-consistent, so is $G'$.
\end{lemma}
\begin{proof} 
Follows directly from \cref{lem:non-termination-MM,lem:infinitely read extensible}.
\end{proof}

Using \cref{lem:non-termination-MM,thm:spinloop termination basic}, we can use model checking to automatically decide the termination of spinloops for any client which does not generate infinite modification orders. The model checker must simply generate all finite executions of a program in which spinloops are executed for at most one iteration. 
This enumerates all non-termination witnesses (modulo staying in spinloops longer than necessary). If any non-termination witness exists, then by \cref{lem:non-termination-MM} there is an infinite, fair extension of the graph in which the spinloop never terminates. 
If there is no non-termination witness, then all spinloops must terminate after reading from $\lMO$-maximal stores. Thus by \cref{thm:spinloop termination basic}, all spinloops terminate.

In \cite{vsync} we have implemented this scheme as an extension of GenMC \cite{genmc}, a model checker which can be parametrized with declarative memory systems satisfying certain conditions (MM1--MM4 in \cite{genmc}). 
We applied the extended model checker to several locks. Among others, the MCS lock, waiter-count based locks, and the Linux qspinlock can get stuck in a spinloop when using incorrect barrier combinations.
Many of these barrier combinations ensure safety of the locks.
Since non-termination problems under weak consistency can not be detected by existing model checkers, and these barrier combinations are otherwise correct, existing model checkers may give verification engineers an unwarranted level of confidence for programs that have spinloops.